\documentclass[11pt]{article}

\usepackage{indentfirst}
\usepackage{times}
\usepackage{amsmath}
\usepackage{amsthm}
\usepackage{color}
\usepackage{eufrak}
\usepackage{verbatim}
\usepackage[makeroom]{cancel}
\usepackage{rotating}
\usepackage{float}
\usepackage{fullpage}
\usepackage{amssymb}
\usepackage{amsthm}
\usepackage{mathrsfs}
\usepackage{epsfig}
\usepackage{graphicx}
\usepackage{indentfirst}
\usepackage{framed}
\usepackage[numbers]{natbib}
\usepackage{color}
\usepackage{hyperref}
\usepackage{graphicx}
\usepackage{float}
\usepackage{bm}
\usepackage{hyperref}
\usepackage[capitalize,noabbrev,nameinlink]{cleveref}
\usepackage{enumitem}
\usepackage{mathtools}
\usepackage{mleftright}
\usepackage{thmtools} 
\usepackage{relsize}\usepackage{scalerel}
\usepackage{tikz}
\usepackage{multicol}
\usepackage[labelfont=bf]{caption}
\usepackage{algorithm}
\usepackage[noend]{algpseudocode}
\usepackage{tablefootnote}
\usetikzlibrary{positioning,shapes,shadows,arrows,calc}
\usepackage{lineno}

\definecolor{corlinks}{RGB}{0,0,150}
\definecolor{cormenu}{RGB}{0,0,150}
\definecolor{corurl}{RGB}{0,0,150}

\hypersetup{
	colorlinks=true,
	urlcolor=corlinks,
	linkcolor=corlinks,
	menucolor=cormenu,
	citecolor=corlinks,
	pdfborder= 0 0 0
}

\definecolor{blue-violet}{rgb}{0.54, 0.17, 0.89}
\definecolor{deep-green}{rgb}{0,0.5,0}

\newtheorem{theorem}{Theorem}[section]
\newtheorem{lemma}[theorem]{Lemma}

\newtheorem{corollary}[theorem]{Corollary}
\newtheorem{proposition}[theorem]{Proposition}

\theoremstyle{definition}

\newtheorem{definition}[theorem]{Definition}

\newtheorem{problem}[theorem]{Problem}

\newcommand{\D}{\mathcal{D}}

\let\phi\varphi
\newcommand{\zo}{\{0,1\}}
\newcommand{\binset}{\{0,1\}}
\newcommand{\bin}{\binset}
\newcommand{\poly}{\mathsf{poly}}

\newcommand{\SIZE}{\mathsf{SIZE}}
\newcommand{\Q}{\mathrm{Q}}

\newcommand{\supp}{\mathsf{Support}}

\DeclareMathOperator*\Exp{{\bf E}}
\DeclareMathOperator*\Prob{{\bf Pr}}
\newcommand{\SAMP}{\mathsf{SAMP}}
 
\DeclareMathOperator{\Ex}{\mathbb{E}}

\DeclareMathOperator*\opt{opt}
\newcommand{\bool}{\mleft\{0,1\mright\}}
\newcommand{\Kt}{\mathsf{Kt}}

\newcommand{\N}{\mathbb{N}}
\newcommand{\nat}{\N}

\newcommand{\rKt}{\mathsf{rKt}}
\newcommand{\pKt}{\mathsf{pKt}}

\newcommand{\rK}{\mathsf{rK}}
\newcommand{\pK}{\mathsf{pK}}
\newcommand{\K}{\mathsf{K}}
\renewcommand{\P}{\mathsf{P}}
\newcommand{\NP}{\mathsf{NP}}

\newcommand{\BPP}{\mathsf{BPP}}
\newcommand{\Ppoly}{\mathsf{P}/\poly}
\newcommand{\AC}{\mathsf{AC}}
\newcommand{\ACC}{\mathsf{ACC}}
\newcommand{\TC}{\mathsf{TC}}

\newcommand{\BPE}{\mathsf{BPE}}

\newcommand{\Dist}{\mathsf{Dist}}
\newcommand{\Avg}{\mathsf{Avg}}
\newcommand{\Heur}{\mathsf{Heur}}

\newcommand{\qq}{\mathsf{q}}

\newcommand{\eqdef}{\triangleq}

\newcommand{\MKtP}{\mathsf{MKtP}}

\newcommand{\EXP}{\mathsf{EXP}}
\renewcommand{\P}{\mathsf{P}}

\newcommand{\pr}{\mathsf{pr}\text{-}}

\newcommand{\DistNP}{\mathsf{DistNP}}

\newcommand{\TV}{\Delta_{\mathsf{TV}}}

\newcommand{\EX}{\mathsf{EX}}
\newcommand{\KL}{\mathsf{KL}}
\newcommand{\KT}{\mathsf{KT}}
\newcommand{\MCSP}{\mathsf{MCSP}}
\newcommand{\pMCSP}{\mathsf{Partial\text{-}MCSP}}
\newcommand{\Ent}[1]{\mathsf{H}({#1})}
\newcommand{\Enc}{\mathsf{Enc}}
\newcommand{\partialMCSP}{\mathsf{Partial}\text{-}\mathsf{MCSP}}

\begin{document}
	
\title{Learning, Compression, and Meta-Complexity\\\vspace{0.4cm}}

\author{
Halley Goldberg\vspace{0.2cm}\thanks{Email:\texttt{~halley.goldberg@cs.ox.ac.uk}}\\
{\small Department of Computer Science}\\
{\small University of Oxford}\vspace{0.6cm}
\and
Mikito Nanashima\vspace{0.2cm}\thanks{Email:\texttt{~nanashima@comp.isct.ac.jp}}\\
{\small School of Computing}\\
{\small Institute of Science Tokyo}\vspace{0.6cm}
\and
Igor C. Oliveira\vspace{0.2cm}\thanks{Email:\texttt{~igor.oliveira@warwick.ac.uk}}\\
{\small Department of Computer Science}\\
{\small University of Warwick}\vspace{0.6cm}
}

\maketitle
	
\vspace{-0.6cm}

\begin{abstract}
The idea that learning is closely tied to compression has long served as both a powerful principle and, at times, an informal analogy across machine learning and learning theory. Here we focus on settings in which that intuition can be made mathematically precise, surveying formal connections among learning algorithms, data compression, and computational complexity theory, with an emphasis on recent progress and current frontiers. We describe how techniques from complexity theory can guide the design of both supervised and unsupervised learning algorithms with provable guarantees, and conversely, how learning algorithms can yield new insights and results in complexity theory. A central theme is the role of data compression as a bridge between learning and complexity, especially through an area called meta-complexity, which studies the computational difficulty of compression problems and time-bounded variants of Kolmogorov complexity.
\end{abstract}

\newpage

\tableofcontents

\normalsize

\newpage

\section{Introduction}\label{sec:introduction}

\subsection{Overview}\label{sec:overview}

We describe connections between learning, compression, and meta-complexity (see \Cref{fig:intro-three-perspectives}). We start off with a brief introduction to each of these themes, emphasising some important conceptual points along the way.

\begin{figure}[H]
\vspace{0.15cm}
\centering
\begin{tikzpicture}[
    font=\small,
    panel/.style={rounded corners=8pt, line width=0.9pt},
    paneltitle/.style={font=\bfseries\large},
    tagline/.style={font=\small, align=center, text width=4.1cm}
]

\filldraw[
    panel,
    draw=blue!45!black,
    fill=blue!2
]
(0,0) rectangle (4.9,5.5);

\node[paneltitle,text=blue!55!black] at (2.45,4.98)
{Learning};

\draw[blue!25!black,line width=0.7pt]
(0.55,4.55) -- (4.35,4.55);

\draw[gray!60,line width=0.8pt,->]
(0.85,1.75) -- (0.85,4.15);

\draw[gray!60,line width=0.8pt,->]
(0.85,1.75) -- (4.25,1.75);

\draw[blue!55,line width=0.9pt,densely dotted,line cap=round]
(1.02,2.05)
.. controls (1.85,2.42) and (2.85,3.03) ..
(4.08,3.70);

\foreach \p in {
    (1.20,2.31),
    (2.62,3.08)
}{
    \filldraw[
        fill=yellow!85!orange,
        draw=white,
        line width=0.6pt
    ]
    \p circle (2.7pt);
}

\foreach \p in {
    (1.70,2.14),
    (2.15,2.41),
    (3.08,2.92)
}{
    \filldraw[
        fill=gray!65,
        draw=white,
        line width=0.6pt
    ]
    \p circle (2.7pt);
}

\filldraw[
    fill=white,
    draw=orange!80!black,
    line width=1.2pt
]
(3.85,3.86) circle (4.1pt);

\node[
    font=\bfseries\scriptsize,
    text=orange!80!black,
    above
]
at (3.85,4.03) {?};

\node[tagline,text width=5.95cm,overlay] at (2.45,0.72)
{Generalise or extrapolate\\from examples};

\begin{scope}[xshift=0.15cm]

\filldraw[
    panel,
    draw=green!45!black,
    fill=green!4
]
(5.2,0) rectangle (10.1,5.5);

\node[paneltitle,text=green!45!black] at (7.65,4.98)
{Compression};

\draw[green!25!black,line width=0.7pt]
(5.75,4.55) -- (9.55,4.55);

\foreach \i in {0,1}{
    \foreach \j in {0,...,2}{
        \pgfmathtruncatemacro{\parity}{mod(\i+\j,2)}
        \ifnum\parity=0
            \fill[green!55!black]
            ({6.03+0.54*\i},{2.63+0.54*\j})
            rectangle ++(0.48,0.48);
        \else
            \fill[green!15]
            ({6.03+0.54*\i},{2.63+0.54*\j})
            rectangle ++(0.48,0.48);
        \fi
        \draw[green!35!black,line width=0.35pt]
        ({6.03+0.54*\i},{2.63+0.54*\j})
        rectangle ++(0.48,0.48);
    }
}

\foreach \i/\j in {2/2,5/2,3/1,4/1,5/1,4/0}{
    \fill[green!55!black]
    ({6.03+0.54*\i},{2.63+0.54*\j})
    rectangle ++(0.48,0.48);
}

\foreach \i/\j in {3/2,4/2,2/1,2/0,3/0,5/0}{
    \fill[green!15]
    ({6.03+0.54*\i},{2.63+0.54*\j})
    rectangle ++(0.48,0.48);
}

\foreach \i in {2,...,5}{
    \foreach \j in {0,...,2}{
        \draw[green!35!black,line width=0.35pt]
        ({6.03+0.54*\i},{2.63+0.54*\j})
        rectangle ++(0.48,0.48);
    }
}

\draw[
    orange!80!black,
    rounded corners=2pt,
    line width=1.2pt
]
(5.96,2.56) rectangle (7.10,4.30);

\draw[green!55!black,line width=1pt,->]
(7.65,2.49) -- (7.65,1.91);

\node[
    draw=green!45!black,
    fill=white,
    rounded corners=3pt,
    inner xsep=8pt,
    inner ysep=3pt,
    font=\ttfamily\scriptsize
]
      at (7.65,1.67) {01 alternation};

\node[tagline,text width=5.2cm] at (7.65,0.72)
{Uncover patterns to obtain compact representations};

\end{scope}

\begin{scope}[xshift=0.30cm]

\filldraw[
    panel,
    draw=orange!55!black,
    fill=orange!4
]
(10.4,0) rectangle (15.3,5.5);

\node[paneltitle,text=orange!75!black] at (12.85,4.98)
{Meta-Complexity};

\draw[orange!30!black,line width=0.7pt]
(10.95,4.55) -- (14.75,4.55);

\filldraw[
    fill=white,
    draw=gray!65,
    rounded corners=3pt,
    line width=0.8pt
]
(11.13,2.56) rectangle (13.76,4.18);

\node[
    font=\ttfamily\scriptsize,
    align=left,
    text=gray!75!black
]
at (12.36,3.36)
{3.1415\\926535\\897932};
\filldraw[
    fill=orange!8,
    fill opacity=0.92,
    draw=orange!75!black,
    line width=1.4pt
]
(13.55,3.18) circle (0.67);

\draw[
    orange!75!black,
    line width=3pt,
    line cap=round
]
(14.07,2.64) -- (14.56,2.15);

\node[
    font=\bfseries\small,
    text=orange!80!black
]
at (13.55,3.18) {?};

\fill[green!45!black]
(11.18,1.79) rectangle (12.30,2.01);

\fill[yellow!70!orange]
(12.30,1.79) rectangle (13.42,2.01);

\fill[orange!75!black]
(13.42,1.79) rectangle (14.54,2.01);

\draw[gray!70,line width=0.6pt,rounded corners=1pt]
(11.18,1.79) rectangle (14.54,2.01);

\fill[gray!80]
(13.05,2.10) -- (12.92,2.29) -- (13.18,2.29) -- cycle;

\node[tagline,text width=4.45cm] at (12.85,0.72)
{Estimate the complexity\\of a given object};

\end{scope}

\end{tikzpicture}

\caption{\small{Three perspectives on structure in data: learning, compression, and meta-complexity. We investigate correspondences among them through several lenses, including Shannon's information theory, modern variants of (time-bounded) Kolmogorov complexity, and the theory of computational pseudorandomness. Interestingly, these connections can be used both to design new learning algorithms and to prove that certain compression problems admit no efficient algorithms.}}
\label{fig:intro-three-perspectives}
\end{figure}
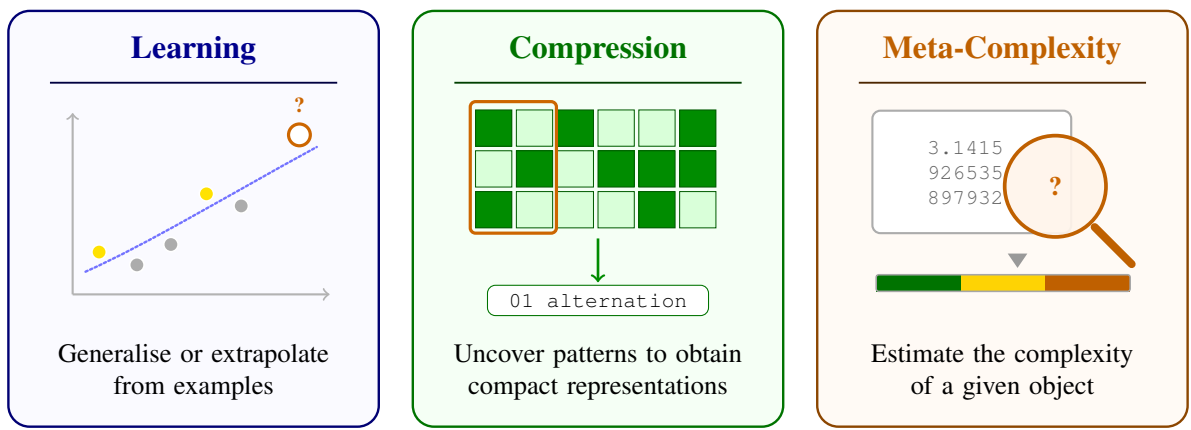

\paragraph{Learning and Compression.}
At a high level, learning and compression are two ways of asking whether observed data admit a concise explanation:
\begin{itemize}
    \item A \emph{learning algorithm} uses examples to construct a model that makes useful predictions on new data.
    \item A \emph{compression algorithm} represents the data using fewer bits, together with a procedure that recovers the original data.
\end{itemize}
The outputs are different, but both tasks succeed by identifying \emph{structure} or \emph{regularity}: a pattern that permits prediction also makes some outcomes more likely than others, and likely outcomes can be assigned shorter descriptions. Conversely, a short description is evidence that the object is not completely unstructured and can therefore be a useful starting point for prediction.

This connection takes different forms in supervised and unsupervised learning.
In \emph{supervised learning}, examples arrive with labels, and the learner seeks a hypothesis that predicts the label of a fresh example.
Here compression may mean inferring from a large labelled sample a short consistent hypothesis, as in Occam's Razor, or retaining a small subsample from which a good hypothesis can be reconstructed.
In \emph{unsupervised learning}, there are no labels.
The learner may instead produce an approximate  distribution that generates the observations, or predict the next part of a sequence from its prefix.
Here the relationship with \emph{source coding} is especially direct: if a probabilistic model assigns probability $p$ to an outcome, then the outcome can be encoded using roughly $\log(1/p)$ bits.
Thus accurate probabilistic predictions lead to short descriptions of likely data, while valid code lengths can conversely be interpreted as probability assignments.
This correspondence connects distribution learning, sequence prediction, and source coding.\\

One of the main goals of this survey is to determine precisely how far this correspondence extends.
We formulate the relevant computational tasks and present reductions showing when a learner can be transformed into a compressor, and when compression can be used to construct a learner.

\paragraph{Two Theories of Information.}
We use two complementary notions of \emph{information}. 
In \emph{Shannon's information theory}, a message is drawn from a \emph{probability distribution}, and entropy measures the optimal expected description length when coding repeated outputs of that source \citep{Shannon1948} (see \citep{cover1999elements}).
This is a statistical and distributional theory: the information content of an outcome is determined by its probability under the chosen source.
In contrast, in \emph{algorithmic information theory}, the Kolmogorov complexity of an \emph{individual string} is the length of its shortest generating program \citep{kolmogorov1965three} (see \citep{DBLP:series/txcs/LiV19}).
No source distribution needs to be specified, and the theory specifies the information content of each string. 

These two theories are related, but not interchangeable. For instance, \emph{source coding}, mentioned above, connects probability distributions with description length: a distribution with a short sampler gives short algorithmic descriptions to its typical outputs. Nevertheless, compressing individual strings in this way may be computationally expensive and may require access to the code of the sampling procedure; moreover, a string with a short program need not be well compressed by a particular sampler and the associated source coding procedure.\\

A major focus of this survey is on connections between learning and compression  under \emph{computational efficiency constraints}. Consider for instance Kolmogorov complexity. It describes the shortest possible program, but it places no bound on how long that program may take to produce its output. Given a string, estimating its Kolmogorov complexity is \emph{uncomputable} in general. Moreover, central properties of the theory do not automatically survive once efficiency is imposed. For example, \emph{symmetry of information} holds for Kolmogorov complexity up to a small additive loss, whereas its time-bounded analogues do not hold in general (see, e.g., \citep{DBLP:conf/stoc/HiraharaILNO23, DBLP:conf/stoc/HuMO26}). Classical Kolmogorov complexity, while philosophically relevant, is not by itself an adequate foundation for computational learning. To study \emph{efficient} learning through compression, we need both a theory of \emph{time-bounded}  descriptions and a theory of the \emph{computational complexity} of finding or recognizing them.

\paragraph{Meta-Complexity.}
These considerations lead naturally to \emph{meta-complexity}, the study of the computational complexity of problems that ask about the complexity of  input objects.
Canonical examples ask whether a string has a short program that generates it within a prescribed time bound, or whether the truth table of a Boolean function can be computed by a small Boolean circuit.
Such problems formalise the task of detecting whether an object is compressible or structured; see \citep{SimonsMetaComplexity2023} for further context.
The central question is no longer merely whether short descriptions exist, but whether they can be recognised or found efficiently.

The complexity of these compression problems is intertwined with some of the most fundamental unresolved questions in theoretical computer science.
In the terminology of Impagliazzo's five worlds \citep{DBLP:conf/coco/Impagliazzo95}, we do not yet know which computational world we inhabit: whether worst-case hard problems remain hard on typical inputs, whether such hardness supports secure cryptography, and how broadly efficient algorithms can exploit structure. Meta-complexity, through the study of the complexity of compression problems, turns out to provide a suitable framework in which to investigate this tension between tractability on typical inputs and cryptographic hardness. In turn, this will allow us to compare the feasibility of different learning tasks and their relation to Impagliazzo's  computational worlds.

There is, however, an important mismatch between learning and the standard formulation of compression problems in meta-complexity.
Most canonical problems in meta-complexity are \emph{decision problems}.
For example, given a string $x$, a time bound $t$, and a description length $s$, the Minimum Time-Bounded Kolmogorov Complexity problem ($\mathsf{MINKT}$) asks whether  some program of size at most $s$ produces $x$ in at most $t$ steps.
The corresponding \emph{search problem} asks for such a program.
Learning is inherently a search task: a learner must output a hypothesis, sampler, or predictor, rather than merely certify that a concise representation exists.
Thus, in principle, an efficient algorithm that merely recognises compressible data does not automatically yield a corresponding search algorithm. Interestingly, many results surveyed here bridge this gap indirectly: a decision procedure separating structured objects from random ones is sufficient to construct a learning algorithm.\\

Having introduced the main themes, we elaborate on the scope of the survey. We first revisit classical connections in which compression yields generalisation in supervised learning. After this, we compare learning tasks of different strengths, from universal sequence extrapolation and distribution learning to more traditional forms of PAC learning, and relate this hierarchy to average-case complexity, meta-complexity, and cryptography.
Next, we explain how algorithms for detecting objects of low description complexity can produce learning algorithms, both through reconstruction procedures in supervised learning and through time-bounded variants of Solomonoff induction in unsupervised learning.
This perspective also reveals novel phenomena such as generic learning speedups and, in certain settings, equivalences between efficient learning and the easiness of an associated meta-complexity problem.
Finally, we discuss an unexpected converse direction, showing how sufficiently powerful learning algorithms can yield computational hardness results and new lower bounds in meta-complexity. For a more detailed description of each section of the survey, we refer to \Cref{sec:organisation} below.

We place particular emphasis on recent developments in these rapidly evolving areas. These include new connections among learning, time-bounded Kolmogorov complexity, meta-complexity, cryptography, average-case complexity, and universal sequence extrapolation. Our aim is not only to state these results, but also to explain the underlying ideas and, where appropriate, provide proofs or detailed proof sketches. The emerging theory remains far from complete, so open problems form an integral part of the survey.

\paragraph{Intended Audience and Required Background.}
The primary audience includes students and researchers interested in   computational complexity, computational learning, and theoretical computer science more broadly. We aim to make the main results and proof ideas accessible without assuming prior expertise in every area involved. The preliminaries review the necessary concepts and definitions, and provide pointers to standard textbooks that cover the required material in more detail.

We also hope that researchers and students in machine learning will find the survey useful as an entry point to the complexity-theoretic foundations of learning and compression.
We warn that some arguments might omit standard background or routine intermediate steps.
Readers may find it helpful to use their preferred AI assistant as an interactive companion for unpacking definitions, reconstructing such steps, or translating between the terminology of theoretical computer science and intuitions from machine learning.\\

We stress that the selection of topics necessarily reflects the authors' background and perspective. In particular, we do not attempt to provide a comprehensive account of learning and compression from the perspective of statistics, information theory, or theoretical machine learning. We refer the reader to \Cref{sec:related_work} for some additional pointers and complementary perspectives. 

Finally, to present a general theory, much of this survey considers expressive learning problems involving unrestricted polynomial-size Boolean circuits or all efficiently samplable example distributions. As explained later in the text, these learning problems are intractable under  cryptographic conjectures. However, many ideas and techniques surveyed here also apply to more restricted learning tasks, and such tasks are more likely to admit efficient algorithms. It is possible, therefore, that some  results can inform practical methods or analyses, and it would be interesting to see further work along these lines.  In fact, going beyond the formal connections covered here, the compression perspective has found numerous concrete applications in both classical and contemporary machine learning. A comprehensive treatment of the rich empirical literature on learning and compression lies beyond the scope of this survey, although we provide some pointers in \Cref{sec:related_work}.

\subsection{Organisation}\label{sec:organisation}

The remainder of the survey is organised as follows.

\begin{itemize}
\item \textbf{\Cref{sec:preliminaries}: Preliminaries.}
We establish the notation and background used throughout the survey.
After reviewing the required probability and information-theoretic notions, we define the supervised and unsupervised learning frameworks considered later, including sequence extrapolation, distribution learning, PAC learning, agnostic learning, and heuristic learning.
We also recall the necessary concepts from cryptography, circuit complexity, Kolmogorov complexity, meta-complexity, and natural properties.

\item \textbf{\Cref{sec:classical_connections}: Classical Connections Between Learning and Compression.}
We present two fundamental roles of compression in supervised learning.
Model compression formalises the Occam principle that a short hypothesis fitting the observed labels is unlikely to generalise poorly, while sample compression retains a small portion of the labelled sample from which a hypothesis can be reconstructed.
These results illustrate, without computational efficiency considerations, how the amount of information retained from a sample controls generalisation.

\item \textbf{\Cref{sec:complexity_learning}: Complexity-Theoretic Foundations of Learning and Compression.}
We compare five increasingly weak ``learning paradises'', ranging from universal sequence extrapolation and distribution learning to heuristic PAC learning, and establish reductions among them.
We then place these learning tasks within Impagliazzo's worlds, relating their feasibility to average-case complexity, meta-complexity, and cryptographic hardness.
Finally, we study white-box learning, where the learner receives the code of a sampler generating its labelled examples, and connect this additional access to the inversion of one-way functions.

\item \textbf{\Cref{sec:learning_via_meta}: Learning via Meta-Complexity and Algorithmic Information Theory.}
We explain how algorithms that distinguish low complexity objects from random ones can be combined with reconstruction procedures of pseudorandom generators to obtain supervised learning algorithms.
For unsupervised learning, \Cref{sec:solomonoff} revisits Solomonoff induction and asks when prediction based on shortest programs can be implemented efficiently using estimates of time-bounded Kolmogorov complexity.
We also present equivalences between learning and meta-complexity in \Cref{sec:learning_meta-equivalences}, and a technique to boost the running time of learning algorithms in \Cref{sec:speedup}.

\item \textbf{\Cref{sec:shannon_connections}: Learning vs Compression via Shannon's Information Theory.}
We develop the correspondence among source coding, universal coding, and distribution learning, first at the information-theoretic level and then under computational efficiency constraints.
The discussion covers arithmetic coding, next-bit prediction, source coding theorems for time-bounded Kolmogorov complexity, conditional sampling, and symmetry of information.
We conclude by relating this framework to autoregressive language modelling, where next-token probabilities determine predictions and code lengths.

\item \textbf{\Cref{sec:learning_to_hardness}: Meta-Complexity and Computational Hardness from Learning.}
We study the converse direction, explaining how nontrivial deterministic, randomized, and quantum learners can imply circuit size lower bounds, a major challenge in computational complexity theory.
We then show how this paradigm can establish unconditional lower bounds for meta-complexity problems.
In particular, \Cref{sec:uncond_lbs} describes a general template in which an assumed algorithm for measuring compressibility produces a learner whose complexity-theoretic consequences contradict the original easiness assumption.

\item \textbf{\Cref{sec:related_work}: Related Work.}
We provide pointers to related connections and research directions outside the survey's main scope, including further interactions among learning, compression, algorithmic information theory, complexity theory, and cryptography. We also provide links to the empirical literature on learning and compression and to related practical aspects of machine learning.
\end{itemize}

\medskip \medskip

\noindent \emph{Acknowledgements.} We thank Shuichi Hirahara for helpful discussions on deriving meta-complexity lower bounds from deterministic learning, which led to Item (\emph{ii}) of \Cref{thm:lb_from_det_learning}. We also thank Valentine Kabanets for helpful feedback on an earlier draft of this survey.\\ 

\medskip

\noindent \emph{Use of AI Tools.} We used AI tools to search for relevant literature, assist with the preparation of diagrams, identify typographical errors, check proofs for correctness, and make  textual edits to improve the exposition.\\

\medskip

\noindent \emph{Funding Information.} This work was supported in part by the UKRI Frontier Research Guarantee Grant EP/Y007999/1 and the Centre for Discrete Mathematics and its Applications (DIMAP) at the University of Warwick. The authors would also like to thank the Isaac Newton Institute for Mathematical Sciences, Cambridge, for support and hospitality during the programme Logical Foundations of Complexity Theory (LFC), where part of the work on this paper was undertaken. The LFC programme was supported by EPSRC grant EP/Z000580/1.

\newpage

\section{Preliminaries}\label{sec:preliminaries}

\subsection{Basic Notation}

For $m \geq 1$, we let $[m] = \{1, \ldots ,m\}$.

We denote the empty string by $\epsilon$, and the length of a string $x$ by $|x|$. For strings $x \in \{0,1\}^n$ and $y \in \{0,1\}^m$, their concatenation is written 
as $x y \in \{0,1\}^{n+m}$. 

Given a string $z \in \{0,1\}^n$ and a set $S \subseteq [n]$, we let $z|_S \in \{0,1\}^{|S|}$ denote the string obtained by restricting $z$ to the coordinates in $S$.

When considering algorithms, PPT denotes probabilistic polynomial time.

\subsection{Probability Theory and Information Theory}

If $x \sim D$, then $x$ is sampled from the distribution $D$. For a set $S$, $x \sim S$ denotes a uniformly random element from $S$. 
For any $a$ in the support of $D$, $D(a)$ denotes its probability.  
We write $\mathcal{U}_\ell$ for the uniform distribution over $\{0,1\}^\ell$. We use $\Delta_\mathsf{TV}(P, Q)$ to denote the statistical distance of two distributions with the same support. We will also use KL-divergence and its relation to statistical distance, as reviewed below.

\begin{definition}[KL Divergence]
Let $P$ and $Q$ be distributions on the same finite set $\Omega$. The \emph{Kullback--Leibler divergence}
from $P$ to $Q$ is defined by
\[
\KL(P \| Q)
\;\eqdef\;
\Ex_{x \sim P}\!\left[\log \frac{P(x)}{Q(x)}\right]
\;=\;
\sum_{x \in \Omega} P(x)\log \frac{P(x)}{Q(x)}.
\]
\end{definition}

\begin{lemma}[Pinsker's Inequality]\label{lem:Pinsker}
Let $P$ and $Q$ be distributions on the same finite set. Then
$$
\Delta_\mathsf{TV}(P,Q)^2 \;\le\; \frac{1}{2}\,\KL(P \| Q).
$$
\end{lemma}

\begin{definition}[Conditional KL Divergence]
Let $(P_X,P_Y)$ and $(Q_X,Q_Y)$ be random pairs on the same finite
product set $\mathcal{X}\times\mathcal{Y}$, and suppose that
$\supp(P_X)\subseteq\supp(Q_X)$.
The \emph{conditional Kullback--Leibler divergence} from $P_Y$
conditioned on $P_X$ to $Q_Y$ conditioned on $Q_X$ is defined by
\[
\KL(P_Y\mid P_X\|Q_Y\mid Q_X)
\;\eqdef\;
\Ex_{x\sim P_X}
\left[
\KL\bigl(
P_Y\mid P_X=x
\|
Q_Y\mid Q_X=x
\bigr)
\right].
\]
\end{definition}

\begin{lemma}[Chain Rule for KL Divergence]\label{lem:chainruleKL}
Let $(P_{1},\ldots,P_{m})$ and
$(Q_{1},\ldots,Q_{m})$ be random vectors on the same finite
product set $\mathcal{X}_1\times\cdots\times\mathcal{X}_m$,
and suppose that 
$\supp(P_{1},\ldots,P_{m})
\subseteq
\supp(Q_{1},\ldots,Q_{m})$.
For every $i\in[m]$, write
\[
P_{{<i}}
\eqdef
(P_{1},\ldots,P_{{i-1}})
\qquad\text{and}\qquad
Q_{{<i}}
\eqdef
(Q_{1},\ldots,Q_{{i-1}}),
\]
where $P_{{<1}}$ and $Q_{{<1}}$ denote the empty tuple. Then
\begin{align*}
    \KL\bigl(
(P_{1},\ldots,P_{m})
\|
(Q_{1},\ldots,Q_{m})
\bigr)
&=
\sum_{i=1}^m
\KL\bigl(
P_{i}\mid P_{<i}
\|
Q_{i}\mid Q_{<i}
\bigr)\\
&= m\cdot \Ex_{i\sim [m]}\left[\KL\bigl(
P_{i}\mid P_{<i}
\|
Q_{i}\mid Q_{<i}\bigr)\right].
\end{align*}

\end{lemma}

We use the following definition of entropy.
\begin{definition}[Entropy]
For a probability distribution $\D$ over a finite set, define the entropy of $\D$ as
\[
\Ent{\D} \eqdef \sum_{x\in \supp{\D}} \D(x)\log \frac{1}{\D(x)}.
\]
\end{definition}

For more information about information theory and its applications, see \citep{cover1999elements,polyanskiy2025information}.

\subsection{Learning Frameworks}\label{sec:learning-frameworks}

In \emph{supervised learning}, the learner receives examples together with the desired answers, usually labels produced by an unknown target function, and the goal is to predict the correct label on new examples. In \emph{unsupervised learning}, the learner receives samples without such labels and must instead capture some structure of the source that generated them. 

In the learning models introduced below, sequence extrapolation and distribution learning are unsupervised, since the input consists only of samples from an unknown process or distribution. On the other hand, PAC learning and heuristic learning are supervised, since the learner is given labelled examples of the form $(x,f(x))$ and aims to predict the label of future examples.

\subsubsection{Sequence Extrapolation}

A common way to learn from data is to ask what should come next. After observing a prefix of a sequence, a learner tries to reproduce the distribution of the next block without being told the process that generated the sequence. This viewpoint captures settings ranging from sequential prediction to universal induction: the past is useful precisely when it reveals regularity in an unknown source. Here we make that intuition computational by assuming only that the sequence is produced by a short randomized program running within a prescribed time bound. Since the source itself may be random, the learner is not expected to guess the next block exactly. Instead, it should make a probabilistic prediction close to the source's true conditional distribution.

Fix a \emph{sequence length} $m \in \mathbb{N}$, a \emph{block length} $n \in \mathbb{N}$, an \emph{error bound}  $\varepsilon > 0$, a \emph{confidence value} $\delta > 0$, a \emph{description length} $s \in \mathbb{N}$, and a time bound $t \in \mathbb{N}$.

We consider \emph{randomized} Turing machines $M$ of description length at most $s$ (according to a fixed universal randomized machine $U^\mathsf{r}$). We will be interested in the random variable $U^\mathsf{r}(M,\epsilon, 1^t)$ encoding the output of $M$ on the empty string $\epsilon$ when $U^\mathsf{r}$ runs over the input pair $(M,\epsilon)$ for $t$ steps. If this simulation does not produce an output, we set the output to $0^{n \cdot m}$. For convenience, for a choice of $n$ and $m$, if the output $x$ of $U^\mathsf{r}(M,\epsilon, 1^t)$ is of length less than $n \cdot m$, we add $0$'s to it and consider instead the string $x' = x0^{nm - |x|}$. Therefore, given $n$ and $m$, we will view $x \sim U^\mathsf{r}(M,\epsilon, 1^t)$ as a \emph{random variable} $(x_1, \ldots, x_m)$ distributed over $\{0,1\}^{n \cdot m}$. For simplicity, we might omit the empty string $\epsilon$.

We consider a task called $(m,n,s,t,\varepsilon,\delta)$-\emph{extrapolation}. In this computational task, we consider a randomized algorithm $L$ such that, for every (unknown) choice of $M$ of length at most $s$ and index $i$, given the first $i-1$ blocks $(x_1, \ldots, x_{i-1})$ of $x \sim U^\mathsf{r}(M,1^t)$ (denoted by $x_{<i}$), the goal is for $L(x_{<i}) \sim \{0,1\}^n$ to be distributed according to $(\mathcal{X}_i \mid {x_{<i}})$, where for a given choice of $n,m, M, t$, the random variable $(\mathcal{X}_i \mid {x_{<i}})$ denotes the conditional distribution of the $i$-th block $x_i$ given preceding blocks $x_1, \ldots, x_{i-1} \in \{0,1\}^{n}$. Formally, we have the following definition.

\begin{definition}[Sequence Extrapolation]
We say that a randomized algorithm $L$ \emph{solves}  $(m,n,s,t,\varepsilon,\delta)$-\emph{extrapolation} if for every choice of $M \in \{0,1\}^{\leq s}$ the following holds:
$$
\Pr_{i \sim [m],\;x \sim U^\mathsf{r}(M,1^t)}\Big [ \Delta_{\mathsf{TV}} \big (L(x_{<i}),(\mathcal{X}_i \mid {x_{<i}}) \big ) \leq \varepsilon \Big ]  \geq 1 - \delta.
$$
We assume that the parameters $(m,n,s,t,\varepsilon^{-1},\delta^{-1})$ are also given to $L$ in unary when considering uniform algorithms and parameter sequences.    
\end{definition}

For broader treatments of sequential prediction and algorithmic approaches to extrapolation, see \citep{CesaBianchiLugosi2006,hutter2024introduction}.

\subsubsection{Learning Distributions}

Sometimes the goal is not merely to predict the next observation, but to construct a model of the source itself. Given independent samples from an unknown distribution, a distribution learner seeks a new sampling procedure whose outputs are statistically indistinguishable from those of the original source, up to a prescribed error. Such a model captures enough of the source's structure to generate fresh representative data, even though it need not reveal how the source was originally implemented. 

Formally, we consider the task of learning an unknown discrete distribution $D$ from a fixed class $\mathcal{D}$ of distributions \citep{kearns1994learnability}. The learner has access to i.i.d.~samples from $D$, and the goal is to output a circuit $C \colon \{0,1\}^T \to \{0,1\}^n$ that samples a distribution that is $\varepsilon$-close to $D$.

\begin{definition}[Learning Distributions \citep{kearns1994learnability}]\label{def:distlearning}
    Let $\mathcal{D} = \{\mathcal{D}_n\}_{n \geq 1}$ be a class of distributions, where each $\mathcal{D}_n$ is a set of distributions supported over $\{0,1\}^n$. The class $\mathcal{D}$ is said to be learnable in time $T(n,\varepsilon,\delta)$ and with sample complexity $m(n,\varepsilon,\delta)$ if there is a randomized algorithm $A$ such that, for every $n$, $D \in \mathcal{D}_n$, $\varepsilon > 0$, $\delta > 0$, and $m = m(n,\varepsilon,\delta)$, $A(1^n,1^{\varepsilon^{-1}}, 1^{\delta^{-1}}, \cdot)$ runs in time at most $T = T(n,\varepsilon, \delta)$ and 
    \begin{equation}\label{eq:dist-learned}
     \Pr_{A,\; \vec{x} \sim D^{(m)}}\Big [ C \leftarrow A(1^n,1^{\varepsilon^{-1}} 1^{\delta^{-1}}, \vec{x}) \;\; \text{and} \;\;    \Delta_{\mathsf{TV}}(C(\mathcal{U}_T), D) \leq \varepsilon \Big ] \geq 1 - \delta,   
    \end{equation}
    where $C(\mathcal{U}_T)$ is the distribution sampled by $C$, and $\vec{x} \sim D^{(m)}$ denotes a tuple of $m$ i.i.d.~samples from $D$.
\end{definition}

For convenience, when \Cref{eq:dist-learned} holds for a distribution $D$, we say that $A$ $\varepsilon$-learns $D$ with probability at least $1 - \delta$.

For functions $s,t \colon \mathbb{N} \to \mathbb{N}$, we can consider the class  $\mathcal{D}^{s,t} = \{\mathcal{D}^{s(n),t(n)}_n\}_{n \geq 1}$ of distributions, where each $\mathcal{D}^{s',t'}_{n'}$ consists of the set of all distributions $D$ supported over $\{0,1\}^{n'}$ sampled by a randomized machine $M$ that runs in time at most $t'$ and has description length at most $s'$.

Similarly, we can consider the distribution learning problem for $\mathcal{D}^{s,t}_n$ when $n, s, t \in \mathbb{N}$ are \emph{input parameters} provided to the learning algorithm $A(1^n,1^{\varepsilon^{-1}}, 1^{\delta^{-1}}, 1^t, 1^s, \cdot)$.

\subsubsection{PAC Learning}

PAC (Probably Approximately Correct) learning formalises the familiar supervised task of learning a rule from labelled examples. An unknown target concept assigns a label to each input, and the learner observes examples drawn from a distribution that determines which inputs matter. From sufficiently many such examples, the learner should produce a hypothesis that makes few mistakes on a fresh input from the same distribution. The name highlights the two allowances that make the goal realistic: the hypothesis need only be approximately correct, and the learner need only succeed with high probability. In other words, PAC learning asks for reliable generalisation to the inputs that are likely to be encountered.

Formally, in the PAC learning framework, one considers an unknown function $f \colon \bool^n \to \bool$ for some $n \in \nat$. Such a function is referred to as a \emph{concept}. A \emph{concept class} $\mathcal{C}$ refers to a set of concepts, with $\mathcal{C}_n$ denoting $\mathcal{C} \cap \{f: \bool^{n} \to \bool\}$ for a fixed $n \in \nat$. The PAC learnability of a concept class is defined with respect to a class $\mathcal{D}$ of probability distribution ensembles $D = \{D_n\}_{n \in \nat}$, where each $D_n$ is supported over $\bool^n$.

PAC learners are provided a \emph{labelled example oracle} $\EX(f, D_n)$, with each query returning an independently sampled pair $(x, f(x))$, where $x \sim D_n$.

\begin{definition}[PAC Learning \cite{DBLP:journals/cacm/Valiant84}; see \cite{DBLP:journals/corr/abs-2511-08791} for a survey]

    Let $\mathcal{C}$ be a concept class, and let $\mathcal{D}$ be a class of distributions. An algorithm $A$ \emph{PAC learns} $\mathcal{C}$ over $\mathcal{D}$ if, for all $D \in \mathcal{D}$, all sufficiently large $n \in \nat$, all $\delta, \varepsilon \in (0,1)$, and all $f \in \mathcal{C}_n$,
    \[\Pr_{A; \; \EX(f, D_n)}\left [ \; A^{\EX(f, D_n)}(1^n, 1^{\varepsilon^{-1}}, \delta^{-1}) \text{ outputs } h \text{ such that } \Pr_{x\sim D_n}[h(x) = f(x)] \geq 1 - \varepsilon \; \right] \geq 1 - \delta.\]
    
\end{definition}

For $s: \nat \to \nat$, we can consider the class $\mathcal{C}^s = \{\mathcal{C}^{s(n)}_n\}_{n \in \nat}$ of concepts, where each $\mathcal{C}^{s(n)}_n$ is the subset of $\mathcal{C}_n$ with representations of size $s(n)$. For example, if $\mathcal{C} = \P/\poly$, then $\mathcal{C}^s = \SIZE[s]$, the set of functions $f: \bool^n \to \bool$ computable by circuits of size $s(n)$. As in the previous section, we define $\mathcal{D}^{a, t}$ to be the class of distributions samplable by a $t$-time machine with description length $a$. To make these parameters explicit, $s, t, a$ may be given as inputs to the learning algorithm.

One may also consider the model of PAC learning with \emph{membership queries}, which is exactly the same as the definition above, except the learner $A$ is additionally provided an $f$-oracle. 

A learner that always outputs a hypothesis $h$ belonging to $\mathcal{C}$ is called \emph{proper}. Otherwise, the learner is \emph{improper}. 

If $\varepsilon$ is not given as an input but the hypothesis output by $A$ satisfies $\Pr_{x\sim D_n}[h(x) = f(x)] \geq 1/2 + \gamma(n)$ with probability $1 - \delta$ for some $\gamma(n) = \poly^{-1}(n)$, then $A$ is referred to as a \emph{weak learner}.

For textbook introductions to computational learning theory and the PAC framework, see \citep{KearnsVazirani:94,ShalevShwartzBenDavid2014}.

\subsubsection{Agnostic Learning}

The PAC model assumes that every observed label is generated by some target concept from the class under consideration. In practice, labels may be noisy, internally inconsistent, or governed by a relationship that the chosen class cannot represent. \emph{Agnostic learning} removes this assumption and asks what can still be learned from an arbitrary joint distribution over inputs and labels. Since nearly perfect prediction may then be impossible, success is measured relative to the best concept available in the class. The learner must perform almost as well as that benchmark without knowing either the source of the labels or the benchmark concept in advance. This formulation preserves a meaningful learning objective while allowing both noise and model misspecification.

In technical terms, one cannot assume that  the given examples are labelled according to a function $f$ belonging to the relevant concept class $\mathcal{C}$. In fact, one cannot assume that the labels correspond to a fixed function at all. In particular, we now consider a class $\mathcal{E}$ of joint distribution ensembles $E = \{E_n\}_{n \in \nat}$, where each $E_n$ is supported over $\bool^n \times \bool$. The example oracle $\EX(E_n)$ simply returns an independent random sample from $E_n$ each time it is queried.

Now, there is clearly no hope in general for an efficient learner to produce a hypothesis with high correctness probability over $E$ (i.e., $\Pr_{(x,b) \sim E_n}[h(x) = b] \approx 1$).\footnote{For example, consider $E_n = (\mathcal{U}_n, \mathcal{U}$). This joint distribution can be sampled efficiently, but no fixed circuit has correctness probability greater than $1/2$.} Therefore, as alluded to above, we measure the success of an agnostic learner relative to the concept class $\mathcal{C}$ in question, asking for a hypothesis that does ``almost as well'' as the best concept in $\mathcal{C}$. To that end, for $n \in \nat$, define
\[{\sf opt}_n(\mathcal{C}, E) \eqdef \max_{c \in \mathcal{C}_n}\left\{\Pr_{(x, b) \sim E_n}[c(x) = b]\right\}.\]

\begin{definition}[Agnostic Learning \cite{DBLP:journals/ml/KearnsSS94, DBLP:journals/iandc/Haussler92}]

    Let $\mathcal{C}$ be a concept class, and let $\mathcal{E}$ be a class of joint distributions. An algorithm $A$ \emph{agnostically learns} $\mathcal{C}$ over $\mathcal{E}$ if, for all $E \in \mathcal{E}$, all sufficiently large $n \in \nat$, and all $\delta, \varepsilon \in (0,1)$,
    \[\Pr_{A; \; \EX(E_n)}\left [ \; A^{\EX(E_n)}(1^n, 1^{\varepsilon^{-1}}, \delta^{-1}) \text{ outputs } h \text{ such that } \Pr_{(x, b)\sim E_n}[h(x) = b] \geq {\sf opt}_n(\mathcal{C}, E) - \varepsilon \; \right] \geq 1 - \delta.\]
    
\end{definition}

In this context, if we say that an algorithm agnostically learns over the distribution class $\mathcal{D}^{a, t} = \{\mathcal{D}^{a(m),t(m)}_m\}_{m \geq 1}$ as defined above, we interpret each $\mathcal{D}^{a(m), t(m)}_m$ as a set of joint example distributions supported over $\bool^{m - 1} \times \bool$.

For a textbook treatment of agnostic learning, see \citep{ShalevShwartzBenDavid2014}.

\subsubsection{Heuristic Learnability}

Standard supervised learning guarantees usually require the learner to succeed for every target concept in the class. In complexity theory and cryptography, however, target functions often arise from a randomized generation procedure, just as the inputs on which they are evaluated arise from their own sampling process. \emph{Heuristic learnability} reflects this situation by measuring prediction success on average over both the chosen target and the chosen input. The learner may therefore fail on exceptional targets or examples, provided that these failures are sufficiently rare in the overall experiment. Even a small advantage over random guessing can be significant in this setting, making the model especially useful for expressing connections among prediction, average-case hardness, and cryptographic security.

Let $\mathcal{C}$ be a concept class, with each $f \in \mathcal{C}_n$ having a representation $\sigma_f$ of size at most polynomial in $n$. Let $P = \{P_n\}_{n \in \nat}$ be a distribution ensemble, with each $P_n$ supported over representations $\sigma_f$ of target functions $f \in \mathcal{C}_n$ such that $f$ can be evaluated in polynomial time given $\sigma_f$. Let $D = \{D_n\}_{n \in \nat}$ be another distribution ensemble, with each $D_n$ supported over example strings $x \in \bool^n$.

\begin{definition}[\cite{DBLP:conf/crypto/BlumFKL93}]\label{def:bfkl}

    For any $\varepsilon: \nat \to (0, 1/2)$, we say that $\mathcal{C}$ is \emph{$\varepsilon$-predictable on average} over $P$ and $D$ if there exists an algorithm $A$ such that for all sufficiently large $n \in \nat$, 
    \[\Pr_{\substack{A; \; \sigma_f \sim P_n; \\ \EX(f, D_n);\\ x' \sim D_n }}\left[A^{\EX(f, D_n)}(x') = f(x')\right] \geq 1 - \varepsilon(n).\]
    
\end{definition}

We say that $\mathcal{C}$ is \emph{strongly predictable on average} over $P$ and $D$ if it is $(1/q(\cdot))$-predictable on average for any polynomial $q$. We say it is \emph{weakly predictable on average} over $P$ and $D$ if, for some polynomial $q'$, it is $(1/2 - 1/q'(\cdot))$-predictable on average. We mention that weak predictability for the case of $\mathcal{C} = \P/\poly$ and all $P, D \in \P\SAMP$ implies strong predictability for such $\mathcal{C}$, $P$, and $D$ \cite{DBLP:conf/crypto/BlumFKL93}.

For $t: \nat \to \nat$, let $\SAMP[t]$ denote the class of distributions samplable uniformly in time at most $t$. In the case of a heuristic learner $A$ working on some fixed concept class $\mathcal{C}^s$, $P \in \SAMP[t_1]$, and $D \in \SAMP[t_2]$, the parameters $s, t_1, t_2$ may be given as additional inputs to $A$. 

For further background on heuristic learnability and related notions, see \citep{nanashima2021theory}.

\subsection{Theoretical Computer Science}\label{sec:prelim_TCS}

We will require several notions from cryptography, complexity theory, Kolmogorov complexity, pseudorandomness, and meta-complexity. For general background in computational complexity, including randomized computation, nonuniform computation, cryptography, and circuit complexity, see \citep{arora2009computational}.

\subsubsection{Cryptography}\label{sec:prelim_crypto}

\paragraph{Infinitely-Often Cryptography.}

For the purposes of cryptography, it is important for security to hold \emph{almost everywhere}. That is, one would hope for a cryptographic primitive that cannot be broken at any (but perhaps finitely many) input lengths. Likewise, for the purposes of learning and positive algorithmic tasks more broadly, one would hope for algorithms to be successful at all but finitely many input lengths. While it is reasonable to ask for this property in both cases, it poses a problem if one would like to investigate dualities between the possibility of secure cryptography and the hardness of tasks such as learning. In particular, if the hardness of learning at some fixed input length translates into the hardness of breaking a cryptographic primitive at a related input length, then the impossibility of learning in the sense described above would only yield security at infinitely many input lengths.

For this reason, to aid in a better understanding of connections between algorithmic hardness and cryptography, we consider cryptography that may only be secure \emph{infinitely often}.

A \emph{pseudorandom function family} is a family of keyed polynomial-time computable functions that look like a random oracle to any efficient observer. Here, we state the definition of an infinitely-often variant.

\begin{definition}[Infinitely-Often Pseudorandom Function]

    An \emph{infinitely-often pseudorandom function}, denoted ioPRF, is a polynomial-time computable function family $f = \{f: \bool^{n} \times \bool^{n} \to \bool\}_{n \in \nat}$ such that, for any inverse polynomial $\delta: \nat \to [0,1]$ and PPT machine $A$, there exist infinitely many $n \in \nat$ such that 
    \[\left | \Pr_{A; \; u \sim \bool^{n}}\left[A^{f(u, -)}(1^n) = 1\right] - \Pr_{A; \; R_{n}}\left[A^{R_{n}(-)}(1^n) = 1\right] \right | < \delta(n),\]
    where $R_{n}:\bool^{n}\to \bool$ is a uniformly random function over the set of all functions with these parameters.
    
\end{definition}

\paragraph{Auxiliary-Input Cryptography.} We also consider variants of cryptographic primitives that are allowed to make use of an additional publicly available non-uniform input string. Since this kind of cryptography generalizes the standard setting in which no such auxiliary string is given, its existence is a weaker assumption.

As with infinitely-often cryptography, auxiliary-input cryptography can arise naturally when investigating connections between secure cryptography and the hardness of learning. In particular, a specific hard instance of learning might translate to a related description of a secure cryptographic primitive. In such settings, one may only obtain auxiliary-input cryptography rather than standard (uniform) cryptography.

An \emph{auxiliary-input function family} 
\[f = \{f_z: \bool^{n(|z|)} \to \bool^{\ell(|z|)}\}_{z \in \bool^*}\]
is said to be \emph{polynomial-time computable} if there is a polynomial-time machine $F$ such that for all $z \in \bool^*$ and $x \in \bool^{n(z)}$, $F(z, x) = f_z(x)$.

A one-way function is a function that may be computed efficiently but cannot be inverted efficiently. This is a basic building block of many other objects in private-key cryptography. Here, we state the definition of the auxiliary-input variant.

\begin{definition}[Auxiliary-Input One-Way Function \cite{DBLP:conf/istcs/OstrovskyW93}]

    An \emph{auxiliary-input one-way function}, denoted AIOWF, is a polynomial-time computable auxiliary-input function family $f$ such that, for any inverse polynomial $\delta$ and PPT machine $A$, there exist infinitely many $z \in \bool^*$ such that 
    \[\Pr_{A; \, x \sim \mathcal{U}_{n(|z|)}}\left[A(z, f_z(x)) \in f_z^{-1}(f_z(x))\right] < \delta(|z|).\]
    
\end{definition}

Impagliazzo and Luby \cite{DBLP:conf/focs/ImpagliazzoL89} proved that, if all candidate one-way functions can be inverted in the sense of finding \emph{some} pre-image of a given output, then they can be inverted in the stronger sense of finding \emph{nearly uniformly random} pre-images. Thus, the existence of one-way functions is equivalent to the existence of \emph{distributional one-way functions}, which are secure against adversaries finding nearly uniformly random pre-images. More specifically, for a uniformly random input $x$ and a poly-time computable function $f$, no adversary $A$ on input $f(x)$ returns $x'$ such that $(x', f(x))$ is close in statistical distance to $(x, f(x))$. Below, the auxiliary-input generalization is stated.

\begin{definition}[Distributional Auxiliary-Input One-Way Function]

    A \emph{distributional auxiliary-input one-way function}, denoted distAIOWF, is a polynomial-time computable auxiliary-input function family $f = \{f_z: \bool^{n(|z|)} \to \bool^{\ell(|z|)}\}_{z \in \bool^*}$ such that, for some inverse polynomial $\delta$, for every PPT machine $A$, there exist infinitely many $z \in \bool^*$ such that 
    \[\Delta\left( \; (A(z, f_z(x)), f_z(x)), \;  (x, f_z(x)) \; \right) > \delta(|z|),\]
    where $x \sim \mathcal{U}_{n(|z|)}$, and $\Delta$ denotes statistical (total variation) distance.
    
\end{definition}

A \emph{pseudorandom generator} is a deterministic machine that maps a short uniformly random \emph{seed} to a longer output string such that the distribution on outputs is indistinguishable from uniformly random to efficient observers. 

\begin{definition}[Auxiliary-Input Pseudorandom Generator]

    An \emph{auxiliary-input pseudorandom generator}, denoted AIPRG, is a polynomial-time computable auxiliary-input function family $G = \{G_z: \bool^{n(|z|)} \to \bool^{\ell(|z|)}\}_{z \in \bool^*}$ such that $\ell(|z|) > n(|z|)$, and for any inverse polynomial $\delta$ and PPT machine $A$, there exist infinitely many $z \in \bool^*$ such that
    \[\left | \Pr_{A; \; v \sim \mathcal{U}_{n(|z|)}}\left[A(z, G_z(v)) = 1\right] - \Pr_{A; \; u \sim \mathcal{U}_{\ell(|z|)}}\left[A(z, u) = 1\right]\right| \leq \delta(|z|).\]

    We use the statement \textit{$\exists$ AIPRG} to denote that the above exists for any choice of $\ell(|z|) = n(|z|)^{O(1)}$.
    
\end{definition}

\begin{definition}[Auxiliary-Input Pseudorandom Function]

    An \emph{auxiliary-input pseudorandom function}, denoted AIPRF, is a polynomial-time computable auxiliary-input function family $f = \{f_z: \bool^{n(|z|)} \times \bool^{n(|z|)} \to \bool\}_{z \in \bool^*}$ such that, for any inverse polynomial $\delta$ and PPT machine $A$, there exist infinitely many $z \in \bool^*$ such that 
    \[\left | \Pr_{A; \; u \sim \bool^{n(|z|)}}\left[A^{f_z(u, -)}(z) = 1\right] - \Pr_{A; \; R_{|z|}}\left[A^{R_{|z|}(-)}(z) = 1\right] \right | < \delta(|z|),\]
    where $R_{|z|}: \bool^{n(|z|)} \to \bool$ is a uniformly random function over the set of all functions with these parameters.
    
\end{definition}

The next statement results from auxiliary-input analogues of standard transformations between one-way functions, distributional one-way functions, pseudorandom generators, and pseudorandom functions \cite{DBLP:conf/focs/ImpagliazzoL89, DBLP:journals/siamcomp/HastadILL99, DBLP:journals/jacm/GoldreichGM86}. See, e.g., \cite{DBLP:conf/focs/ApplebaumBX08} for a discussion.

\begin{lemma}\label{lem:ai-crypto}
    The following statements are equivalent.
    \begin{itemize}
        \item $\exists$ AIOWF;
        \item $\exists$ distAIOWF;
        \item $\exists$ AIPRG;
        \item $\exists$ AIPRF.
    \end{itemize}
\end{lemma}

For more information about one-way functions, pseudorandom generators, and pseudorandom functions, see \citep{DBLP:books/cu/Goldreich2001, katz2007introduction}.

\subsubsection{Kolmogorov Complexity}

Kolmogorov complexity can be viewed as a universal measure of compression: it quantifies how concisely an individual object can be described by an effective procedure. Any particular compression scheme provides a way to encode certain objects succinctly, while a universal programming language can simulate every such scheme with only a fixed additional overhead. Consequently, Kolmogorov complexity captures, up to an additive constant, the best compression achievable by any computable method. Similarly, time-bounded Kolmogorov complexity captures the best compression achievable when decompression is required to run within a prescribed time bound.\footnote{We refer to \citep{blelloch2001introduction} for an introduction to data compression.}

Formally, we fix a time-efficient universal machine $U$, and define time-bounded Kolmogorov complexity with respect to $U$ (see \citep{DBLP:series/txcs/LiV19} for the required background in Kolmogorov complexity). In the context of Kolmogorov complexity, whenever we say that a program $M$ running on input $y$ outputs $x$ in $t$ steps, formally, this means that $U(M, y, 1^t) = x$.

	\begin{definition}[$\mathsf{K}^t$ Complexity]
	For $x,y\in\bool^*$ and $t\in\mathbb{N}$, the \emph{$t$-time-bounded Kolmogorov complexity} of $x$    conditioned on $y$ is defined as 
		\[
		\mathsf{K}^{t}(x\mid y) \eqdef \min_{M \in \{0,1\}^*} \mleft\{|M|  \,\,\middle\vert\,\,  U(M, y, 1^t) \textnormal{ outputs $x$}\mright\}.
		\]
        We define $\K^{t}(x)$ as $\K^{t}(x\mid \epsilon)$.
	\end{definition}

The definition of $\K^t$ requires us to fix the time bound in advance. On the other hand, Levin's notion of time-bounded Kolmogorov complexity simultaneously considers description length and running time.

\begin{definition}[Levin's $\Kt$ Complexity \citep{DBLP:journals/iandc/Levin84}] For $x,y\in\bool^*$, the \emph{Levin complexity} of $x$  conditioned on $y$ is defined as 
    \[
		\Kt(x \mid y) \eqdef \min_{t \in \mathbb{N},\;M \in \{0,1\}^*} \mleft\{|M| + \lceil \log t \rceil   \,\,\middle\vert\,\,  U(M, y, 1^t) \textnormal{ outputs $x$}\mright\}.
		\]
         We define $\Kt(x)$ as $\Kt(x\mid \epsilon)$.
\end{definition}

   When no time bound is considered, we simply write $\K(x)$ to denote the (classical time-unbounded) Kolmogorov complexity of the string $x$.

We also consider \emph{randomized} notions of time-bounded Kolmogorov complexity, in which the universal machine is given an additional random string as input. These complexity measures are particularly suitable in computational contexts where computations can be randomized, as is often the case with learning algorithms. Note that these measures include an additional parameter $\lambda \in [0,1]$ to denote the probability with which the machine outputs the string in question. If $\lambda$ is omitted, we take $\lambda = 2/3$.
    
\begin{definition}[$\mathsf{rK}^t$]
    For $x,y\in\bool^*$, $\lambda\in [0,1]$, $t\in\mathbb{N}$, the \emph{randomized $t$-time-bounded Kolmogorov complexity} of $x$ conditioned on $y$ is defined as
   \[
        \rK^{t}_\lambda(x \mid y) \eqdef \min_{M\in\bool^*}\mleft\{|M|
  \,\,\middle\vert\,\,  \Prob_{r\sim\bool^{t}}\mleft[U(M, y, r, 1^t) \textnormal{ outputs $x$}\mright]\geq \lambda\mright\}.
    \]
    We define $\rK^t(x)$ as $\rK^t(x\mid \epsilon)$.
\end{definition}

Note that in $\rK^t$, the string still has a deterministic, standard shorter description encoding a program; it is the program's execution that is randomized, so the natural decoding procedure uses randomness to recover the string. Randomization can provide significant savings in description length in  some concrete settings \citep{DBLP:conf/focs/0001LORS23}. Nevertheless, under derandomization assumptions from complexity theory, it follows that $\K^{\poly(t)}(x)$ and $\rK^{\poly(t)}(x)$ are related up to an additive logarithmic term for every large enough string $x$ (see, e.g.,~\citep[Appendix A.2]{GKLO22}).

In recent years, another variant of time-bounded Kolmogorov complexity has been particularly useful due to its coding properties \citep{DBLP:conf/icalp/LuOZ22} and connections to pseudorandomness.

\begin{definition}[$\mathsf{pK}^t$ \citep{GKLO22}]
    For $x,y\in\bool^*$, $\lambda\in [0,1]$, $t\in\mathbb{N}$, the \emph{probabilistic $t$-time-bounded Kolmogorov complexity} of $x$ conditioned on $y$ is defined as
   \[
        \pK^{t}_\lambda(x \mid y) \eqdef \min \left \{ k \,\,\,\middle\vert\,\,\, \Pr_{r \sim \{0,1\}^t} \mleft [ \exists\,M \in \{0,1\}^{\leq k}~\text{such that}~U(M,y,r,1^t)~\text{outputs}~x \mright ] \geq \lambda \right \}.
    \]
    We define $\pK^t(x)$ as $\pK^t(x\mid \epsilon)$.
\end{definition}

Equivalently, $\pK^{t}(x) \leq k$ if $\Pr_r[ \K^t(x \mid r) \leq k] \geq 2/3$. As a consequence, a string $x$ has bounded $\pK^t$ complexity if for most choices of $r$ it admits a small description given access to $r$. Note that here the description of $x$ depends on $r$. This captures the setting where both encoder and decoder share access to a random string $r$. In contrast, in $\rK^t$ there is a fixed randomised machine that produces $x$ with high probability over its internal randomness.\footnote{For the reader familiar with interactive proofs in the context of complexity theory, we mention that the distinction between $\rK^t$ and $\pK^t$ is analogous to the definitions of the complexity classes $\mathsf{MA}$ and $\mathsf{AM}$, respectively.}

We will also consider a randomized variant of Levin's $\Kt$ complexity that admits particularly useful coding properties (see \citep{DBLP:conf/icalp/HuLO26} and references therein for additional context).

\begin{definition}[$\mathsf{rKt}$ \cite{DBLP:conf/icalp/Oliveira19}]
    For $x,y\in\bool^*$, $\lambda\in [0,1]$, the \emph{randomized Levin complexity} of $x$ conditioned on $y$ is defined as
   \[
        \rKt_\lambda(x \mid y) \eqdef \min_{t \in \nat,\; M \in \bool^*}\mleft\{|M|  + \lceil \log t \rceil 
  \,\,\middle\vert\,\,  \Prob_{r\sim\bool^{t}}\mleft[U(M, y, r, 1^t) \textnormal{ outputs $x$}\mright]\geq \lambda\mright\}.
    \]
    We define $\rKt(x)$ as $\rKt(x\mid \epsilon)$.
\end{definition}

For completeness, we mention that one can also consider a ``probabilistic'' variant of $\Kt$ called $\pKt$ \citep{DBLP:conf/tcc/HiraharaLO24}, which is similar to the extension of $\K^t$ to $\pK^t$. While we will not employ $\pKt$ in the survey, we refer to \citep{pass2026valuepromptllmrelativekolmogorovcomplexity} for a recent application of a related notion in the context of machine learning.

A choice of universal machine $U$ gives rise to a certain \emph{universal distribution} that may be sampled by returning the output of $U$ on a uniformly random input. In the case of time-unbounded Kolmogorov complexity, for any string $x$, $\K(x)$ is approximately equal to the logarithm of the inverse of the probability of $x$ under this distribution.

Analogously, one may define a \emph{time-bounded universal distribution} for a given time bound $t \in \nat$. Namely, for a fixed universal machine $U$, \cite{DBLP:conf/focs/ImpagliazzoL90} define $\mathcal{Q}^t$ as the distribution of the output of $U(r, 1^t)$ for a uniformly random string $r \sim \bool^t$. For any auxiliary string $y \in \bool^*$, let $\mathcal{Q}^{t, y}$ be the distribution with the same definition as $\mathcal{Q}^t$, except the universal machine $U$ is additionally given access to $y$.

The following is the complexity measure corresponding to this family of distributions.

\begin{definition}[$\qq^t$ \cite{DBLP:conf/focs/HiraharaN23}]

    For $x,y\in\bool^*$ and $t\in\mathbb{N}$, the \emph{$\qq^t$-complexity} of $x$ conditioned on $y$ is defined as 
		\[
		\mathsf{q}^{t}(x\mid y) \eqdef \log\left ( \frac{1}{\Pr[\mathcal{Q}^{t, y} = x]}\right ).
		\]
        We define $\qq^{t}(x)$ as $\qq^{t}(x\mid \epsilon)$; that is, $\log(1/\Pr[\mathcal{Q}^t = x])$.
    
\end{definition}

We refer to the textbooks \citep{DBLP:series/txcs/LiV19,shen2017kolmogorov} for a comprehensive treatment of Kolmogorov complexity. For a modern perspective on time-bounded Kolmogorov complexity, see \citep{DBLP:journals/eatcs/LuO22}.

\subsubsection{Circuit Complexity}\label{sec:prelim_circuits}

A \emph{boolean circuit} is a finite directed acyclic graph with one or more source nodes, one sink node, and each internal node corresponding to some boolean function $g: \bool^* \to \bool$ from a fixed set $G$, called a \textit{basis}. A standard choice of $G$ is the set $\{\mathsf{AND}, \mathsf{OR}, \mathsf{NOT}\}$. Each node of a boolean circuit is called a \emph{gate}, with source nodes being called \emph{input gates} and the sink node being called the \emph{output gate}.\footnote{One may also consider circuits with multiple output gates, but we are mainly interested here in circuits representing boolean functions, which have only one bit of output.} A boolean circuit $C$ with $n$ input nodes induces a boolean function $f: \bool^n \to \bool$, computed by associating each bit of a given input string $x \in \bool^n$ with the corresponding input gate of $C$, performing the computation specified by the structure of $C$ and the functions of the internal gates, and taking $f(x) \in \bool$ to be the value resulting at the output gate. 

Since a single boolean circuit has a fixed number of input gates, it computes a
function on inputs of only one length. To capture a boolean function
$f : \bool^* \to \bool$ defined on all input lengths, we consider a \emph{circuit family}
$\mathcal{C} = \{C_n\}_{n \in \mathbb{N}}$, where each $C_n$ has $n$ input gates
and computes the restriction of $f$ to $\bool^n$. Two measures of the cost of a
circuit are its \emph{size}, i.e. the number of gates it contains, and its \emph{depth}, i.e. the
length of the longest path from an input gate to the output gate. 

We use the term \emph{fan-in} to refer to the number of inputs to a gate. For unrestricted boolean circuits, we adopt the convention that the fan-in of each gate is at most two. For bounded-depth circuit classes, we allow the fan-in to be unbounded unless stated otherwise. Note that ${\sf NOT}$ gates always have fan-in equal to 1.

From the perspective of compression, the truth table of a Boolean function $f$ can be seen as a long string of length $N = 2^n$. A small circuit $C$ that computes $f$ can then be interpreted as a compressed representation with a local decoder: given an index $i \in [N]$, $C(i)$ outputs the $i$-th bit of $f$. 

Next, we define several classes of boolean functions, which may all be specified according to a class of boolean circuit families. 

\begin{definition}[$\Ppoly$]
    $\Ppoly$ is the class of boolean functions computable by a circuit family $\{C_n\}$ over the basis $\{\mathsf{AND}, \mathsf{OR}, \mathsf{NOT}\}$, where for some constant $d \in \nat$, each $C_n$ has size at most $n^d$.
\end{definition}

\begin{definition}[$\AC^0$]
    $\AC^0$ is the class of boolean functions computable by a circuit family
    $\{C_n\}$ over the basis
    $\{\mathsf{AND}, \mathsf{OR}, \mathsf{NOT}\}$, where for some constants $c, d \in \nat$, each $C_n$ has size at most $n^c$ and depth at most $d$.
\end{definition}

\begin{definition}[$\TC^0$]
    The class $\TC^0$ has the same definition as $\AC^0$, except circuits are over the larger basis $\{\mathsf{AND}, \mathsf{OR}, \mathsf{NOT}, \mathsf{THR}\}$, where ${\sf THR}$ is a \emph{threshold} gate (also known as a weighted majority), outputting 1 iff $\sum_i w_i \cdot z_i \geq t$, where $w_i, t \in \mathbb{R}$.\footnote{Since each $\mathsf{AND}$, $\mathsf{OR}$, and $\mathsf{NOT}$ gate can be simulated by a $\mathsf{THR}$ gate with appropriate weights $w_i$ and threshold $t$, we typically assume that $\TC^0$ circuits contain $\mathsf{THR}$ gates only.}
\end{definition}

$\TC^0$ circuits, also known as threshold circuits, can be interpreted as a discrete analogue of neural networks. We recall that a current frontier in complexity theory is to establish super-polynomial circuit size lower bounds against depth-$2$ threshold circuits for an explicit function (see \citep{DBLP:conf/stoc/KaneW16, DBLP:conf/stoc/ChenTW26} for partial progress and related discussion). We refer to \Cref{fig:depthtwothreshold} for a diagram representation of this circuit class.

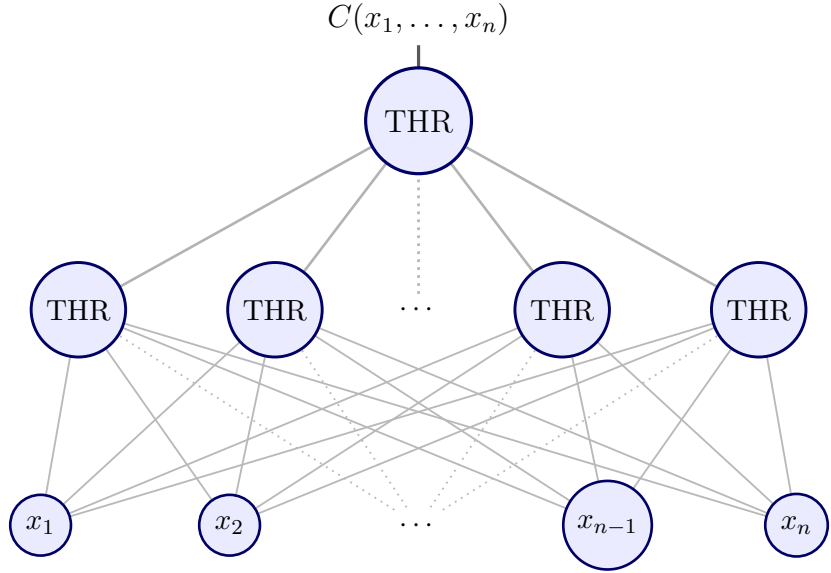
\begin{figure}[ht]
\centering
\begin{tikzpicture}[
    topgate/.style={
        circle,
        draw=blue!40!black,
        fill=blue!8,
        line width=1.2pt,
        minimum size=1.4cm,
        font=\large
    },
    bottomgate/.style={
        circle,
        draw=blue!40!black,
        fill=blue!8,
        line width=1.1pt,
        minimum size=1.25cm,
        font=\normalsize
    },
    inputnode/.style={
        circle,
        draw=blue!40!black,
        fill=blue!8,
        line width=1pt,
        minimum size=0.75cm,
        font=\normalsize
    },
    dotsnode/.style={
        fill=white,
        inner xsep=5pt,
        inner ysep=3pt
    }
]

\coordinate (out) at (6,6);
\coordinate (g1) at (1.5,3.5);
\coordinate (g2) at (4.1,3.5);
\coordinate (gdots) at (6,3.5);
\coordinate (g3) at (7.9,3.5);
\coordinate (g4) at (10.5,3.5);

\coordinate (i1) at (1.0,0.65);
\coordinate (i2) at (3.5,0.65);
\coordinate (idots) at (6,0.65);
\coordinate (i3) at (8.5,0.65);
\coordinate (i4) at (11.0,0.65);

\foreach \g in {g1,g2,g3,g4}
{
    \draw[
        gray!60,
        line width=1pt
    ]
    (\g) -- (out);
}

\draw[
    gray!60,
    line width=1pt,
    dotted
]
(gdots) -- (out);

\foreach \i in {i1,i2,i3,i4}
{
    \foreach \g in {g1,g2,g3,g4}
    {
        \draw[
            gray!55,
            line width=0.7pt
        ]
        (\i) -- (\g);
    }
}

\foreach \g in {g1,g2,g3,g4}
{
    \draw[
        gray!55,
        line width=0.7pt,
        dotted
    ]
    (idots) -- (\g);
}

\draw[
    gray!70!black,
    line width=1.2pt
]
(out) -- (6,7);

\node[above,font=\large] at (6,7)
{$C(x_1,\ldots,x_n)$};

\node[topgate] at (out)
{$\mathrm{THR}$};

\node[bottomgate] at (g1)
{$\mathrm{THR}$};

\node[bottomgate] at (g2)
{$\mathrm{THR}$};

\node[bottomgate] at (g3)
{$\mathrm{THR}$};

\node[bottomgate] at (g4)
{$\mathrm{THR}$};

\node[inputnode] at (i1)
{$x_1$};

\node[inputnode] at (i2)
{$x_2$};

\node[inputnode] at (i3)
{$x_{n-1}$};

\node[inputnode] at (i4)
{$x_n$};

\node[dotsnode] at (gdots)
{$\cdots$};

\node[dotsnode] at (idots)
{$\cdots$};

\end{tikzpicture}

\caption{\small{A depth-$2$ threshold circuit $C$.
The top threshold gate receives the outputs taking values in
$\{0,1\}$ from the threshold gates in the bottom layer,
each of which is connected to the input variables
$x_1,\ldots,x_n$.}
}
\label{fig:depthtwothreshold}
\end{figure}

\begin{definition}[$\AC^0$ with ${\sf MOD}_m$ gates]

    For a modulus $m \geq 2$, let $\mathsf{MOD}_m$ output $1$ iff the number of its inputs equal to $1$ is divisible by $m$. $\AC^0[m]$ has the same definition as $\AC^0$, except circuits are over the basis 
    \\$\{\mathsf{AND}, \mathsf{OR}, \mathsf{NOT}, \mathsf{MOD}_m\}$. The case of $m = 2$ gives the parity gate, sometimes written $\AC^0[\oplus]$. $\ACC^0$ is the class defined as $\bigcup_{m \geq 2}
    \AC^0[m]$.
    
\end{definition}

The following inclusion of circuit classes hold:
$$
\AC^0 \; \subseteq \;  \AC^0[\oplus] \;  \subseteq \;  \ACC^0 \;  \subseteq \;  \TC^0 \;  \subseteq \;  \P/\poly.
$$

For standard background on Boolean circuits, nonuniform computation, and circuit classes, see \citep{arora2009computational, jukna2012boolean}.

\subsubsection{Meta-Complexity}

Broadly speaking, problems in \emph{meta-complexity} ask about the complexity of a given object. One may restrict these problems to different kinds of object (for example, the input may simply be a string, or it may be understood to represent a partly or fully specified truth-table), and one may ask about different measures of complexity (for example, one may ask about the circuit complexity or the Kolmogorov complexity of the input).

We begin with the Minimum Circuit Size Problem \cite{DBLP:conf/stoc/KabanetsC00}, a prototypical example of a problem in meta-complexity.

\begin{definition}[$\MCSP$]

    The \emph{Minimum Circuit Size Problem}, denoted $\MCSP$, is the problem defined as follows. The input consists of a truth-table $T$, represented as a string, along with a size threshold $s \in \nat$. $(T, s)$ is a yes-instance of $\MCSP$ if and only if there exists a boolean circuit of size at most $s$ computing the function with truth-table $T$.
    
\end{definition}

\begin{definition}[$\partialMCSP$]

    $\partialMCSP$ is the problem defined as follows. For some $n, m \in \nat$, the input consists of a collection of input/output pairs $((x_i, b_i))_{i \in [m]}$, where $x_i \in \bool^n$ and $b_i \in \bool$ for each $i \in [m]$, and a threshold parameter $s \in \nat$. $((x_i, b_i)_{i \in [m]}, s)$ is a yes-instance of $\pMCSP$ iff there exists a boolean circuit $C$ of size at most $s$ such that $C(x_i) = b_i$ for all $i \in [m]$.

\end{definition}

Next, we will consider computational problems related to the task of estimating the time-bounded Kolmogorov complexity of an input string (see \citep{SimonsMetaComplexity2023} for further context and resources).

\begin{definition}[$\mathsf{MINKT}$]

    Inputs to the problem $\mathsf{MINKT}$ consist of a string $x$ along with a size threshold $s \in \nat$ and time parameter $t \in \nat$. $(x, 1^s, 1^t)$ is a yes-instance of $\mathsf{MINKT}$ if and only if $\K^t(x) \leq s$. 

\end{definition}

Observe that $\mathsf{MINKT}$ is in $\NP$. It is not in $\P$ under the existence of cryptographically secure one-way functions \citep{DBLP:conf/stoc/KabanetsC00}. Interestingly, while believed to be intractable, $\mathsf{MINKT}$ is not known to be $\NP$-complete (see \citep{HLO26} and references therein). 

\newcommand{\GapMINKT}{\mathsf{Gap}\text{-}\mathsf{MINKT}}

For the connections to learning and cryptography discussed in \Cref{sec:Imp_worlds}, the complexity of certain robust variants of $\mathsf{MINKT}$ is more relevant. We will informally refer to $\GapMINKT$ as the task of estimating the $\K^t$ complexity of an input string $x$ of length $n$ up to an additive term $O(\log n)$.\footnote{A bit more formally, in this promise problem we are interested in distinguishing $\K^t(x) \leq s$ from $\K^{\poly(t)}(x) \geq s + O(\log n)$.} We say that ``$\GapMINKT \in \P$'' if there is a polynomial-time algorithm that solves this problem on every input. We refer to \citep{DBLP:conf/stoc/Hirahara21} for the relevant formal definition. We say that  ``$\GapMINKT \in \mathsf{HeurP}$'' if there is a polynomial-time algorithm that is heuristically correct, i.e., it solves the problem on most inputs (but might output an incorrect estimate on a small fraction of inputs). We refer to \citep{DBLP:conf/crypto/LiuP21} for the relevant formal definition.

It is known that the task of estimating the $\K^t$ complexity of an input string is essentially equivalent (under polynomial-time reductions) to the following problem \citep{DBLP:conf/stoc/Hirahara21}.

\begin{definition}[The $\mathsf{GapK}^t\text{-}\mathsf{vs}\text{-}\mathsf{K}$ Problem]
       We say that $\mathsf{GapK}^t\text{-}\mathsf{vs}\text{-}\mathsf{K} \in \pr\BPP$ if there is a randomized polynomial-time algorithm $A$ such
that for every input $x \in \{0,1\}^*$ of length $n$ and every time bound $t \in \mathbb{N}$,
$$
\Pr_A \Big [\,\K(x) \leq A(x, 1^t) \leq \K^t(x) + O(\log n)\,\Big] \geq 2/3.
$$
\end{definition}

Similarly, we consider the complexity of computing Levin's $\Kt$ complexity. In the $\mathsf{MKtP}$ problem, we are given as input a tuple $(x,1^s)$, and the task is to decide if $\Kt(x) \leq s$. This problem is in the complexity class $\EXP$, i.e., deterministic exponential time. 

While the statement $\mathsf{MINKT} \notin \P$ immediately implies that $\NP \nsubseteq \P$, a proof that $\MKtP \notin \P$ merely implies that $\EXP \nsubseteq \P$, which is already known by the time hierarchy theorem. Still, it remains a longstanding open problem to establish this lower bound for $\MKtP$ \citep{DBLP:journals/siamcomp/AllenderBKMR06}. We will revisit this question in \Cref{sec:uncond_lbs}.

We refer to some general surveys of meta-complexity for broader context  \citep{allender2017complexity, Allender21_survey,DBLP:journals/eatcs/Hirahara22,10.1145/3744715.3744722}.

\subsubsection{Natural Properties}\label{sec:natural_prop}

In an influential work, \citep{DBLP:journals/jcss/RazborovR97} showed that many combinatorial circuit size lower bound arguments implicitly provide an efficient algorithm that distinguishes ``structured'' functions from ``random'' functions. As explained later in the survey, this provides a fruitful connection between complexity lower bounds and algorithmic upper bounds. 

Let $\mathcal{C}$ be a Boolean circuit class, and for a size parameter $s \colon \mathbb{N} \to \mathbb{N}$, let $\mathcal{C}[s]$ denote the set of Boolean functions $f \colon \{0,1\}^n \to \{0,1\}$ computed by a $\mathcal{C}$-circuit of size at most $s(n)$. We write $\mathcal{C}_n[s(n)]$ when we restrict the class to Boolean functions on $n$ input bits.

A function $f \colon \{0,1\}^n \to \{0,1\}$ can be represented by a string $\mathsf{tt}(f) \in \{0,1\}^N$ of length $N = 2^n$ that describes the truthtable of $f$. Similarly, given a string $w \in \{0,1\}^N$, we write $\mathsf{fnc}^w \colon \{0,1\}^n \to \{0,1\}$ to denote the corresponding Boolean function. 

In the discussion below, we will consider algorithms that receive as input the description of a Boolean function as an $N$-bit string. Roughly speaking, a natural property is an efficient algorithm that rejects every function in $\mathcal{C}[s]$ but accepts a large set of functions. In other words, it provides an efficient test to reject functions of low complexity while at the same time accepting a large collection of ``random'' functions.

\begin{definition}[Natural Properties \citep{DBLP:journals/jcss/RazborovR97}]\label{def:nat_prop}
    We say that there is a $\mathsf{prBPP}$\emph{-natural property} against a circuit class $\mathcal{C}[s]$ if there is a probabilistic polynomial-time algorithm $A$ operating over input strings of length $N = 2^n$ for which the following conditions hold:
    \begin{itemize}
        \item \emph{Constructive}: $A$ runs in time polynomial in its input length $N$.
        \item \emph{Largeness}: For every $n$ and $N = 2^n$, there is a set $\mathcal{F}_n \subseteq \{0,1\}^N$ with $\Pr_{w \sim \mathcal{U}_N}[w \in \mathcal{F}_n] \geq 1/2$ such that for every $w \in \mathcal{F}_n$, we have $\Pr_A[A(w) = 1] \geq 2/3$.
        \item $\mathcal{C}[s]$-\emph{usefulness}: For every $n$ and $f \colon \{0,1\}^n \to \{0,1\}$ such that $f \in \mathcal{C}_n[s(n)]$, we have $\Pr_A[A(\mathsf{tt}(f)) = 1] \leq 1/3$.
    \end{itemize}
\end{definition}

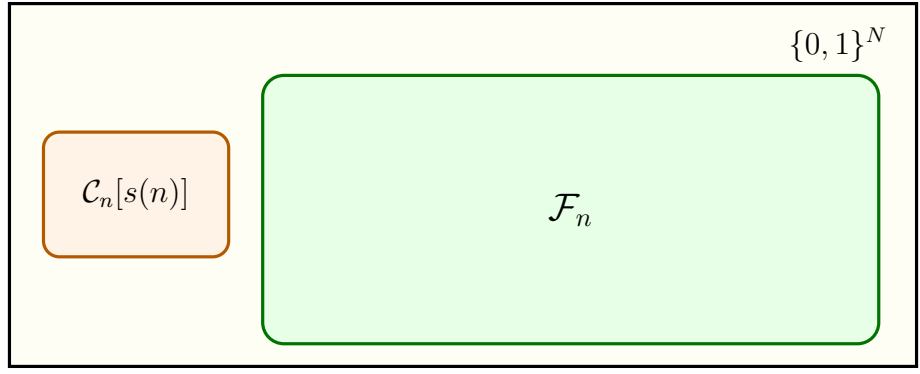
\begin{figure}[t]
\centering
\begin{tikzpicture}[font=\small]

\filldraw[
    fill=yellow!5,
    draw=black,
    line width=1.2pt
]
(0,0) rectangle (12,4.8);

\filldraw[
    fill=orange!10,
    draw=orange!70!black,
    line width=1.2pt,
    rounded corners=6pt
]
(0.45,1.45) rectangle (2.9,3.1);

\node[align=center,font=\large]
at (1.675,2.275)
{$\mathcal{C}_{n}[s(n)]$};

\filldraw[
    fill=green!10,
    draw=green!45!black,
    line width=1.2pt,
    rounded corners=8pt
]
(3.35,0.3) rectangle (11.5,3.85);

\node[align=center,font=\Large]
at (7.425,2.075)
{$\mathcal{F}_{n}$};

\node[anchor=north east,font=\large]
at (11.75,4.62)
{$\{0,1\}^{N}$};

\end{tikzpicture}

\caption{\small{A diagram representation of a natural property (\Cref{def:nat_prop}).
The outer rectangle is the ground set  of all
$N$-bit strings.
The small region represents the ``structured''
strings corresponding to functions in the low complexity class
$\mathcal{C}_{n}[s(n)]$.
The large region represents the ``random'' strings in
$\mathcal{F}_{n}\subseteq\{0,1\}^{N}$.
Its relative size illustrates the \emph{largeness} condition, while
$\mathcal{F}_{n}\cap\mathcal{C}_{n}[s(n)]=\varnothing$
illustrates \emph{usefulness} against $\mathcal{C}_{n}[s(n)]$.
A natural property is \emph{constructive}: it postulates the existence of an efficient
algorithm that distinguishes these two sets, thereby separating
the structured strings in $\mathcal{C}_{n}[s(n)]$ from the
random strings in $\mathcal{F}_{n}$.}
}
\label{fig:naturalproperty}
\end{figure}

In the definition above, we use $\mathsf{prBPP}$ instead of $\mathsf{BPP}$ because the randomized procedure $A$ might accept some strings with probability close to $1/2$. In the terminology of complexity theory, it might define a promise problem as opposed to a language. We refer to \Cref{fig:naturalproperty} for a pictorial description of a natural property.

For broader background on circuit lower bounds and the natural proofs barrier, see \citep{arora2009computational, jukna2012boolean}.

\subsubsection{Pseudorandomness}\label{sec:pseudorandomness}

A pseudorandom generator with \emph{seed length} $n$ and \emph{output length} $m>n$ is a deterministic map
$$
G \colon \bool^n \to \bool^m.
$$
For a uniform seed $s\sim\mathcal{U}_n$, we write $G(\mathcal{U}_n)$ for the distribution of $G(s)$. The quantity $m-n$ is called the \emph{stretch} of $G$. More generally, one considers families $\{G_n\}_{n\in\mathbb{N}}$ whose output length $m=m(n)$ grows with the seed length.

\begin{definition}[Distinguishers and Pseudorandomness]\label{def:prg}
A \emph{distinguisher} is a Boolean function, or more generally a randomized procedure, $D\colon\bool^m\to\bool$. Its distinguishing advantage against $G$ is
$$
\mathsf{Adv}_G(D)
\eqdef
\left|
\Pr_{s\sim\mathcal{U}_n}[D(G(s))=1]
-
\Pr_{y\sim\mathcal{U}_m}[D(y)=1]
\right|,
$$
where the probabilities also include the internal randomness of $D$, if any. We say that $D$ \emph{$\varepsilon$-distinguishes} the output of $G$ from the uniform distribution if $\mathsf{Adv}_G(D) > \varepsilon$. Given a class $\mathcal{C}$ of distinguishers, we say that $G$ \emph{fools $\mathcal{C}$ with error $\varepsilon$} if $\mathsf{Adv}_G(D)\leq\varepsilon$ for every $D\in\mathcal{C}$. Equivalently, $G$ is \emph{$\varepsilon$-pseudorandom} against $\mathcal{C}$.
\end{definition}

Pseudorandom generators commonly appear in two regimes. In the \emph{cryptographic regime}, $G$ is efficient and must fool distinguishers whose running time may be substantially larger than the time needed for one evaluation of $G$. The existence of cryptographic pseudorandom generators is equivalent, under polynomial-time reductions, to the existence of one-way functions and pseudorandom functions. These primitives are often used to establish conditional hardness results for learning.

In the \emph{complexity-theoretic regime}, the generator may use more time than the distinguishers it fools. Here a generator is often accompanied by a \emph{reconstruction procedure} that converts a distinguisher into useful computational structure. Such generators and reconstruction procedures will later be used to construct learning algorithms.

The survey will use several variants of these notions. We record the basic definitions here as a guide for the reader and refer to Vadhan's monograph \citep{DBLP:journals/fttcs/Vadhan12} for a comprehensive treatment.

\newpage

\section{Classical Connections Between Learning and Compression}\label{sec:classical_connections}

We start off with two classical connections between learning and compression. They allow us to link learnability, a statistical notion, with compressibility, a combinatorial one. In many contexts, this offers a particularly fruitful perspective for the design of learning algorithms.

\subsection{Model Compression}

We start with one of the most basic connections between learning and compression. In the (realizable) PAC learning setting, a labelled sample may contain many examples, but a learner can summarize the information they exhibit by a comparatively short hypothesis (``model'') that agrees with every observed label. We refer to this operation as \emph{model compression}.\footnote{This differs from the \emph{sample compression} schemes considered in the next subsection, where the learner retains a small part of the sample from which a hypothesis can later be reconstructed.}

Why should a compressed model generalize well, i.e., remain accurate on fresh examples that are not part of the training data? The answer is the principle commonly known as \emph{Occam's Razor}: among explanations that fit the observations, a simple explanation is less likely to fit merely by accident \cite{DBLP:journals/ipl/BlumerEHW87}. The informal justification is easy to understand: any fixed substantially incorrect hypothesis is unlikely to agree with many independent examples, and a short model description leaves few candidate hypotheses over which this unlikely event must be union bounded. Thus compression of the model controls generalization.

We now formalize this intuition. Let $\mathcal{C}$ be a concept class and $\mathcal{H}$ a hypothesis class, where $\mathcal{H}$ need not be equal to $\mathcal{C}$. For a distribution $D_n$ over $\bool^n$, define
$$
    \mathsf{err}_{D_n}(h,c) \eqdef \Pr_{x\sim D_n}[h(x)\neq c(x)].
$$

\begin{definition}[Consistent Hypothesis Finder]\label{def:consistent_hypothesis_finder}
    A \emph{consistent hypothesis finder} for $\mathcal{C}$ using $\mathcal{H}$ is a procedure $A$ with the following property. For every $n,m\geq 1$, every $c\in\mathcal{C}_n$, and every labelled sample
    $$
        S=\bigl \{(x_1,c(x_1)),\ldots,(x_m,c(x_m))\bigr \},
    $$
     $A(S)$ outputs a hypothesis $h\in\mathcal{H}_n$ such that $h(x_i)=c(x_i)$ for every $i\in[m]$.
\end{definition}

Thus a consistent hypothesis finder solves a search problem with no explicit reference to a distribution or to the PAC parameters $\varepsilon$ and $\delta$. Note that the hypothesis class $\mathcal{H}$ records the set of possible models, and it may differ from the target class $\mathcal{C}$. The next result shows that consistency implies generalization whenever the range of possible output models is sufficiently small and the number of samples is large enough.

\begin{theorem}[PAC Learning by Finding a Consistent Hypothesis \citep{DBLP:journals/ipl/BlumerEHW87}]\label{thm:consistent_hypothesis_finder}
    Let $\mathcal{C}$ be a concept class, $\mathcal{H}$ be a hypothesis class, and suppose that $A$ is a consistent hypothesis finder for $\mathcal{C}$ using $\mathcal{H}$. Fix $n\geq 1$ and assume that $\mathcal{H}_n$ is finite. Let $D_n$ be any distribution over $\bool^n$, let $c\in\mathcal{C}_n$, and let $\varepsilon,\delta\in(0,1)$. If $A$ is run on $m$ independent examples from $\EX(c,D_n)$, where
    $$
        m\geq \frac{1}{\varepsilon}\left(\ln|\mathcal{H}_n|+\ln\frac{1}{\delta}\right),
    $$
    then, with probability at least $1-\delta$ over the sample, its output hypothesis $h$ satisfies $\mathsf{err}_{D_n}(h,c)\leq\varepsilon$.
\end{theorem}

\begin{proof}
Let $c$ be the unknown target concept, and let $\mathcal{B}\subseteq\mathcal{H}_n$ be the set of \emph{bad} hypotheses $h$ satisfying $\mathsf{err}_{D_n}(h,c)>\varepsilon$. Fix $h\in\mathcal{B}$. A single example $(x,c(x))$ drawn from $\EX(c,D_n)$ is consistent with $h$ with probability $1-\mathsf{err}_{D_n}(h,c)<1-\varepsilon$. Since the $m$ examples are independent, the probability that $h$ is consistent with the entire sample is at most
$$
    (1-\varepsilon)^m\leq e^{-\varepsilon m}.
$$
If the consistent hypothesis finder outputs a bad hypothesis, then the sample must be consistent with at least one member of $\mathcal{B}$. Consequently, a union bound gives
$$
\begin{aligned}
    \Pr\bigl[\mathsf{err}_{D_n}(A(S),c)>\varepsilon\bigr]
    &\leq \Pr\bigl[\exists h\in\mathcal{B}\text{ consistent with }S\bigr] \\
    &\leq \sum_{h\in\mathcal{B}}\Pr[h\text{ is consistent with }S] \\
    &\leq |\mathcal{H}_n|\,e^{-\varepsilon m} \\
    &\leq \delta,
\end{aligned}
$$
where the final inequality follows from the assumed lower bound on $m$.
\end{proof}

The role of model compression can now be made explicit. Suppose every $h \in  \mathcal{H}_n$ can be encoded using $\ell(n)$ bits. Then $\ln|\mathcal{H}_n|\leq \ell(n) \ln 2$ and it is therefore enough to take
$$
    m\geq \frac{1}{\varepsilon}\left(\ell(n)\ln 2+\ln\frac{1}{\delta}\right).
$$
Thus the number of examples required for generalization is controlled by the description length of the output model. The training sample is described by  $m(n + 1)$ bits, while the explanation extracted from it uses only $\ell(n)$ bits. If the consistent hypothesis finder runs in polynomial time, $\ell(n)=\poly(n)$, and its output hypotheses are efficiently evaluatable under the encoding, then this construction gives an efficient PAC learner. 

To sum up, compressing the observations into a short model prevents overfitting because there are too few short models for a substantially incorrect one to agree with all the data, except with small probability. 

\paragraph{A Classical Example: Learning Through Polynomial Threshold Representations.}
An important example of the preceding model compression perspective is the class of \emph{linear threshold functions} (LTFs), also known as halfspaces or perceptrons. Such a function has the form
$
h(x)=\operatorname{sign}(w_0+\langle w,x\rangle).
$
Given a labelled sample realizable by a linear threshold function, consistency can be expressed as a linear program. After encoding the labels as elements of $\{\pm 1\}$, one seeks $w_0,w$ satisfying
$
y_i(w_0+\langle w,x_i\rangle)\geq 1
$
for every sample point $(x_i,y_i)$. Realizability guarantees feasibility after a suitable rescaling, and linear programming finds rational weights of polynomial description length in polynomial time. Thus linear programming provides an efficient consistent hypothesis finder. Moreover, there are only $2^{O(n^2)}$ distinct linear threshold functions on $\{0,1\}^n$, so the preceding \Cref{thm:consistent_hypothesis_finder} shows that the resulting hypothesis generalizes from polynomially many examples. This yields a distribution-independent PAC learner for linear threshold functions.

Feature expansion extends this observation to \emph{polynomial threshold functions} (PTFs), the natural extension of LTFs to higher degree. Let $\Phi_d(x)$ be the vector whose coordinates are all multilinear monomials $\prod_{i\in T}x_i$ with $|T|\leq d$. Its dimension is
$
N=\sum_{j=0}^{d}\binom{n}{j}\leq (en/d)^d,
$
and every degree-$d$ polynomial threshold function becomes a linear threshold function over the expanded representation $\Phi_d(x)$. Mapping each example through $\Phi_d$ and applying the preceding linear programming algorithm therefore gives a PAC learner with running time
$
\poly\bigl(n^d,1/\varepsilon,\log(1/\delta)\bigr).
$
Consequently, proving that a concept class has polynomial threshold degree at most $d$ immediately gives a learning algorithm with this running time. The learner need not know a representation of the target, since the linear program searches for any consistent polynomial threshold function.

This approach yields learning algorithms with provable guarantees for rich concept classes such as decision trees, DNF formulas, bounded-depth Boolean formulas, and more. We refer to the survey \citep{hellerstein2007pac} for a comprehensive discussion of this technique and its applications.\\

The Occam principle and closely related ideas appear in the literature under several names. It admits several generalizations. For instance, infinite hypothesis classes can be treated using the theory of VC dimension (e.g., LTFs over $\mathbb{R}^n$), while exact consistency can be relaxed by allowing the hypothesis finder to output an $h$ that agrees with most examples in the sample. For further background, we refer to standard textbooks such as \citep{KearnsVazirani:94, ShalevShwartzBenDavid2014}.

\subsection{Sample Compression}

The preceding subsection compressed a labelled sample into a short model. Sample compression takes a different approach: it retains only a few of the observed examples, together with a small amount of side information, and reconstructs a hypothesis from this summary. The reconstructed hypothesis need not itself have a short standalone description. What matters is that, relative to the original sample, it is determined by very few examples.

Fix an input length $n$. A \emph{sample compression scheme} for a concept class $\mathcal{C}_n$ with parameters $(k,\ell)$ consists of a compression map $\mathsf{Comp}$ and a reconstruction map $\mathsf{Rec}$ with the following guarantee. Given any labelled sample sequence
$$
    S=\bigl((x_1,c(x_1)),\ldots,(x_m,c(x_m))\bigr)
$$
consistent with some unknown $c\in\mathcal{C}_n$, the map $\mathsf{Comp}$ outputs a pair $(S|_I,\sigma)$, where $I\subseteq[m]$ has size at most $k$ and $\sigma\in\{0,1\}^\ell$. Here $S|_I$ denotes the subsequence of $S$ indexed by $I$, listed in increasing order of indices. The reconstruction map uses only this pair to produce a hypothesis
$$
    h=\mathsf{Rec}(S|_I,\sigma)
$$
that is consistent with the entire original sample $S$. As usual, strings shorter than $\ell$ can be padded, and the reconstructed hypothesis may be improper. We call $k+\ell$ the \emph{size} of the scheme. 

\paragraph{A Classical Example: Rectangles in the Plane.}
Although we stated the definition over the Boolean cube, sample compression applies without change to arbitrary (possibly infinite) instance domains. Consider the class of axis-aligned rectangles in $\mathbb{R}^2$, where an example is labelled positive precisely when it belongs to an unknown target rectangle $R^\star$. Given a realizable sample containing at least one positive example, $\mathsf{Comp}$ retains a leftmost, a rightmost, a lowest, and a highest positive example, breaking ties according to their order in the sample and retaining only distinct examples. The reconstruction map $\mathsf{Rec}$ outputs the smallest axis-aligned rectangle $R_S$ containing these points. (If there are no positive examples, the compression is empty and $\mathsf{Rec}$ returns the empty hypothesis.) Every positive example lies in $R_S$ by the choice of the four extreme points. On the other hand, $R_S\subseteq R^\star$, since it is the bounding rectangle of points belonging to the axis-aligned rectangle $R^\star$. Therefore every negative example lies outside $R_S$, and the reconstructed rectangle is consistent with the entire sample. This gives a sample compression scheme with parameters $(k = 4, \ell = 0)$. We refer to \Cref{fig:sample_compression_rectangles} for a pictorial description.\\

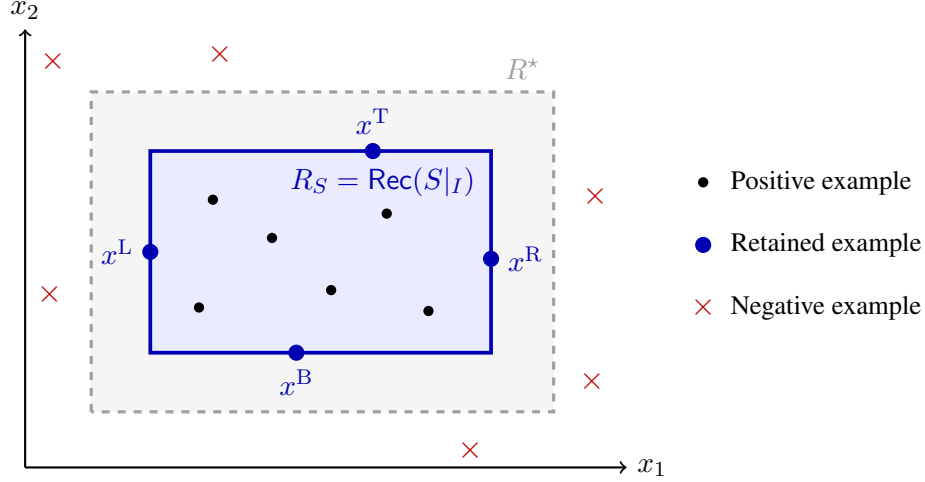
\begin{figure}[t]
\centering
\begin{tikzpicture}[scale=0.92]
    \draw[thick,->] (0,0) -- (8.65,0) node[right] {$x_1$};
    \draw[thick,->] (0,0) -- (0,6.30) node[above] {$x_2$};

    \fill[gray!8] (0.95,0.80) rectangle (7.60,5.40);
    \draw[gray!75,dashed,line width=1.2pt]
        (0.95,0.80) rectangle (7.60,5.40);
    \node[gray!75,anchor=south east] at (7.55,5.44)
        {$R^\star$};

    \fill[blue!8] (1.80,1.65) rectangle (6.70,4.55);
    \draw[blue!70!black,line width=1.4pt]
        (1.80,1.65) rectangle (6.70,4.55);
    \node[blue!70!black,anchor=north east] at (6.62,4.47)
        {$R_S=\mathsf{Rec}(S|_I)$};

    \foreach \p in {
        (2.50,2.30),
        (2.70,3.85),
        (3.55,3.30),
        (4.40,2.55),
        (5.20,3.65),
        (5.80,2.25)
    }{
        \fill[black] \p circle (2.1pt);
    }

    \fill[blue!70!black] (1.80,3.10) circle (3.3pt);
    \fill[blue!70!black] (6.70,3.00) circle (3.3pt);
    \fill[blue!70!black] (3.90,1.65) circle (3.3pt);
    \fill[blue!70!black] (5.00,4.55) circle (3.3pt);

    \node[blue!70!black,anchor=east] at (1.70,3.10)
        {$x^{\mathrm{L}}$};
    \node[blue!70!black,anchor=west] at (6.80,3.00)
        {$x^{\mathrm{R}}$};
    \node[blue!70!black,anchor=north] at (3.90,1.55)
        {$x^{\mathrm{B}}$};
    \node[blue!70!black,anchor=south] at (5.00,4.65)
        {$x^{\mathrm{T}}$};

    \foreach \p in {
        (0.35,2.50),
        (2.80,5.95),
        (8.20,3.90),
        (6.40,0.25),
        (8.15,1.25),
        (0.40,5.85)
    }{
        \node[red!75!black,font=\large] at \p {$\times$};
    }

    \fill[black] (9.75,4.10) circle (2.1pt);
    \node[anchor=west,font=\small] at (10.00,4.10)
        {Positive example};

    \fill[blue!70!black] (9.75,3.20) circle (3.3pt);
    \node[anchor=west,font=\small] at (10.00,3.20)
        {Retained example};

    \node[red!75!black,font=\large] at (9.75,2.30)
        {$\times$};
    \node[anchor=west,font=\small] at (10.00,2.30)
        {Negative example};
\end{tikzpicture}
\caption{\small{A sample compression scheme for axis-aligned rectangles in the plane.
The four extreme positive examples determine the reconstructed rectangle
$R_S$. Since $R_S\subseteq R^\star$, where $R^\star$ denotes the unknown target rectangle, $R_S$ contains every positive sample
point and no negative sample point. All labels in the original sample can be reconstructed from the four extreme positive points.}}
\label{fig:sample_compression_rectangles}
\end{figure}

Next, we explain why any such compression scheme yields a PAC learner.

\begin{theorem}[PAC Learning from Sample Compression \citep{littlestone1986relating}]\label{thm:sample_compression_pac}
    Suppose that $\mathcal{C}_n$ has a sample compression scheme with parameters $(k,\ell)$. Let $D_n$ be any distribution over $\bool^n$, let $c\in\mathcal{C}_n$, and let $\varepsilon,\delta\in(0,1)$. Draw $m$ independent examples from $\EX(c,D_n)$, apply the compression map, and output the reconstructed hypothesis $h$. Then
    $$
        \Pr\bigl[\mathsf{err}_{D_n}(h,c)>\varepsilon\bigr]
        \leq 2^\ell\sum_{i=0}^{k}\binom{m}{i}e^{-\varepsilon(m-i)}.
    $$
    In particular, it suffices to take
    $$
        m =O\!\left(\frac{k\log(2/\varepsilon)+\ell+\log(1/\delta)}{\varepsilon}\right)
    $$
    to obtain error at most $\varepsilon$ with probability at least $1-\delta$.
\end{theorem}

\begin{proof}[Proof Sketch.]
    The subtlety compared with the model compression argument is that the reconstructed hypothesis depends on the retained examples, which are themselves random. We therefore cannot fix the hypothesis before drawing the sample and charge it for agreeing with all $m$ examples.

    Instead, fix a set of indices $I\subseteq[m]$ of size at most $k$ and a message $\sigma\in\{0,1\}^\ell$, and condition on the examples indexed by $I$. The reconstructed hypothesis $h_{I,\sigma}=\mathsf{Rec}(S|_I,\sigma)$ is now fixed, while the remaining $m-|I|$ examples are still independent of it. If $\mathsf{err}_{D_n}(h_{I,\sigma},c)>\varepsilon$, the probability that it agrees with all these remaining examples is therefore at most
    $$
        (1-\varepsilon)^{m-|I|}\leq e^{-\varepsilon(m-|I|)}.
    $$
    Although the compression map may choose $(I,\sigma)$ adaptively after seeing the entire sample, if the scheme outputs a bad hypothesis then this event occurs for its selected pair. A union bound over the at most $2^\ell\binom{m}{i}$ choices with $|I|=i$, for every $0\leq i\leq k$, therefore gives the claimed probability bound. The standard estimate $\sum_{i=0}^k\binom{m}{i}\leq (em/k)^k$, together with the stated choice of $m$, makes this bound at most $\delta$.
\end{proof}

Thus a bounded sample compression scheme yields a PAC learner. If compression, reconstruction, and evaluation of the reconstructed hypothesis are efficient, then the resulting learner is efficient as well. Without these requirements, the conclusion is information-theoretic only, providing an upper bound on the sample complexity of learning $\mathcal{C}_n$.

Remarkably, the converse also holds at the qualitative level in the information-theoretic setting. A (possibly infinite) binary concept class is PAC learnable exactly when it has finite VC dimension, and Moran and Yehudayoff \citep{moran2016sample} proved that every class of VC dimension $d$ has a sample compression scheme of size $2^{O(d)}$. Consequently, for binary concept classes, the existence of a bounded sample compression scheme characterizes PAC learnability, even though the known quantitative relationship is far from tight.

\begin{problem}[Sample Compression Conjecture]\label{prob:sample_compression}
    Does every binary concept class of VC dimension $d$ admit a sample compression scheme of size $O(d)$?
\end{problem}

Closing the existing exponential gap would show that the amount of information one needs to retain from a sample is governed, up to constant factors, by the same parameter that controls PAC sample complexity.

\newpage

\section{Complexity-Theoretic Foundations of Learning and Compression}\label{sec:complexity_learning}

\subsection{Learning Paradises}\label{sec:paradises}

We consider five possible learning scenarios, ordered from a strong unsupervised setting to a far weaker supervised setting in which only weak heuristic learnability is assumed. Still, even the weakest of these assumptions asks for learning power over \emph{unrestricted} Boolean circuits, a capability well beyond current theoretical results and practical algorithms (see \Cref{sec:Imp_worlds} for more details). We therefore call each scenario a learning paradise.

\begin{itemize}
    \item \textbf{P.I: Efficient Extrapolation.}\\
    There are polynomials $p_1(\cdot), p_2(\cdot)$ and a randomized algorithm $L$ running in time $p_1(n,m,s,t,1/\varepsilon,1/\delta)$ that solves the $(m,n,s,t,\varepsilon,\delta)$-extrapolation task whenever $m \geq p_2(s,1/\varepsilon,1/\delta)$. 
    \item \textbf{P.II: Efficient Distribution Learning.}\\
    There is a polynomial $p(\cdot)$ and  a randomized algorithm $A$ running in time $p(n,s,t,1/\varepsilon,1/\delta)$ that $\varepsilon$-learns every distribution $D \in \mathcal{D}^{s,t}_n$ with probability at least $1 - \delta$. 
    \item \textbf{P.III: Efficient Agnostic Learning (Under Feasible Example Distributions).} There is a polynomial $p(\cdot)$ and a randomized algorithm $A$ running in time $p(n, 1/\varepsilon, \log(1/\delta), s, t, a)$ that agnostically learns $\SIZE[s]$ over $\mathcal{D}^{a, t}$.
    \item \textbf{P.IV: Efficient PAC Learning (with MQs and Under Feasible Example Distributions).}
    There is a polynomial $p(\cdot)$ and a randomized algorithm $A$ running in time $p(n, 1/\varepsilon, \log(1/\delta), s, t, a)$ that PAC learns $\SIZE[s]$ over $\mathcal{D}^{a, t}$ with membership queries.
    \item \textbf{P.V: Efficient Heuristic Learning (without MQs).}
    For any polynomials $s(\cdot), t_1(\cdot), t_2(\cdot)$ and distribution families $P \in \SAMP[t_1]$ and $D \in \SAMP[t_2]$, there is a polynomial $p(\cdot)$ and a randomized algorithm $A$ running in time $p(n, s, t_1, t_2)$ that weakly predicts $\SIZE[s]$ on average over $P$ and $D$.

\end{itemize}

There are many natural variants of assumption P.IV. For example, one can vary whether the learner is allowed to make membership queries, whether the examples come from the uniform distribution or from arbitrary distributions, and whether the learner is proper. To keep the presentation simple and to avoid the overly pessimistic distribution-free model, we adopt the membership query PAC model over feasible distributions as a representative supervised learning scenario. 

Learning under feasible distributions is a natural restriction in complexity theoretic and cryptographic applications. There are two distributions to keep track of. In PAC Learning (P.IV), the feasible distribution is over input examples $x$, while the target concept $f$ is an arbitrary function from the class. In Heuristic Learning (P.V), there is also a feasible distribution over the concepts themselves. Cryptographic reductions often have both kinds of efficient sampling: the target concept is generated by an efficient setup or key generation experiment, while the inputs are generated by an efficient challenge or encryption experiment. In particular, attacks on public-key encryption can often be phrased as learning a polynomial-size predicate, with the concept sampled by key generation and the examples sampled by encryption. We refer to \Cref{sec:related_work} for additional discussion.

\subsection{Basic Relations}

We now explain the basic ordering among the five learning paradises:
$$
\text{P.I} \;\; \implies \;\; \text{P.II}  \;\; \implies  \;\; \text{P.III}  \;\; \implies \;\; \text{P.IV}  \;\; \implies \;\; \text{P.V}.
$$
Throughout the proof sketches below, parameters are rescaled by fixed polynomials. We suppress these routine changes in notation.

\begin{lemma}[P.I $\implies$ P.II]\label{lem:worldI-worldII}

    If Efficient Extrapolation is possible, then Efficient Distribution Learning is possible.
    
\end{lemma}
\begin{proof}[Proof Sketch.]

    Fix a distribution $D \in \mathcal{D}^{s,t}_n$, and let $M$ be a randomized machine of description length at most $s$ and running time at most $t$ that samples $D$. We reduce learning $D$ from samples to extrapolating a sequence whose coordinates are independent samples from $D$.

    Choose a sequence length $m$ large enough for the extrapolation guarantee, with accuracy and confidence parameters $\eta$ and $\gamma$ to be fixed below. Let $M^{\otimes m}$ be the randomized machine that runs $M$ independently $m$ times and outputs the resulting sequence $(X_1,\ldots,X_m)$. Its description length is $s' = s + O(\log m)$ and its running time is $t' = O(m(t+n))$. Since $m$ only needs to dominate a polynomial in the relevant parameters, we may choose it polynomially large so that the assumption P.I applies to $(m,n,s',t',\eta,\gamma)$.

    For this product source, the conditional distribution of the next block is always the original distribution:
    $$(X_i \mid X_{<i}=x_{<i}) \equiv D$$
    for every $i$ and every prefix $x_{<i}$ in the support. Given sample access to $D$, the distribution learner draws samples $x_1,\ldots,x_m$, chooses $I \sim [m]$, and forms the prefix $y=x_{<I}$. It outputs the sampler $C_y$ that has $y$ hardwired and, on fresh random coins, runs the extrapolation algorithm $L$ on input $y$.

    The distribution of $(I,y)$ in this procedure is exactly the distribution of $(i,X_{<i})$ in the extrapolation experiment for $M^{\otimes m}$. Therefore, by assumption P.I,
    $$\Pr_{I,y}\left[\Delta_{\mathsf{TV}}(C_y,D) \leq \eta\right] \geq 1-\gamma.$$
    Taking $\eta=\varepsilon$ and $\gamma=\delta$ gives a distribution learner for $D$. The returned sampler is efficient because it only simulates the polynomial-time extrapolator with a polynomial-length prefix hardwired.
\end{proof}

\begin{lemma}[P.II $\implies$ P.III; \cite{DBLP:conf/colt/Xiao10}]\label{lem:worldII-worldIII}

    If Efficient Distribution Learning is possible, then Efficient Agnostic Learning is possible.
    
\end{lemma}
\begin{proof}[Proof Sketch.]

    Fix an agnostic learning instance consisting of a joint example distribution $E_n \in \mathcal{D}^{a,t}_{n+1}$. Since $E_n$ is feasibly samplable, it belongs to the distribution class handled by  P.II, with polynomially related parameters.

    Simulating the distribution learner on samples obtained from $\EX(E_n)$ gives, with high probability, an explicit sampler $S=(S^{(0)},S^{(1)})$ for a distribution $\widehat{E}$ satisfying $\Delta_{\mathsf{TV}}(\widehat{E},E_n) \leq \eta$. If we could sample from the conditional distribution of $S^{(1)}$ given $S^{(0)}=x$, then we would obtain labels that are close to the true labels according to $E_n$. The only obstacle is that $S$ is an explicit (i.e.,~white-box) sampler for pairs, not an explicit conditional sampler.

    Interestingly, distribution learning also supplies the ingredient needed to remove this obstacle. By \Cref{lem:dist-break-prf}, P.II rules out AIPRFs; by \Cref{lem:ai-crypto}, this rules out AIOWFs; and by \Cref{lem:rand-h}, the nonexistence of AIOWFs allows one to produce a randomised hypothesis $h$ such that $(x,h(x))$, with $x$ drawn from the marginal $S^{(0)}$, is close to $\widehat{E}$.

    By the triangle inequality, the joint distribution $(x,h(x))$ is close to $E_n$. Then, as in \Cref{lem:aiowf-all-wb}, by sampling many choices of internal randomness for $h$ and taking the majority hypothesis, one obtains, with high probability, a deterministic hypothesis $H$ such that 
    \[\Pr_{(x, b) \sim E_n}[H(x) = b] \geq \opt_n(\mathcal{F}_{all}, E_n) - \varepsilon' \geq \opt_n(\mathcal{C}, E_n) - \varepsilon', \]
    where $\mathcal{F}_{all}$ denotes the set of all boolean functions $f: \bool^n \to \bool$, $\mathcal{C} \subseteq \mathcal{F}_{all}$ is any desired concept class, and $\varepsilon'$ depends on the closeness of $E_n$ and $(x, h(x))$. Choosing $\eta$ and all failure probabilities as small enough polynomial fractions of $\varepsilon$ and $\delta$ yields the desired agnostic learning guarantee. More details about this approach can be found in \Cref{sec:wb}.
    
\end{proof}

\begin{lemma}[P.III $\implies$ P.IV]\label{lem:worldIII-worldIV}

    If Efficient Agnostic Learning is possible, then Efficient PAC Learning is possible.
    
\end{lemma}
\begin{proof}

    This implication is immediate from the definitions. More specifically, observe that PAC learning a target function $f \in \SIZE[s]$ over a distribution $D_n \in \mathcal{D}_n^{a,t}$ is equivalent to agnostically learning $\SIZE[s]$ over the distribution $E_n \in \mathcal{D}_{n + 1}^{\poly(a,s), \poly(t,s)}$ that samples $x \sim D_n$ and outputs $(x, f(x))$. Since 
    \[\opt_n(\SIZE[s], E) = 1,\]
    the definition of an agnostic learner $A$ guarantees that, with probability at least $1 - \delta$, $A(1^n, 1^{\varepsilon^{-1}}, \delta^{-1})$ with oracle access to $\EX(f, D_n) = \EX(E_n)$ outputs $h$  such that $\Pr_{x \sim D_n}[h(x) = f(x)] \geq 1 - \varepsilon$. This satisfies the requirement of a PAC learner.

\end{proof}

\begin{lemma}[P.IV $\implies$ P.V]\label{lem:worldIV-worldV}

    If Efficient PAC Learning is possible, then Efficient Heuristic Learning is possible.
    
\end{lemma}
\begin{proof}[Proof Sketch.] Note that, under our definitions, the PAC learning algorithm (P.IV) is allowed to make membership queries, while the heuristic learner (P.V) is not. This means in particular that we cannot simply simulate the PAC learner when designing the heuristic learner from random examples.

    The first step is to show that P.IV rules out auxiliary-input pseudorandom functions. Given a candidate AIPRF $F_z$, fix an auxiliary string $z$ and a random key $u$, and view $g(x)=F_z(u,x)$ as a polynomial size circuit. The PAC learner can be run under the uniform distribution, with both its random examples and its membership queries answered by the oracle being tested. If the oracle is $F_z(u,\cdot)$, the learner outputs a hypothesis with high agreement with the oracle on a fresh uniform input. If the oracle is a truly random function, then after only polynomially many queries its value on a fresh input is independent of the learner's hypothesis except with negligible collision probability. Testing agreement on fresh inputs therefore distinguishes the AIPRF from a random function, contradicting security. (We refer the reader to \Cref{lem:heur-break-prf} for a more detailed argument.)

    In turn, by \Cref{lem:ai-crypto}, the failure of AIPRFs gives the failure of AIOWFs, and in particular rules out infinitely-often one-way functions and pseudorandom generators. We can then apply the framework  of Blum, Furst, Kearns, and Lipton \cite{DBLP:conf/crypto/BlumFKL93}, discussed in \Cref{sec:learn-via-dist}. Their argument starts from the fact that every candidate generator can be distinguished at every large enough input length, applies a hybrid argument to obtain a next-bit predictor for the direct-product generator associated with a random target function, and interprets that predictor as a learner from random examples. We refer to \Cref{sec:learn-via-dist} for more details.

    Consequently, for every feasible distribution $P$ over polynomial-size target circuits and every feasible example distribution $D$, there is an efficient learner that predicts $f(y)$ with advantage $1/\poly(n)$, where $f \sim P$, the learner receives examples from $\EX(f,D)$, and $y \sim D$. This is exactly weak heuristic learning as in P.V.
\end{proof}

\subsection{Impagliazzo's  Worlds and Meta-Complexity}\label{sec:Imp_worlds}

In an influential paper, Impagliazzo \citep{DBLP:conf/coco/Impagliazzo95} introduced five potential computational worlds connected to the complexity of $\NP$ problems and the possibility of secure cryptography. The point was not only to ask whether $\P=\NP$, but to distinguish several possible realities according to whether hard problems remain hard on typical inputs and whether such hardness can support secure cryptography. The two relevant worlds in the context of our discussion are Heuristica and Pessiland. 

In Heuristica, $\NP$ is hard in the worst case, but this hardness disappears on average: every $\NP$ problem has an efficient algorithm that is correct on almost all inputs. In the notation of the diagram in \Cref{fig:learning-paradises-impagliazzo}, this corresponds to a world where $\P \neq \NP$ but $\Dist\NP \subseteq \Avg\P$. 

In Pessiland, hardness exists on average, but the hard problems are not structured enough to yield secure cryptography. The relevant feature of Pessiland is the failure of one-way functions, equivalently the existence of efficient inverters for every candidate one-way function.

We refer to \citep{DBLP:journals/fttcs/BogdanovT06} for a detailed discussion of Heuristica, Pessiland, and their connections to cryptography and average-case complexity theory.

The learning paradises from \Cref{sec:paradises} sit between these two worlds, as summarized in \Cref{fig:learning-paradises-impagliazzo}. The left side of the diagram starts from the average-case easiness of $\NP$ promised by Heuristica and passes to meta-complexity via ``$\GapMINKT \in \P$'', i.e., the worst-case easiness of estimating time-bounded Kolmogorov complexity.\footnote{Note that this is a nontrivial result: from the assumption that every problem in $\NP$ is easy on average, we derive that a specific $\NP$ problem from meta-complexity is easy in the worst case. This worst-case easiness seems particularly useful when applying the assumption to construct learning algorithms that must work on every input concept.} This leads to extrapolation and then, by the reductions in \Cref{sec:paradises}, to progressively weaker learning paradises. The right side records that the weakest of these paradises is already tied to the nonexistence of one-way functions and to average-case algorithms for estimating time-bounded Kolmogorov complexity. Thus the feasibility of different learning tasks is closely tied to meta-complexity, i.e., the  complexity of estimating time-bounded Kolmogorov complexity. 

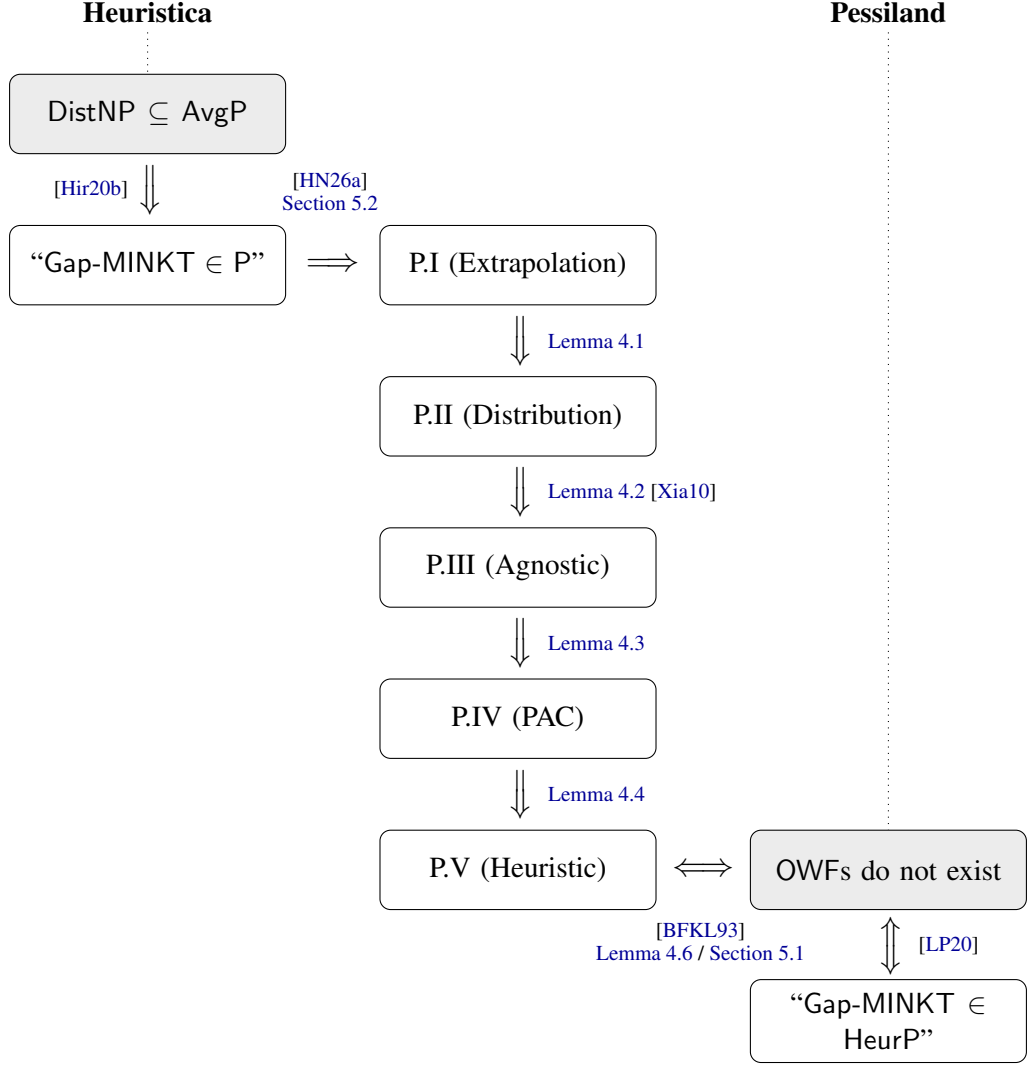
\begin{figure}[!ht]
\centering
\begin{tikzpicture}[
    every node/.style={align=center},
    statement/.style={
        draw,
        rounded corners,
        inner sep=5pt,
        text width=3.3cm,
        minimum height=1.05cm
    },
    worlddef/.style={
        statement,
        fill=gray!15
    },
    heading/.style={font=\bfseries},
    arrow/.style={font=\large},
    arrowlabel/.style={
        font=\scriptsize,
        fill=white,
        inner sep=1.5pt
    }
]
\newcommand{\longimp}{\scalebox{1.05}{$\Longrightarrow$}}
\newcommand{\longiff}{\scalebox{1.05}{$\Longleftrightarrow$}}
\newcommand{\downimp}{\rotatebox[origin=c]{270}{\longimp}}
\newcommand{\downiff}{\rotatebox[origin=c]{270}{\longiff}}
\node[heading] (heuristica) at (0,0) {Heuristica};
\node[worlddef] (distnp) at (0,-1.35) {$\Dist\NP \subseteq \Avg\P$};
\node[statement] (gapP) at (0,-3.35) {``$\GapMINKT \in \P$''};
\node[statement] (p1) at (4.9,-3.35) {P.I (Extrapolation)};
\node[statement] (p2) at (4.9,-5.35) {P.II (Distribution)};
\node[statement] (p3) at (4.9,-7.35) {P.III (Agnostic)};
\node[statement] (p4) at (4.9,-9.35) {P.IV (PAC)};
\node[statement] (p5) at (4.9,-11.35) {P.V (Heuristic)};
\node[heading] (pessiland) at (9.8,0) {Pessiland};
\node[worlddef] (nowf) at (9.8,-11.35) {$\mathsf{OWF}$s do not exist};
\node[statement] (gapHeurP) at (9.8,-13.35) {``$\GapMINKT \in \Heur\P$''};
\draw[dotted] (heuristica) -- (distnp);
\draw[dotted] (pessiland) -- (nowf);
\node[arrow] at (0,-2.35) {\downimp};
\node[arrowlabel, anchor=west] at (-1.3,-2.35) {\cite{DBLP:conf/coco/Hirahara20}};
\node[arrow] at (2.45,-3.35) {\longimp};
\node[arrowlabel] at (2.45,-2.40) {
    \begin{tabular}{c}
        \cite{DBLP:journals/eccc/HN25}\\[-0.1em]
        \Cref{sec:solomonoff}
    \end{tabular}
};
\node[arrow] at (4.9,-4.35) {\downimp};
\node[arrowlabel, anchor=west] at (5.24,-4.35) {\Cref{lem:worldI-worldII}};
\node[arrow] at (4.9,-6.35) {\downimp};
\node[arrowlabel, anchor=west] at (5.24,-6.35) {\Cref{lem:worldII-worldIII} \cite{DBLP:conf/colt/Xiao10}};
\node[arrow] at (4.9,-8.35) {\downimp};
\node[arrowlabel, anchor=west] at (5.24,-8.35) {\Cref{lem:worldIII-worldIV}};
\node[arrow] at (4.9,-10.35) {\downimp};
\node[arrowlabel, anchor=west] at (5.24,-10.35) {\Cref{lem:worldIV-worldV}};
\node[arrow] at (7.35,-11.35) {\longiff};
\node[arrowlabel] at (7.35,-12.32) {
    \begin{tabular}{c}
        \cite{DBLP:conf/crypto/BlumFKL93}\\[-0.1em]
        \Cref{lem:heur-break-prf} / \Cref{sec:learn-via-dist}
    \end{tabular}
};
\node[arrow] at (9.8,-12.35) {\downiff};
\node[arrowlabel, anchor=west] at (10.14,-12.35) {\cite{DBLP:conf/focs/LiuP20}};
\end{tikzpicture}
\caption{\small{Learning paradises P.I -- P.V (\Cref{sec:paradises}) between Impagliazzo's worlds \citep{DBLP:conf/coco/Impagliazzo95} Heuristica and Pessiland, including their relationship to compression and meta-complexity via the computational task of estimating time-bounded Kolmogorov complexity.}}
\label{fig:learning-paradises-impagliazzo}
\end{figure}

Another consequence is that secure cryptography excludes all five learning paradises. Indeed, P.V rules out one-way functions by \Cref{lem:heur-break-prf}. Consequently, in a cryptographic world, unrestricted Boolean circuits cannot be provably learned efficiently even with the weak guarantees of P.V: any such theorem would also refute standard cryptographic assumptions.\footnote{This is not in apparent contradiction to the astonishing success of machine learning algorithms in practice. Recall that the learning guarantees here concern unrestricted Boolean circuits. Learning algorithms deployed in practice are not known to satisfy such guarantees, and they are not known to break standard cryptographic assumptions.}

A precise characterization of the learning power associated with both Heuristica and Pessiland is possible, although in a different learning setting \citep{HN26}. It would be very interesting to obtain an equally precise understanding of how the learning paradises P.I through P.V relate to standard complexity and cryptographic statements.

While the preceding discussion has focused on the relationship between learning and compression, and on its connections to cryptography and average-case complexity, the broader role of compression problems in complexity theory remains far from fully understood. A central open question is whether efficient algorithms for estimating time-bounded Kolmogorov complexity would yield
efficient algorithms for $\NP$-complete problems. This is open even in an average-case sense: we do not know whether $\GapMINKT \in \P$ implies that $\NP$ is easy on average. Informally, such a result would say that the ability to perform compression, or to recognize strings admitting short efficient descriptions, is powerful enough to solve broad classes of computational search and optimization tasks. From the perspective of learning, this would elevate the role of compression beyond efficient prediction: it would show in a certain sense that compression also underlies efficient planning and decision-making. For recent papers on the computational complexity of  compression, see \citep{DBLP:conf/focs/Hirahara22, HLO26} and references therein.

\begin{problem}[Computational Complexity of Compression]\label{prob:complexity_compression}
 What is the hardness of the task of computing/estimating the time-bounded Kolmogorov complexity of an input string? Formally, the following central problems from meta-complexity remain open:
 \begin{itemize}
     \item Does $\GapMINKT \in \P$ imply that $\Dist\NP \subseteq \Avg\P$? In other words, does the easiness of estimating complexity imply that $\NP$ is easy on average?
     \item Does $\mathsf{MINKT} \in \P$ imply that $\NP \subseteq \P$? In other words, does the easiness of (exactly) computing complexity imply that $\NP$ is easy in the worst case?
 \end{itemize}
\end{problem}

\subsection{Auxiliary Results}

We collect here some lemmas needed in results mentioned in the preceding sections.

\begin{lemma}\label{lem:heur-break-prf}

    If Efficient Heuristic Learning is possible, then ioPRFs do not exist. It follows that (infinitely-often) one-way functions and pseudorandom generators do not exist.
    
\end{lemma}
\begin{proof}

    Let $f = \{f: \bool^{n} \times \bool^{n} \to \bool\}_{n \in \nat}$ be an ioPRF, and let $n \in \nat$ be a sufficiently large input length at which it is secure. Define a distribution $P_n$ over target functions as follows: sample $u \sim \bool^n$ uniformly at random and output a description of the circuit $f_u \colon \bool^n \to \bool$ defined by $f_u(x) = f(u, x)$ (with $u$ hardcoded). Since $u$ is sampled uniformly and $f$ is polynomial-time computable, $P_n \in \P\SAMP$. Let $D_n = \mathcal{U}_n \in \P\SAMP$.

    By assumption, Boolean circuits are heuristically learnable over $\P\SAMP$. Hence there exists an efficient 
    machine $A$ and a polynomial $q'$ such that for all sufficiently large $n$,
    \begin{align}\label{eq:heur-setting}
        \Pr_{\substack{u \sim \mathcal{U}_n;\\ 
        \EX(f_u, \mathcal{U}_n);\\ 
        x' \sim \mathcal{U}_n}}
        \!\left[A^{\EX(f_u,\,\mathcal{U}_n)}(x') = f_u(x')\right] 
        \;\geq\; \frac{1}{2} + \frac{1}{q'(n)}.
    \end{align}
    We now construct an efficient distinguisher $D$ for $f$ as follows.
    \begin{quote}
        Given oracle access to some function $\mathcal{O}$, which is either $f(u,\cdot)$ for a random $u \sim \bool^n$ or a truly random function $R_n \colon \bool^n \to \bool$, $D$ draws $x' \sim \mathcal{U}_n$, queries $\mathcal{O}$ to obtain $b$, and outputs $1$ if and only if $A^{\EX(\mathcal{O}, \mathcal{U}_n)}(x') = b$, where the oracle $\EX(\mathcal{O}, \mathcal{U}_n)$ may be simulated by sampling $x \sim \mathcal{U}_n$ and querying $\mathcal{O}(x)$ to get the label.
    \end{quote}
    When the oracle $\mathcal{O}$ is $f(u, \cdot)$ for $u \sim \bool^n$, the learner $A$ is run in a setting exactly matching the setting of \Cref{eq:heur-setting}, so $D$ outputs $1$ with probability at least 
    $(1/2) + (1/q'(n))$.

    When the oracle $\mathcal{O}$ is $R_n$, $b = R_n(x')$ is uniformly random and independent of the output of $A$ on any input not previously inspected by $A$. Since $A$ inspects at most polynomially many inputs, $D$ outputs $1$ with probability exponentially close to $1/2$.

    Therefore,
    \[
        \left|\Pr_{u \sim \bool^n}\!\left[D^{f(u,\cdot)}(1^n) = 1\right] 
        - \Pr_{R_n}\!\left[D^{R_n(\cdot)}(1^n) = 1\right]\right| 
        \;=\; \Omega\! \left( \frac{1}{q'(n)} \right ),
    \]
    contradicting the assumption that $f$ is an ioPRF. Thus, ioPRFs do not exist. This implies that standard PRFs likewise do not exist. The non-existence of both infinitely-often and standard OWFs and PRGs follows by well-known reductions among these primitives \cite{DBLP:journals/siamcomp/HastadILL99, DBLP:journals/jacm/GoldreichGM86}.
\end{proof}

Using a very similar argument, we can rule out \emph{auxiliary input} PRFs assuming easiness of PAC learning, even when membership queries are allowed.
\begin{lemma}\label{lem:pac-break-prf}

    If Efficient PAC Learning is possible, then AIPRFs do not exist.
        
\end{lemma}
\begin{proof}[Proof Sketch.]

    The proof is very similar to that of \Cref{lem:heur-break-prf}. Let $f_z: \bool^{n(|z|)} \times \bool^{n(|z|)} \to \bool$ be a purported AIPRF. The key idea is that, for \emph{every} choice of $z \in \bool^{*}$ and $u \in \bool^{n(|z|)}$, $f_z(u, -) \in \SIZE[p]$ for some fixed polynomial $p$. We are guaranteed that the learner $A$ will succeed (with high probability) on every such target function. Another difference is that $A$ is allowed to make membership queries to the target function, which the distinguisher $D$ may simulate by querying $\mathcal{O}$. 
\end{proof}

Finally, the same consequence holds assuming easiness of distributional learning. This lemma is a key ingredient in the proof that distribution learning is at least as strong as agnostic learning.

\begin{lemma}[\cite{kearns1994learnability, DBLP:conf/colt/Xiao10}]\label{lem:dist-break-prf}
    If Efficient Distribution Learning is possible, then AIPRFs do not exist. 
\end{lemma}
\begin{proof}[Proof Sketch.]
    Let $f=\{f_z\}$ be a purported AIPRF, and let $A$ be the distribution learner promised by P.II. Fix $z \in \bool^*$, write $n=n(|z|)$, and for an oracle $\mathcal{O}\colon \bool^n \to \bool$ define the graph distribution
    $$
        G_{\mathcal{O}} \equiv (X,\mathcal{O}(X)) \quad \text{where } X \sim \mathcal{U}_n .
    $$
    If $\mathcal{O}=f_z(u,\cdot)$, then $G_{\mathcal{O}}$ is sampled by the efficient sampler that has $(z,u)$ hardwired, chooses $X \sim \mathcal{U}_n$, and outputs $(X,f_z(u,X))$. Thus, for every $z$ and $u$, P.II applies to $G_{f_z(u,\cdot)}$ with polynomially related size and time parameters.

    We construct an oracle distinguisher $D$. Given $z$ and oracle access to $\mathcal{O}$, it simulates sample access to $G_{\mathcal{O}}$ by drawing $X \sim \mathcal{U}_n$ and querying $\mathcal{O}(X)$. It runs $A$ with a small constant accuracy parameter $\eta$ and confidence parameter $\eta$, obtaining a sampler $C$ over pairs $(x,b) \in \bool^n \times \bool$. It then validates $C$ as follows. Letting $Q$ be the set of inputs $x \in \bool^n$ queried to $\mathcal{O}$ during the training phase, take one fresh sample $(x', b') \sim C$, and accept if and only if $x' \notin Q$ and $b' = \mathcal{O}(x')$.

    On the pseudorandom side, $A$ is being run on samples from a distribution in the class handled by P.II. Hence, with probability at least $1-\eta$, the output sampler $C$ satisfies $\Delta_{\mathsf{TV}}(C(\mathcal{U}),G_{f_z(u,\cdot)}) \leq \eta$. Total variation closeness implies that $C$ is supported almost entirely on the graph of $f_z(u,\cdot)$ (so $b' = \mathcal{O}(x')$ with high probability) and that its first marginal is close to $\mathcal{U}_n$ (so $x' \notin Q$ with high probability). Therefore the validation test accepts with probability $1-O(\eta)$.

    On the random side, let $\mathcal{O}=R$ be a uniformly random function. After conditioning on the polynomially many graph points revealed while simulating the samples for $A$, the values of $R$ at fresh points are still independent fair bits. The ``$x' \notin Q$'' test rules out the only simple way for the sampler $C$ to pass the graph consistency test, namely putting noticeable mass on the already revealed inputs. Since $C$ depends only on the coins of $A$ and $R|_Q$, $R(x')$ is uniform and independent of $C$ for $x' \notin Q$. Consequently, the output of any polynomial-time learner passes the ``$b' = \mathcal{O}(x')$'' test with probability bounded away from the pseudorandom case, after fixing $\eta$ sufficiently small.

    Thus $D$ distinguishes $f_z(u,\cdot)$ for $u \sim \mathcal{U}_n$ from a truly random function with inverse polynomial advantage for all sufficiently large relevant input lengths. This contradicts the auxiliary input pseudorandomness condition. Therefore AIPRFs do not exist.
\end{proof}

\subsection{White-Box Learning and One-Way Functions}\label{sec:wb}

In this section, we explore connections between  \emph{white-box learning}, in which a labelled example distribution is given explicitly as a sampling circuit, and the existence of (auxiliary-input) one-way functions. These connections are useful as tools to establish the ordering among the learning paradises of \Cref{sec:paradises}, particularly that P.II (Efficient Distribution Learning) implies P.III (Efficient Agnostic Learning), and may also be of interest in their own right.

\subsubsection{White-Box Learning Frameworks}

The model of white-box learning we consider is the same as Agnostic Learning (P.III), except that the learner is additionally provided an explicit description of a sampler for the example distribution in question. For this reason, we restrict ourselves here to distribution families $E = \{E_n\}_{n \in \nat}$ samplable by polynomial-size circuits.

More specifically, consider a joint labelled example distribution $E_n \in \P\SAMP/\poly$ supported over $\bool^n \times \bool$, where for $(x, b) \sim E_n$, $x \in \bool^n$ is called the ``example string'' and $b \in \bool$ is called the ``label''. We let $E^{(0)}_n$ denote the marginal distribution over example strings and $E_n^{(1)}$ that over the labels. 

As mentioned above, the learner is provided a sampling circuit $S_n$ for $E_n$. Of course, this could potentially make the task of learning much easier. However, note that the learning task is still non-trivial: even though we can easily sample input/output pairs for the function in question, it may still be hard to compute the function on a desired input.

As in the setting of agnostic learning, we consider optimality with respect to a fixed concept class, $\mathcal{C}$. That is, define
\[{\sf opt}_n(\mathcal{C}, E) \eqdef \max_{c \in \mathcal{C}_n}\left\{\Pr_{(x, b) \sim E_n}[c(x) = b]\right\}.\]

Since a white-box learner has the sampler $S_n$ available, it can simulate this circuit on random inputs to obtain labelled examples distributed according to $E_n$. For this reason, it does not require oracle access to $\EX(E_n)$.

White-box learning is defined as follows.

\begin{definition}[White-Box Learning \cite{DBLP:conf/focs/ApplebaumBX08, xiao2009new}]\label{def:wb-learn}

    Let $\mathcal{E} \subseteq \P\SAMP/\poly$ be a class of joint distributions. An algorithm $A$ \emph{white-box learns} $\mathcal{C}$ over $\mathcal{E}$ if, for every $E = \{E_n\}_{n \in \nat} \in \mathcal{E}$ and sequence $S = \{S_n\}_{n \in \nat}$ with $S_n$ being a sampling circuit for $E_n$, for all sufficiently large $n \in \nat$, for all $\delta, \varepsilon \in (0,1)$,
    \[\Pr_{A}\left [ \; A(1^n, 1^{\varepsilon^{-1}}, \delta^{-1}, S_n) \text{ outputs } h \text{ such that } \Pr_{(x, b)\sim E_n}[h(x) = b] \geq {\sf opt}_n(\mathcal{C}, E) - \varepsilon \; \right] \geq 1 - \delta.\]
    
\end{definition}

We will primarily be interested in White-Box Learning for the following choices of concept class $\mathcal{C}$: $\mathcal{C} = \SIZE[s]$ for a fixed function $s: \nat \to \nat$, and $\mathcal{C} = \mathcal{F}_{all}$, where $(\mathcal{F}_{all})_n$ denotes the set of all boolean functions $f: \bool^n \to \bool$. In particular, in the latter case, the learner aims to approximate the Bayes-optimal predictor for $\mathcal{E}$.

\subsubsection{White-Box Learning and Breaking Auxiliary-Input One-Way Functions}\label{sec:white_owf}

In this section, we establish the following:
$$
\neg \exists \text{AIOWF}  \;\; \iff \;\; \text{White-Box Learning $\mathcal{F}_{all}$}.
$$
\vspace{-0.6cm}

The implication from $\neg \exists \text{AIOWF}$ to the easiness of White-Box Learning $\mathcal{F}_{all}$ was proved in \cite{xiao2009new} and implicit in \cite{DBLP:conf/focs/ApplebaumBX08}.

Also note that, since White-Box Learning $\mathcal{F}_{all}$ implies White-Box Learning $\mathcal{C}$ for any desired concept class $\mathcal{C} \subseteq \mathcal{F}_{all}$, the above implies that 
$$
\neg \exists \text{AIOWF}  \;\; \implies \;\; \text{White-Box Learning $\SIZE[s]$}
$$
for any $s: \nat \to \nat$. The latter may be used in a proof that P.II (Efficient Distribution Learning) $\implies$ P.III (Efficient Agnostic Learning), as in \Cref{lem:worldII-worldIII}.

\cite{xiao2009new} leaves as an open question whether White-Box Learning $\SIZE[s]$, for a fixed polynomial $s$, implies that AIOWFs can be broken. Note that a positive answer would show the equivalence of White-Box Learning with respect to $\SIZE[s]$ and $\mathcal{F}_{all}$ over distributions samplable in $\P\SAMP/\poly$.

\begin{problem}[White-Box Learning and AIOWF]

    Does White-Box Learning $\SIZE[s]$ imply the non-existence of AIOWF? Alternatively, is there evidence to the contrary?

\end{problem}

We begin with the following auxiliary lemma.

\begin{lemma}\label{lem:rand-h}

    Assume $\neg \exists$ AIOWF. Let $E$ be a family of labelled example distributions with sampling circuits $S$, as in \Cref{def:wb-learn}. There exists a polynomial-time algorithm that, for all sufficiently large $n \in \nat$, on input $(1^n, 1^{(\varepsilon')^{-1}}, S_n)$, outputs a \emph{randomized} hypothesis $h$ such that $\Delta((x, h(x)), E_n) \leq \varepsilon'(n)$, where $x \sim E^{(0)}_n$.
    
\end{lemma}
\begin{proof}

    Let $\{E_n = (E_n^{(0)}, E_n^{(1)})\}_{n \in \nat}$, with each $E_n$ supported over $\bool^n \times \bool^{m(n)}$, be a labelled example distribution samplable by poly-size circuits $\{S_n = (S_n^{(0)}, S_n^{(1)})\}_{n \in \nat}$, each taking randomness of length $\rho(n) = \poly(n)$ as input. Define a poly-time computable auxiliary input function $f$ as follows.
    \[f_{S_n^{(0)}}(r) = S_n^{(0)}(r).\]
    That is, $f$ has the sampler $S_n^{(0)}$ hard-coded as an auxiliary input. Its input is a choice of randomness $r$, and its output is the output of $S_n^{(0)}$ run on $r$.

    Assuming $\neg$ $\exists$ AIOWF, by \Cref{lem:ai-crypto}, $\neg \exists$ distAIOWF. For any desired inverse polynomial $\varepsilon'$, let $I$ be a PPT $\varepsilon'$-distributional inverter for $f$. That is, for any sufficiently large $n \in \nat$,
    \begin{align}\label{eq:I-inverts}
        \Delta\left(\; (f_{S_n^{(0)}}(r), I(S_n^{(0)}, f_{S_n^{(0)}}(r))), \; (f_{S_n^{(0)}}(r), r) \;\right) \leq \varepsilon'(n),
    \end{align}
    where $r \sim \mathcal{U}_{\rho(n)}$.
    
    The desired algorithm will output the following randomized hypothesis $h$:
    \begin{quote}
        On input $x \in \bool^n$, $h$ invokes $I$ on $x$ with auxiliary input $S_n^{(0)}$ to produce a pre-image $r$; then, $h$ outputs $S_n^{(1)}(r)$.
    \end{quote} 
    Observe that for $r \sim \mathcal{U}_{\rho(n)}$ and $x = S_n^{(0)}(r)$ (that is, $x \sim E_n^{(0)}$),
    \begin{align}
        \Delta((x, h(x)), E_n) &= \Delta((x, S_n^{(1)}(I(S_n^{(0)}, x))), (x, S_n^{(1)}(r))) \tag{definition of $h$}
        \\&\leq \Delta((x, I(S_n^{(0)}, x)), (x, r)) \notag
        \\&= \Delta((f_{S_n^{(0)}}(r), I(S_n^{(0)}, x)), (f_{S_n^{(0)}}(r), r)) \tag{definition of $f$}
        \\&\leq \varepsilon'(n). \tag{Eq.~(\ref{eq:I-inverts})}
    \end{align}

    This completes the proof.
    
\end{proof}

The next two lemmas establish an equivalence between $\neg \exists$ AIOWF and the easiness of White-Box Learning $\mathcal{F}_{all}$.

\begin{lemma}[\cite{DBLP:conf/focs/ApplebaumBX08, xiao2009new}]\label{lem:aiowf-all-wb}

    If $\neg \exists$ AIOWF, then White-Box Learning $\mathcal{F}_{all}$ over $\P\SAMP/\poly$ is easy.
    
\end{lemma}
\begin{proof}

    The desired white-box learner, given a sampler $S_n$ for $E_n$ and parameters $\varepsilon$ and $\delta$, will construct $h$ from $S_n$ as in \Cref{lem:rand-h} with parameter $\varepsilon' = \varepsilon^4$. It will then sample $m \eqdef 8\varepsilon^{-4}\cdot n \cdot \log(1/\delta)$ choices of randomness $r_1, .., r_{m}$ for $h$, and output the deterministic hypothesis $H$ that, on input $x$, outputs ${\sf maj}_{i \in [m]}\{h(x; r_i)\}$.

    We now argue for the correctness of this white-box learner. Recall that $h$ is such that $\Delta((x, h(x)), E_n) \leq \varepsilon'$, where $x \sim E^{(0)}_n$. Let $f_{opt}: \bool^n \to \bool$ be the function that maximises $\Pr_{(x, b) \sim E_n}[f_{opt}(x) = b]$, and, for $x \in \bool^n$, let $B_x$ denote the conditional distribution of $E_n^{(1)}$ given $E_n^{(0)} = x$. Define 
    \[\alpha(x) = \max_{b \in \bool}\left\{\Pr[B_x = b]\right\}.\]
    Note that $f_{opt}(x) = \arg\max_{b \in \bool} \{\Pr[B_x = b]\}$, so $\Pr[f_{opt}(x) = B_x] = \alpha(x)$.

    By Markov's inequality, the set $\mathcal{X}_{bad} = \{x : \Delta(h(x), B_x) > \varepsilon^2\}$ satisfies $\Pr_{x \sim E_n^{(0)}}[x \in \mathcal{X}_{bad}] \leq \varepsilon^2$. For all $x \notin \mathcal{X}_{bad}$, we have $|\Pr[h(x) = b] - \Pr[B_x = b]| \leq \varepsilon^2$ for every $b \in \bool$.
    
    For $x \notin \mathcal{X}_{bad}$, we consider two cases.

    \paragraph{Case I: $\alpha(x) > 1/2 + 2\varepsilon^2$.} Note that
    \begin{align*}
        \Pr_{r_i}[h(x; r_i) = f_{opt}(x)] &\geq \alpha(x) - \varepsilon^2
        \\&> \frac{1}{2} + \varepsilon^2.
    \end{align*}
    Then, by a Chernoff bound, 
    \[\Pr_{r_1, ..., r_{m}}\left[H(x) = f_{opt}(x)\right] \geq 1 - \delta\cdot 2^{-n}.\]
    \hfill\mbox{$\diamond$}

    \paragraph{Case II: $\alpha(x) \leq 1/2 + 2\varepsilon^2$.} Note that $\alpha(x) \geq 1/2$ also holds in this case, by definition of $\alpha(x)$. Consider any fixed $H$. If $H(x) = f_{opt}(x)$, then $\Pr[H(x) = B_x] \geq 1/2$. On the other hand, if $H(x) = 1 - f_{opt}(x)$, then $\Pr[H(x) = B_x] \geq 1 - \Pr[f_{opt}(x) = B_x] \geq 1/2 - 2\varepsilon^2$. It follows that 
    \[ \Pr[f_{opt} = B_x] - \Pr[H(x) = B_x] \leq \frac{1}{2} + 2\varepsilon^2 - \left (\frac{1}{2} -2\varepsilon^2\right ) = 4\varepsilon^2.\]
    \hfill\mbox{$\diamond$}

    By a union bound, with probability $1 - \delta$ over the choice of $r_1, ..., r_{m}$, all $x$ satisfying Case I satisfy $H(x) = f_{opt}(x)$. Then, with probability $1 - \delta$, for all $x$ satisfying either case, 
    \[\Pr[f_{opt} = B_x] - \Pr[H(x) = B_x] \leq 4\varepsilon^2.\]
    Overall, with probability at least $1 - \delta$ over the randomness of the white-box learner, it outputs a hypothesis $H$ such that 
    \[\Pr_{(x, b) \sim E_n}[H(x) = b] \geq {\sf opt}_n(\mathcal{F}_{all}, E) - 4\varepsilon^2 - \varepsilon^2 \geq {\sf opt}_n(\mathcal{F}_{all}, E) - \varepsilon,\]
    assuming $\varepsilon < 1/5$ without loss of generality.
    
    This completes the proof. 
    
\end{proof}

\begin{lemma}\label{lem:all-wb-aiowf}

    If White-Box Learning $\mathcal{F}_{all}$ over $\P\SAMP/\poly$ is easy, then $\neg \exists$ AIOWF.
    
\end{lemma}
\begin{proof}

    Toward a contradiction, suppose there exists AIOWF. By \Cref{lem:ai-crypto}, there likewise exists an AIPRG, $G_z :\bool^n \to \bool^{2n}$. By an exhaustive search running in time exponential in $n$, it is possible to distinguish outputs of $G_z$ from truly random strings. Let $D$ be this inefficient distinguisher. In particular, for any given $z$, 
    \[\left | \Pr_{v \sim \mathcal{U}_n}\left[D(z, G_z(v)) = 1\right] - \Pr_{u \sim \mathcal{U}_{2n}}\left[D(z, u) = 1\right]\right| \geq 1 - o(1).\]
    By Yao's hybrid argument (see \Cref{sec:yao}), there exists an inefficient next-bit predictor $P^D$ such that, for any $z$, 
    \[\Pr_{P^D; \; i \sim [2n]; \; v \sim \mathcal{U}_n}\left[ P^D(z, G_z(v)_{< i}) = G_z(v)_{i}\right] > \frac{1}{2} + \frac{1}{8n},\]
    where $y_j$ denotes the $j^{\text{th}}$ bit of a string $y$, and $y_{<j}$ denotes the prefix of $y$ up to and including position $j - 1$.

    Now, consider the distribution $E_z$, which samples a random index $i \sim [2n]$ and a random seed $v \sim \mathcal{U}_n$, and then outputs 
    \[(z, i, G_z(v)_{<i}, 0^{2n - i + 1}; G_z(v)_i).\] 
    It is easy to check that $E_z$ is samplable by a polynomial-size circuit, which can be constructed in polynomial time given $z$. Moreover, by the existence of $P^D$, we have that 
    \[\opt_{N}(\mathcal{F}_{all}, E_z) \geq \frac{1}{2} + \frac{1}{8n},\]
    where $N = |z| + 2n + O(\log n)$ is the length of an example string from $E_z$. Going forward, we omit the trailing 0s used for padding.

    We claim that a white-box learner for $E_z$ may be used to break $G_z$. Let $A'$ be a white-box learner with respect to $\mathcal{F}_{all}$. Consider the following adversary $A$:
    \begin{quote}
        On input $(z, y)$, construct a sampler $S_z$ for $E_z$, and simulate $A'(S_z)$ with accuracy parameter $\varepsilon = n^{-2}$ and confidence $\delta = 2^{-n}$ to obtain a hypothesis $h$. Sample a random $j \sim [2n]$, and accept iff $h(z, j, y_{<j}) = y_j$.
    \end{quote}

    First consider the case that $y \sim G_z(\mathcal{U}_n)$. In this case, 
    \begin{align*}
        \Pr_{A, y}\left [ A(z, y) = 1 \right] &= \Pr_{A', y, j}\left [ h(z, y_{<j}) = y_j \right]
        \\&= \Pr_{A', (x, b) \sim E_z}\left [ h(x) = b \right]
        \\&\geq \frac{1}{2} + \frac{1}{8n} - \frac{1}{n^2} - \frac{1}{2^n}
        \\&> \frac{1}{2} + \frac{1}{10n}.
    \end{align*}

    Next, in the case that $y \sim \mathcal{U}_{2n}$, $y_j$ is a uniformly random bit independent of $(z, y_{<j})$, so 
    \[\Pr_{A, y}\left [ A(z, y) = 1 \right] = \frac{1}{2}.\]

    We conclude that $A$ breaks $G_z$, as desired.
\end{proof}

\newpage

\section{Learning via Meta-Complexity and Algorithmic Information Theory}\label{sec:learning_via_meta}

\subsection{Detecting Structure: Supervised Learning via Distinguishers}\label{sec:learn-via-dist}

Valiant introduced his framework of \emph{PAC learning} in 1984 \cite{DBLP:journals/cacm/Valiant84}; in the same paper, he already recognized a tension between the possibility of efficient learning and the possibility of cryptography. He points out in particular that if PAC learning is easy, then it would be easy to break a pseudorandom function. Though the converse of this specific statement remains unknown, the connection also goes in the other direction, where algorithms for ``detecting structure from randomness'' may be used to carry out learning tasks.

In this section, we describe a general framework for establishing relationships in the latter direction. In this framework, algorithms for \emph{meta-complexity} (namely, deciding whether a given object is of ``high'' or ``low'' complexity) are used to break cryptographic and complexity-theoretic pseudorandom generators, yielding learning algorithms by applying reconstruction properties of those generators. We consider two different styles of generator, which we refer to as the Direct Product Generator and the Nisan-Wigderson Generator, respectively. 

The Direct Product Generator takes the form
\[G^f_{DP}(\overline{x}) = (x_1, ..., x_m; b_1, ..., b_m),\]
where $\overline{x} = (x_1, ..., x_m) \in (\bool^n)^m$ are input examples, and $(b_1, ..., b_m) \in \bool^m$ are some labels associated with a target function $f\colon \bool^n \to \bool$.

The Nisan-Wigderson Generator takes the form
\[G^f_{NW}(z) = (f(z|_{S_1}), ..., f(z|_{S_m})),\]
where $z \in \bool^{O(n^2)}$ is an input specifying some examples, $f \colon \bool^n \to \bool$ is a target function, and $S_1, ..., S_m$ are \emph{NW-designs}. An NW-design is a collection of $m$ subsets (where $m$ is the desired output length) chosen from a small ground set of size $\ell$. They are constructed to have bounded intersections: for instance, any two distinct subsets $S_i, S_j$ might overlap in at most $\log m$ elements. Roughly speaking, because the subsets intersect very little, the inputs fed into a ``hard'' function $f$ producing the output bits remain mostly independent.

In the results we survey below, we will see various instantiations of these generators, which may be broken under various assumptions about the easiness of meta-complexity. We mention in advance that learners obtained by breaking the Direct Product Generator can work from \emph{random examples}, whereas learners obtained by breaking the Nisan-Wigderson Generator require \emph{membership queries}. In cases where the corresponding assumption from meta-complexity is known to hold, the approach has led to unconditional learning algorithms.

\subsubsection{From Distinguishers to Predictors}\label{sec:yao}

The central technique driving all the results in this section is Yao's hybrid argument, which converts a \textit{distinguisher} into a \textit{next-bit predictor} \cite{DBLP:conf/focs/Yao82a}. First, we explain what it means for an algorithm to be a distinguisher. Consider an arbitrary joint distribution $(X, Y)$, where $X$ is supported over $\bool^{m}$ and $Y$ over $\bool^{\ell}$. Let $A$ be a function (which we will typically identify with an algorithm) mapping $m + \ell$ input bits to one output bit. We consider the case that $A$ \emph{distinguishes} $(X, Y)$ from $(X, \mathcal{U}_{\ell})$ with advantage $\varepsilon > 0$, where $\mathcal{U}_{\ell}$ denotes the uniform distribution over $\bool^{\ell}$. Namely,
\[\left | \Pr_{A}[A(X,Y) = 1] - \Pr_{A}[A(X, \mathcal{U}_{\ell}) = 1] \right | \geq \varepsilon.\]
Yao showed that, from any such $A$, one can construct another algorithm $P_A: \bool^{m} \times \bool^{<\ell} \to \bool$ taking an output of $X$ together with some prefix of an output of $Y$ and producing the next unseen bit of $Y$ with better-than-trivial success probability. In particular,
\[\Pr_{P_A, i \sim [\ell]}[P_A(X, Y_1, ..., Y_{i - 1}) = Y_i] \geq \frac{1}{2} + \frac{\varepsilon}{2\ell}.\]
We refer to such an algorithm $P_A$ as a \emph{next-bit predictor}.

To prove this statement, Yao employs a hybrid argument. Roughly, one defines ``hybrid'' distributions
\[H_j = (X, Y_1, ..., Y_j, \mathcal{U}_{\ell - j})\]
for $j \in \{0, ..., \ell\}$, so $H_0 = (X, \mathcal{U}_{\ell})$ and $H_{\ell} = (X,Y)$. Observe that
\[\Pr[A(H_0) = 1] - \Pr[A(H_\ell) = 1] = \sum_{j \in \{0, ... \ell - 1\}}\Pr[A(H_j) = 1] - \Pr[A(H_{j + 1}) = 1].\]
Supposing that $\Pr[A(X,\mathcal{U}_{\ell}) = 1] - \Pr[A(X, Y) = 1] \geq \varepsilon$, so the above is at least $\varepsilon$, it holds that the difference $\Pr[A(H_j) = 1] - \Pr[A(H_{j + 1}) = 1]$ is at least $\varepsilon/\ell$ in expectation over the choice of $j$.\footnote{The case that $\Pr[A(X, Y) = 1] - \Pr[A(X, \mathcal{U}_{\ell}) = 1] \geq \varepsilon$ can be handled similarly.} This expected advantage of $A$ at the $j^{\text{th}}$ bit of $Y$ can then be leveraged to obtain a next-bit predictor at the corresponding position: when given a random bit $u \sim \mathcal{U}$ at position $j + 1$, $A$ is slightly less likely to output $1$ when $u = Y_{j+1}$ than when $u \neq Y_{j+1}$. We refer the reader to \cite{DBLP:journals/fttcs/Vadhan12} for further detail.

\subsubsection{Direct Product and Random Examples}\label{sec:direct_product_random_examples}

\paragraph{From PRG Distinguishers to Learning Random Target Functions.} 

An early instantiation of the framework described above can be seen in the work of Blum et al. \cite{DBLP:conf/crypto/BlumFKL93}, where the authors show that breaking a cryptographic pseudorandom generator yields a learner working in a certain average-case setting. In more detail, from a successful distinguisher against a candidate PRG construction, one obtains a learner succeeding with high probability over randomly sampled target functions (see \Cref{def:bfkl}). 

The starting point is the assumption that for every candidate PRG $G: \bool^n \to \bool^{\ell(n)}$ there exists an efficient algorithm $A: \bool^{\ell(n)} \to \bool$ distinguishing outputs of $G$ from uniformly random strings. That is, for some polynomial $p$, for all sufficiently large $n$,
\[\left | \Pr[A(G(\mathcal{U}_n)) = 1] - \Pr[A(\mathcal{U}_{\ell(n)}) = 1] \right | \geq \frac{1}{p(n)}.\]

Let $P = \{P_n\}_{n \geq 1}$ be a collection of distributions over target functions. For the sake of this exposition, we will discuss the case of weak prediction over example strings drawn from the uniform distribution.\footnote{\cite{DBLP:conf/crypto/BlumFKL93} use a different argument, inverting an appropriate one-way function, to yield strong prediction over examples drawn from other distributions in $\P\SAMP$. See Theorem 11 in \cite{DBLP:conf/crypto/BlumFKL93}.} We define a candidate PRG $G$ as follows, where $m = m(n)$ is a sufficiently large parameter.
\begin{quote}
    On input $(r_P, x_1, ..., x_m)$, with $r_P \in \bool^{\poly(n)}$ and each $x_i \in \bool^n$, $G$ samples $\sigma_f \sim P_n$ using randomness $r_P$, where $\sigma_f$ is the representation of a function $f$. For each $i \in [m]$, it then computes $b_i \eqdef f(x_i)$. $G$ outputs $(x_1, ..., x_m, b_1, ..., b_m)$.
\end{quote}
The goal is then to transform a distinguisher $A$ for $G$ into a learner for $P_n$ over $\mathcal{U}_n$ in the sense of Definition~\ref{def:bfkl}. By definitions of $G$ and $A$, we have 
\[\left | \Pr_{\substack{\sigma_f \sim P_n; \\ x_1, ..., x_m \sim \mathcal{U}_n}}[A(x_1, ..., x_m, f(x_1), ..., f(x_m)) = 1] - \Pr[A(\mathcal{U}_{m\cdot n + m}) = 1] \right | \geq \frac{1}{p(n)}.\]
By Yao's hybrid argument, for some PPT machine $A'$ built from $A$, 
\[\Pr_{\substack{\sigma_f \sim P_n; \;A'; \; i \sim [m]; \\ x_1, ..., x_m \sim \mathcal{U}_n}}[A'(x_1, ..., x_m, f(x_1), ..., f(x_{i-1})) = f(x_i)] \geq \frac{1}{2} + \frac{1}{2m\cdot p(n)}.\]
We may then define a learner $A''$ as follows.
\begin{quote}

    On input $((x_1, f(x_1)), ..., (x_m, f(x_m)); x')$, sample $i \sim [m]$, and output 
    \\$A'(x_1, ..., x_{i-1}, x', x_{i+1}, ..., x_m, f(x_1), ..., f(x_{i - 1}))$.
    
\end{quote}
By the definitions of $A'$ and $A''$, we have that 
\[\Pr_{\substack{\sigma_f \sim P_n; \\ (x_i, f(x_i))_{i \in [m(n)]} \sim \EX(f, \,\mathcal{U}_n)^{m(n)}; \\ x' \sim \mathcal{U}_n }}\left[A''((x_1, f(x_1)), ..., (x_{m(n)}, f(x_{m(n)})), x') = f(x')\right] \geq \frac{1}{2} + \frac{1}{2m \cdot p(n)}.\]
Thus, $A''$ satisfies the requirement of weak prediction.

\paragraph{From RRHS Refutation to PAC Learning over Arbitrary Distributions.}

Since the insight of \cite{DBLP:conf/crypto/BlumFKL93}, several papers continued to employ distinguishers to obtain learners. In an elegant example of such work, Vadhan \cite{DBLP:conf/colt/Vadhan17} formalises a distinguishing task to capture \emph{worst-case} PAC learning over \emph{arbitrary} example distributions. Specifically, for a concept class $P$ and its dual class $P^* = Q$ (where for an evaluation function ${\sf Eval}: \{0,1\}^s \times \{0,1\}^t \to \{0,1\}$, $p_y(x) = q_x(y) = \mathsf{Eval}(y,x)$), Vadhan defines the \emph{Random Right-Hand Side (RRHS) Refutation} problem for $Q$, where the task is to distinguish
\begin{itemize}
    \item $(x_1, b_1), ..., (x_m, b_m)$, for independently random bits $b_i \sim \mathcal{U}$, from
    \item $(x_1, b_1), ..., (x_m, b_m)$, where $b_i = q_{x_i}(y) = p_y(x_i)$ for a fixed concept $f = p_y \in P$.
\end{itemize}
Observe that this is precisely the task of distinguishing the output of the Direct Product Generator
\[G^{f}_{DP}(x_1, \ldots, x_m) = \bigl(x_1, \ldots, x_m,\, f(x_1), \ldots, f(x_m)\bigr)\]
from 
\[(x_1, \ldots, x_m, \mathcal{U}_m),\]
for arbitrary strings $x_1, ..., x_m$.

We are interested here in the reduction from PAC learning to RRHS Refutation. Let $A$ be an RRHS Refuter for some concept class $\mathcal{C}$: that is, for any $f \in \mathcal{C}$ and any choice of $x_1, ..., x_m$, $A$ distinguishes the two distributions above. As in the work of \cite{DBLP:conf/crypto/BlumFKL93}, Yao's hybrid argument yields the existence of a machine $A'$ such that
\[\Pr_{r \sim \mathcal{U}_{\poly(n)}; \;i \sim [m]}\left[\;A'(x_1, ..., x_m, f(x_1), ..., f(x_{i-1}); r) = f(x_i)\;\right] \geq \frac{1}{2} + \frac{1}{\poly(n)}.\]
Letting $D_n$ be an arbitrary example distribution, an averaging argument implies that for some polynomial $p$, with probability at least $1/p(n)$ over the choice of $r$, $i$, and $x_1, ..., x_{i-1}, x_{i+1}, ..., x_m \sim D^{m-1}_n$, 
\[\Pr_{x_i \sim D_n}\left[\;A'(x_1, ..., x_m, f(x_1), ..., f(x_{i-1}); r) = f(x_i)\;\right] \geq \frac{1}{2} + \frac{1}{p(n)}.\]
Now, let $B$ be the algorithm that takes $m$ examples $(x_1, f(x_1)) ..., (x_m, f(x_m))$ from its oracle $\EX(f, D_n)$, chooses a random $i \sim [m]$ and randomness $r$ for $A'$, and then outputs the hypothesis 
\[h(x) \eqdef  A'(x_1, ..., x_{i - 1}, x, x_{i + 1}, ..., x_m, f(x_1), ..., f(x_{i-1}); r).\] 
By the above, 
\[\Pr_{B; \; \EX(f, D_n)}\left [\; B^{\EX(f, D_n)} \text{ outputs } h \text{ such that } \Pr_{x\sim D_n}[h(x) = f(x)] \geq \frac{1}{2} + \frac{1}{p(n)}\;\right ] \geq \frac{1}{p(n)}.\]
Note that the confidence parameter $\delta = 1 - p(n)^{-1}$ can be boosted by parallel repetition, using empirical estimation to determine the best hypothesis produced. 

In this way, the easiness of RRHS Refutation yields a weak PAC learner. A classic result of Schapire shows that such a learner can be \emph{boosted} to obtain a (strong) PAC learner \cite{DBLP:journals/ml/Schapire90}.

\paragraph{Correlative RRHS Refutation and Agnostic Learning.}

In work independent of \cite{DBLP:conf/colt/Vadhan17}, Kothari and Livni \cite{DBLP:conf/innovations/KothariL18} give a related refutation problem, denoted \emph{correlative RRHS Refutation}, to characterise the harder task of \emph{agnostic learning}. Recall that in the setting of agnostic learning, examples from the oracle $\EX$ are not guaranteed to correspond exactly to any fixed target function from the relevant concept class; the learner is only asked to produce a hypothesis that is ``almost'' as accurate as the most accurate function from the concept class \cite{DBLP:journals/ml/KearnsSS94, DBLP:journals/iandc/Haussler92}.

One difference compared to the RRHS Refutation task of \cite{DBLP:conf/colt/Vadhan17} is that correlative RRHS Refutation is defined in a \emph{distribution-specific} setting. Then, the task of correlative RRHS Refutation over some particular example distribution characterises that of agnostic learning over the \textit{same} distribution, thanks to the existence of distribution-specific boosting algorithms for agnostic learners \cite{DBLP:conf/nips/KalaiK09, DBLP:conf/innovations/Feldman10}. This is unlike the setting of \cite{DBLP:conf/colt/Vadhan17}, which requires an RRHS refuter working on \emph{every} distribution to obtain a learner working on every distribution.

As in the works discussed above, a weak learner in \cite{DBLP:conf/innovations/KothariL18} arises from a hybrid argument. More specifically, suppose we have a \textit{correlative RRHS refuter} that is capable of detecting samples with correlation $\delta$ with circuits in $\SIZE[s]$. Namely, let $\mathcal{E}$ be a joint example distribution satisfying
\[\max_{c \in \SIZE[s(n)]}\left\{\Pr_{(x,b) \sim \mathcal{E}_n}\left[c(x) = b\right]\right\} \geq \frac{1}{2} + \delta,\]
and suppose we have an algorithm $A$ that $\Omega(1)$-distinguishes
\begin{itemize}
    \item $(x_1, b_1), ..., (x_m, b_m)$, for independently random $(x_i, -) \sim \mathcal{E}_n$ and $b_i \sim \mathcal{U}$, from
    \item $(x_1, b_1), ..., (x_m, b_m)$, where $(x_i, b_i) \sim \mathcal{E}_n$ independently.
\end{itemize}
Then, by a hybrid argument, given $m$ random examples from $\mathcal{E}_n$, $A$ yields a construction of a hypothesis $h$ such that 
\[\Pr_{(x, b) \sim \mathcal{E}_n}[h(x) = b] \geq \frac{1}{2} + \frac{1}{O(m)}.\]
This is sufficient to apply known boosting algorithms \cite{DBLP:conf/nips/KalaiK09, DBLP:conf/innovations/Feldman10} to obtain a (strong) agnostic learner over $\mathcal{D}$.

\paragraph{From Average Case Easiness of NP to Agnostic Learning.}

Hirahara and Nanashima \cite{DBLP:conf/focs/HiraharaN21} build on \cite{DBLP:conf/innovations/KothariL18} to prove that efficient agnostic learning of $\P/\poly$ over distributions in $\P\SAMP/\poly$ is possible under the assumption that $\DistNP \subseteq \Avg\P$. The paper also builds on an earlier line of work due to Hirahara, which proved that several promise problems in meta-complexity exhibit \textit{worst-case to average-case reductions} (e.g. \cite{DBLP:conf/focs/Hirahara18, DBLP:conf/focs/Hirahara20}). In particular, if an appropriate `gap' version of the time-bounded Kolmogorov complexity problem, denoted $\GapMINKT$, can be solved in deterministic poly-time on average over the uniform distribution, then it is likewise easy in the worst case. 

The key contribution of \cite{DBLP:conf/focs/HiraharaN21} consists in reducing correlative RRHS Refutation to $\GapMINKT$. The intuition is that samples drawn according to the `random' case are of high complexity with high probability, whereas the $\K^t$ complexity of samples from the `correlative' case can be bounded. For simplicity, consider the setting of realizable PAC learning, so the examples in the latter case are actually consistent with some function in $\SIZE[s]$ for a polynomial $s$. 

In the `random' case, by Symmetry of Information for Kolmogorov complexity and a counting argument, one can show that for any time-bound $t$, with high probability over $x_i \sim \mathcal{D}_n$ and $b_i \sim \mathcal{U}$,
\[\K^t(x_1, ..., x_m, b_1, ..., b_m) \gtrsim \K(x_1, ..., x_m) + m.\]
On the other hand, if the labels $b_i$ are equal to $f(x_i)$ for some function $f \in \SIZE[s]$, for a sufficiently large polynomial $t'$, we have
\begin{align*}
    \K^{\poly(t')}(x_1, ..., x_m, b_1, ..., b_m) &\leq \K^{t'}(x_1, ..., x_m) + O(s(n)) + O(\log n) 
    \\&\ll \K^{t'}(x_1, ..., x_m) + m,
\end{align*}
provided that $m$ is sufficiently larger than $s$. This follows by encoding a description of an $s$-size circuit for $f$, and computing each $b_i$ by running this circuit on input $x_i$. The idea is then to run an algorithm for $\GapMINKT$ on the given samples to distinguish ``high'' complexity from ``low'' complexity. 

At a technical level, several difficulties arise from the differing (finite and infinite) time-bounds involved, and from the fact that one cannot exactly compute $\K^t$ complexity for a desired $t \in \nat$ under the assumption $\DistNP \subseteq \Avg\P$. In fact, the definition of $\GapMINKT$ is such that one only gets a poly-time algorithm $A$ such that $A(1^t, x)$ outputs $k \in \nat$ such that 
\[\K^{p(t)}(x) - \log p(t) \leq k \leq \K^t(x)\]
for some polynomial $p$. Roughly speaking, the authors overcome this by proving that, for a sufficiently large time-bound $t$ (depending on the running time to sample from a given distribution $\mathcal{D}_n \in \P\SAMP/a(n)$), 
\[\K^t(\overline{x}) \approx \K(\overline{x}) \pm a(n)\]
with high probability over strings $\overline{x} = (x_1, ..., x_m)$ sampled from $(\mathcal{D}_n)^m$. Thus, the differences between complexity measures with different time-bounds can be made immaterial by taking the number of samples $m$ sufficiently larger than $a(n)$. 

The approximate equality above follows from a \emph{coding theorem} for time-bounded Kolmogorov complexity, which holds under the assumption $\DistNP \subseteq \Avg\P$. In particular, for any distribution $\mathcal{D}_n \in \P\SAMP/a(n)$ and any string $x$ in the support of $\mathcal{D}_n$, for sufficiently large $t \in \nat$ (depending on the exact complexity of sampling $\mathcal{D}$),
\[\K^t(x) \leq \log(1/\mathcal{D}_n(x)) + a(n) + O(\log n).\]
Since $\mathcal{D}_n \in \P\SAMP/a(n)$ implies $(\mathcal{D}_n)^m \in \P\SAMP/a(n)$ when $m = \poly(n)$, one obtains a similar statement for $\overline{x} \in \supp((\mathcal{D}_n)^m)$. Given this statement, it holds that
\begin{align*}
    \mathop{\mathbb{E}}_{\overline{x} \sim \mathcal{D}_n^m}[\K^t(\overline{x}) - \K(\overline{x})] &\leq \mathop{\mathbb{E}}_{\overline{x} \sim \mathcal{D}_n^m}[\K^t(\overline{x})] - \mathsf{H}(\mathcal{D}_n^m) + O(\log n)
    \\&\leq \mathop{\mathbb{E}}_{\overline{x} \sim \mathcal{D}_n^m}[\log(1/\mathcal{D}_n^m(\overline{x}))] + a(n) + O(\log n) - \mathsf{H}(\mathcal{D}_n^m)
    \\&\leq a(n) + O(\log n),
\end{align*}
yielding the desired closeness between time-bounded and -unbounded Kolmogorov complexity in expectation over $\mathcal{D}_n^m$.

A subsequent work of Goldberg, Kabanets, Lu, and Oliveira \cite{GKLO22} extends the result from \cite{DBLP:conf/focs/HiraharaN21}, obtaining the same conclusion of efficient agnostic learning from the weaker assumption that $\Dist\NP \subseteq \Avg\BPP$: that is, average-case easiness of $\NP$ for \emph{randomised} algorithms. In this setting, one does not have derandomisation (which is available to \cite{DBLP:conf/focs/HiraharaN21} via \cite{DBLP:journals/mst/BuhrmanFP05}), and thus does not have a worst-case to average-case reduction for $\GapMINKT$. Instead, the authors obtain worst-case easiness of a notion of time-bounded \emph{probabilistic} Kolmogorov complexity, $\pK^t$, which generalizes $\K^t$ and $\rK^t$. Given this, the rest of the proof proceeds similarly to \cite{DBLP:conf/focs/HiraharaN21}.

\paragraph{Agnostic Learning from Average Case Kolmogorov Complexity.}

Goldberg and Kabanets \cite{DBLP:conf/coco/GoldbergK23} further weaken the prerequisite assumptions for agnostic learning, exhibiting a direct connection from the average-case easiness of time-bounded (deterministic) Kolmogorov complexity to the existence of agnostic learners. The authors consider average-case easiness of $\K^t$ complexity as well as another complexity measure denoted $\KT$ \cite{DBLP:journals/siamcomp/AllenderBKMR06}, showing that agnostic learners follow from either assumption. $\KT$ complexity is known to be polynomially related to boolean circuit size, bringing the assumption close to a `natural property' in the sense of Razborov and Rudich \cite{DBLP:journals/jcss/RazborovR97}, and thus drawing a connection to a line of work obtaining learners from natural properties, discussed further in \Cref{sec:nw-mem} below. For both $\K^t$ and $\KT$, \cite{DBLP:conf/coco/GoldbergK23} assume average-case easiness for randomised algorithms over the uniform distribution. 

Compared to the settings of \cite{DBLP:conf/focs/HiraharaN21} and \cite{GKLO22}, a few additional challenges arise. One difficulty has to do with the task of showing, as in \cite{DBLP:conf/focs/HiraharaN21}, that $\K^t(\overline{x}) \approx \K(\overline{x})$ with high probability over strings $\overline{x} \sim (\mathcal{D}_n)^m$. As mentioned above, in the prior works, this was proved via coding theorems for time-bounded Kolmogorov complexity. In the case of \cite{DBLP:conf/focs/HiraharaN21}, a coding theorem holds for $\K^t$ under the stronger assumption of average-case easiness of $\NP$. In the case of \cite{GKLO22}, due to \cite{DBLP:conf/icalp/LuOZ22}, a coding theorem holds unconditionally for $\pK^t$. In the setting of \cite{DBLP:conf/coco/GoldbergK23}, though a coding theorem is not available, a similar statement can be shown to hold with high probability in the presence of an additional random string $r$. Namely, for sufficiently large $t$,
\[\Pr_{r \sim \mathcal{U}_{t}}\left[\K^{\poly(t)}(x, r) \leq \log(1/\mathcal{D}_n(x)) + a(n) + |r| + O(\log n)\right] \geq 1 - o(1).\]
By incorporating this additional randomness $r$ into the algorithm for correlative RRHS refutation, one may obtain an agnostic learner as desired.

Consideration of $\KT$ rather than $\K^t$ leads to additional challenges, primarily stemming from the fact that $\KT$ complexity does not admit a worst-case to average-case reduction in the regime of a small approximation gap. On the other hand, the assumed average-case easiness of $\KT$ is still sufficient to learn \emph{over the uniform distribution}. Hence, one method to obtain a learner is to reduce the task of learning over arbitrary distributions in $\P\SAMP/\poly$ to the task of learning over the uniform distribution. This can be accomplished by using an average-case algorithm for $\KT$ as an inverter for a distributional one-way function; more details about this general approach can be found in the discussion of \cite{DBLP:conf/alt/BinnendykCKRS22} below. However, in the case of non-uniformly sampled target distributions, the learner resulting from this argument needs to know the distribution explicitly.\footnote{Unlike the setting of white-box learning discussed in \Cref{sec:wb}, here, the learner only needs the code of a sampler for the marginal distribution over example strings without labels.} \cite{DBLP:conf/coco/GoldbergK23} show that this caveat can be removed by assuming worst-case easiness of computing $\KT$ complexity, employing a more complicated amortization argument to bound the $\KT$ complexity of many concatenated samples from $\mathcal{D}_n$ in lieu of a coding theorem (see also \cite{DBLP:journals/siamcomp/AllenderGMMM18}). It is left open whether the caveat can be removed assuming only average-case easiness of $\KT$.

\subsubsection{Nisan-Wigderson and Membership Queries}\label{sec:nw-mem}

All the results we have covered so far have employed distinguishers for variants of the Direct Product Generator, $G^f_{DP}$, to construct efficient learners. We now shift attention to learners arising from distinguishers for the Nisan-Wigderson Generator, $G^f_{NW}$.

The NW PRG is very well known for its applications in hardness/randomness tradeoffs \cite{DBLP:journals/jcss/NisanW94, DBLP:journals/cc/BabaiFNW93, IW97}. In that context, one instantiates the PRG with a function $f$ known to be hard. The NW PRG then takes the form
\[G^f_{NW}(z) = \left(\; f(z|_{S_1}), \; f(z|_{S_2}), \; ..., \; f(z|_{S_{m}}) \;\right),\]
where $f \colon \bool^n \to \bool$, $|z| = d = O(n^2)$, $m < 2^n$, and $S_1, ..., S_{m} \subseteq [d]$ are ``NW designs'', with $|S_i| = n$ and $|S_i \cap S_j| \leq \log m$ for all $i \neq j \in [m]$. A key idea is that the generator admits a \emph{reconstruction procedure}, $R$. Namely, if $D$ is an efficient algorithm $(1/\poly(m))$-distinguishing $G_f$ from uniform, then $R^{f}(D)$ outputs a $\poly(m)$-size circuit approximating $f$ with advantage $1/\poly(m)$ with high probability. For the purposes of the hardness/randomness paradigm, if we assume in advance that $\poly(m)$-size circuits cannot $(1/2 + \poly^{-1}(m))$-approximate $f$, then we reach a contradiction, and we conclude that the NW PRG is indeed pseudorandom.

\paragraph{Learning from Natural Properties via NW Reconstruction.}  Carmosino et al.~\cite{DBLP:conf/coco/CarmosinoIKK16} provide the first and most influential example of this approach, showing that a distinguisher for $G_{NW}^f$ implies the easiness of PAC learning with membership queries.  As in the previous section, for the purposes of learning, \cite{DBLP:conf/coco/CarmosinoIKK16} apply this reconstruction framework in reverse. The starting point is an assumption of the existence of a \textit{natural property}. Razborov and Rudich introduced the notion of a natural property to pinpoint limitations of existing techniques for proving circuit lower bounds \cite{DBLP:journals/jcss/RazborovR97}. In particular, known proofs of circuit lower bounds were based on natural properties for those circuit classes (including lower bounds for ${\sf AC}^0$ \cite{DBLP:journals/mst/FurstSS84} and ${\sf AC}^0[p]$ for prime $p$ \cite{Razborov1987LowerBO, DBLP:conf/stoc/Smolensky87}). On the other hand, there is no (say, $\P/\poly$-constructive) natural property for general boolean circuits unless cryptographic PRGs do not exist. Very roughly, a natural property for a given circuit class $C$ is equivalent to average-case easiness of the Minimum Circuit Size Problem for $C$: the existence of a poly-time algorithm accepting all truth-tables of low $C$-complexity and rejecting a large enough fraction of truth-tables with high $C$-complexity (see \Cref{def:nat_prop} for the formal details). The assumption of \cite{DBLP:conf/coco/CarmosinoIKK16} is in this sense an assumption of the easiness of meta-complexity. The assumed natural property algorithm becomes a distinguisher for $G^f_{NW}$, and one applies a hybrid argument to obtain a next-bit predictor as specified by the reconstruction procedure. 

There are several differences here compared to the setting of $G_{DP}$, which arise from how the generator makes use of its input. In the case of $G_{DP}$, several ``input'' strings $x_1, ..., x_m$ are given explicitly and independently, and each $x_i$ leads to a distinct label $b_i$ depending on $f$. In the case of $G_{NW}$, the input strings to $f$ are given implicitly by a single seed $z$; ``extracting'' the strings $x_i$ from $z$ requires computing some restrictions based on the designs $S_i$. Due to this more ``economical'' use of the input, $G_{NW}$ can be made to yield a much longer output from a short seed. For this reason, a natural property suffices to distinguish outputs of $G_{NW}$ from uniformly random strings. On the other hand, in the case of $G_{DP}$, inputs of length $n \cdot m$ lead to outputs of length only $n\cdot m + m$, and it is not clear whether outputs of $G_{DP}$ have circuit complexity less than this.

A drawback of $G_{NW}$ compared to $G_{DP}$ is that the reconstruction procedure requires oracle access to the function $f$. This is due to the fact that a given seed $z \sim \mathcal{U}_d$ gives rise to \emph{dependent} input strings $x_1, ..., x_m = z|_{S_1}, ..., z|_{S_m}$. Due to the overlap between sets in the NW design, if one wants to know the value of $f$ on a particular choice of $x$, this imposes restrictions on the values of the other strings $x_j$ appearing in the hybrid. The NW reconstruction procedure handles this with a pre-processing stage, making oracle queries to $f$ to build a table containing $f(x_j)$ for all possible completions of $x_j$ depending on $x$. Ultimately, this results in a learner requiring membership queries.

Regardless of the limitations of learners resulting from natural properties, the connection to the literature on circuit lower bounds makes their construction worthwhile: by the known natural property for ${\sf AC}^0[p]$, Carmosino et al. unconditionally obtain a quasi-polynomial time learning algorithm for this circuit class (see \Cref{fig:cikk-learning-ac0p}).

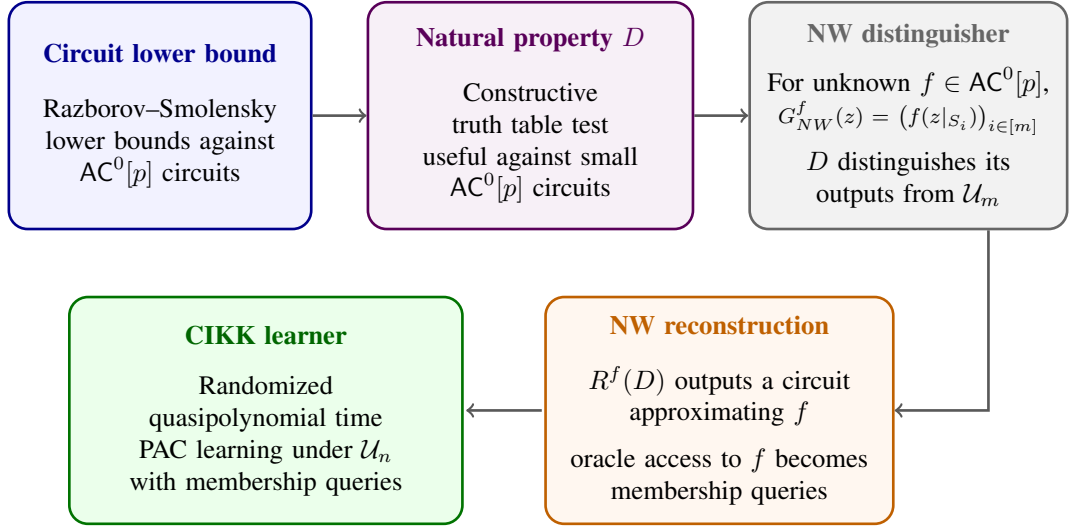
\begin{figure}[ht]
\centering
\begin{tikzpicture}[
    lowerboundbox/.style={
        rounded corners=7pt,
        draw=blue!55!black,
        fill=blue!5,
        line width=1pt,
        minimum width=4.0cm,
        minimum height=3.0cm
    },
    naturalbox/.style={
        rounded corners=7pt,
        draw=violet!70!black,
        fill=violet!6,
        line width=1pt,
        minimum width=4.3cm,
        minimum height=3.0cm
    },
    nwbox/.style={
        rounded corners=7pt,
        draw=gray!80!black,
        fill=gray!10,
        line width=1pt,
        minimum width=4.2cm,
        minimum height=3.0cm
    },
    reconstructionbox/.style={
        rounded corners=7pt,
        draw=orange!75!black,
        fill=orange!6,
        line width=1pt,
        minimum width=4.6cm,
        minimum height=3.0cm
    },
    learnerbox/.style={
        rounded corners=7pt,
        draw=green!45!black,
        fill=green!8,
        line width=1pt,
        minimum width=5.2cm,
        minimum height=3.0cm
    },
    flowarrow/.style={
        ->,
        line width=1pt,
        draw=gray!70!black
    }
]

\node[font=\bfseries\large] at (7,7.55)
{The CIKK learning algorithm for $\mathsf{AC}^0[p]$};

\node[lowerboundbox] (lowerboundstage) at (2.1,5.45) {};
\node[naturalbox] (naturalstage) at (7,5.45) {};
\node[nwbox] (nwstage) at (12,5.45) {};
\node[reconstructionbox] (reconstructionstage) at (9.5,1.55) {};
\node[learnerbox] at (3.5,1.55) {};

\node[font=\small,align=center,text width=3.6cm] at (2.1,5.45)
{\textbf{\textcolor{blue!55!black}{Circuit lower bound}}\\[0.3cm]
Razborov--Smolensky\\
lower bounds against\\
$\mathsf{AC}^0[p]$ circuits};

\node[font=\small,align=center,text width=3.9cm] at (7,5.45)
{\textbf{\textcolor{violet!70!black}{Natural property  $D$}}\\[0.3cm]
Constructive truth table test\\
useful against small\\
$\mathsf{AC}^0[p]$ circuits};

\node[font=\small,align=center,text width=3.8cm] at (12,5.45)
{\textbf{\textcolor{gray!80!black}{NW distinguisher}}\\[0.22cm]
For unknown $f\in\mathsf{AC}^0[p]$,\\
{\scriptsize
$G^f_{NW}(z)=\big(f(z|_{S_i})\big)_{i\in[m]}$
}\\[0.2cm]
$D$ distinguishes its\\
outputs from $\mathcal{U}_m$};

\node[font=\small,align=center,text width=4.2cm] at (9.5,1.55)
{\textbf{\textcolor{orange!75!black}{NW reconstruction}}\\[0.3cm]
$R^f(D)$ outputs a circuit\\
approximating $f$\\[0.2cm]
oracle access to $f$ becomes\\
membership queries};

\node[font=\small,align=center,text width=4.8cm] at (3.5,1.55)
{\textbf{\textcolor{green!40!black}{CIKK learner}}\\[0.25cm]
Randomized quasipolynomial time\\
PAC learning under $\mathcal{U}_n$\\
with membership queries};

\draw[flowarrow]
    (lowerboundstage.east) -- (naturalstage.west);

\draw[flowarrow]
    (naturalstage.east) -- (nwstage.west);

\coordinate (nwstart) at
    ($(nwstage.south west)!0.75!(nwstage.south east)$);

\draw[flowarrow]
    (nwstart) --
    (nwstart |- reconstructionstage.east) --
    (reconstructionstage.east);

\draw[flowarrow]
    (7.15,1.55) -- (6.15,1.55);

\end{tikzpicture}

\caption{\small{The diagram presents the key steps of the $\AC^0[p]$ learning algorithm of Carmosino, Impagliazzo, Kabanets, and Kolokolova \citep{DBLP:conf/coco/CarmosinoIKK16}, a breakthrough result in computational learning theory. The Razborov and Smolensky circuit lower bounds \citep{Razborov1987LowerBO,DBLP:conf/stoc/Smolensky87} yield a natural property \citep{DBLP:journals/jcss/RazborovR97} against $\mathsf{AC}^0[p]$. For an unknown target $f\in\mathsf{AC}^0[p]$, this property distinguishes the output of the Nisan-Wigderson generator $G^f_{NW}$ \citep{DBLP:journals/jcss/NisanW94} from a random truth table. The NW reconstruction procedure converts the distinguisher into a hypothesis using membership query access to $f$. To date, this complexity-theoretic approach remains the only known learning technique for this circuit class.}}
\label{fig:cikk-learning-ac0p}
\end{figure}

In a 2017 follow-up, the authors further show that, from so-called \textit{tolerant} natural properties, one obtains \emph{agnostic} learners \cite{DBLP:conf/approx/CarmosinoIKK17}. A tolerant natural property is essentially equivalent to easiness of the Minimum Approximate Circuit Size Problem, ${\sf MACSP}$, which asks, for a given truth-table $T$ and parameters $s$ and $\delta$, whether there exists a circuit of size $s$ computing $T$ on at least a $1 - \delta$ fraction of inputs. Under this framework, the authors generalize the learning algorithm for ${\sf AC}^0[p]$ to target functions with distance $1/\poly(n)$ from ${\sf AC}^0[p]$.\\

\begin{problem}[Learning {$\AC^0[p]$} from random examples]\label{prob:learn-ac0p-random-examples}
Is there a subexponential time algorithm that learns $\AC^0[p]$ circuits under the uniform distribution using random examples instead of membership queries?
\end{problem}

For recent work related to \Cref{prob:learn-ac0p-random-examples}, we refer to \citep{DBLP:conf/stoc/GurJKRSS24, CJ26}.

\paragraph{From Uniform Learning to Samplable Distributions via Distributional Inversion.}

The approach of \cite{DBLP:conf/coco/CarmosinoIKK16} as described above only yields a PAC learner working on the uniform distribution. A subsequent work due to Binnendyk et al.~\cite{DBLP:conf/alt/BinnendykCKRS22}, following the \emph{indirect learning} approach of \cite{DBLP:conf/colt/FurstJS91}, shows that the result of \cite{DBLP:conf/coco/CarmosinoIKK16} can be extended to obtain learners working over any efficiently samplable distribution, provided (1) that the code of the sampler is known to the learner, and (2) that the target function \emph{composed} with this sampler belongs to the learner's concept class. Below, we restrict our attention to the case that the concept class as well as the sampler complexity of the target distribution are unrestricted polynomial-size circuits.

It is well-known that a natural property would preclude the existence of a PRG with polynomial stretch: outputs of the generator have bounded circuit complexity, while uniformly random strings almost always have maximum circuit complexity \cite{DBLP:journals/jcss/RazborovR97}. Then, by a sequence of results from cryptography \cite{DBLP:journals/jacm/GoldreichGM86, DBLP:journals/siamcomp/HastadILL99, DBLP:conf/focs/Yao82a, DBLP:conf/focs/ImpagliazzoL89}, a natural property implies the existence of efficient distributional inverters for every poly-time function. Roughly, when given the output $y$ of some function $g: \bool^n \to \bool^{m(n)}$, such an inverter outputs a \emph{nearly uniformly distributed} member of $\{x \mid g(x) = y\}$.

Let $A$ be the PAC learner of \cite{DBLP:conf/coco/CarmosinoIKK16} working over the uniform distribution with membership queries. Suppose we would like to learn some target function $f$ over a distribution $D \in \P\SAMP$. The definition of $\P\SAMP$ means that $D$ can be sampled by a polynomial-time machine; i.e. there is a poly-time algorithm $S$ such that $S(r)$ is distributed according to $D$ when the string $r$ is uniformly distributed. Let $I$ be a distributional inverter for $S$.

The key idea is to simulate $A$ on the target function $f(S(-))$. Clearly, knowing $S$ and given membership-query access to $f$, we can simulate membership queries to $f(S(-))$. Then by \cite{DBLP:conf/coco/CarmosinoIKK16}, with high probability, $A$ outputs a hypothesis $h$ such that 
\[\Pr_{r \sim \mathcal{U}_{\poly(n)}}[h(r) = f(S(r))] \geq 1 - 1/\poly(n).\]
We output the hypothesis $h'(-) \eqdef h(I(-))$. 

Intuitively, on input $x \sim D$, we have $h'(x) = h(I(S(r))) \approx h(r)$, where $r \sim \mathcal{U}_{\poly(n)}$. Since $h(r) = f(S(r)) = f(x)$ with high probability, one can show that the hypothesis $h'$ successfully approximates $f$ over $D$, as desired.\\

We refer to \citep{DBLP:journals/cc/Pich25} for additional settings where one can derive learning consequences from complexity lower bound assumptions.

\subsection{Solomonoff Induction: Unsupervised Learning via Kolmogorov Complexity}\label{sec:solomonoff}

In this section, we explain how to efficiently achieve sequence prediction under a meta-complexity assumption about the feasibility of estimating time-bounded Kolmogorov complexity. Formally, we have the following result.

\begin{theorem}[{\citep[Theorem 2.1]{DBLP:journals/eccc/HN25}}]\label{thm:HN25_learning}
If $\mathsf{GapK}^t\text{-}\mathsf{vs}\text{-}\mathsf{K} \in \pr\BPP$, then $(m,n,s,t,\varepsilon,\delta)$-extrapolation can be solved in randomized  time $\poly(n,m,s,t, 1/\varepsilon,1/\delta)$ for $m \geq \widetilde{O}(n^2 \cdot s \cdot \varepsilon^{-6} \cdot \delta^{-5})$.  
\end{theorem}

First, we review Solomonoff induction, which considers this extrapolation task in the time-unbounded setting.

\subsubsection{The Classical Theory}

In Solomonoff induction \citep{Solomonoff1964a, Solomonoff1964b} (see \citep{hutter2024introduction} for a modern introduction), given a prefix
$x_{<i} = (x_1,\ldots,x_{i-1})$, the extrapolation algorithm  predicts the next block $x_i \in \{0,1\}^n$ by \emph{conditioning} a certain ``universal'' distribution on the event that
the generated sequence begins with $x_{<i}$. Equivalently, one gives higher probability weight to continuations that make
the full string $(x_{<i},x_i)$ have small Kolmogorov complexity.\footnote{In this section, we take $\K$ to denote \emph{prefix complexity}, which is the same as our stated definition of Kolmogorov complexity, except it requires that the set of inputs on which the chosen universal machine halts is \textit{prefix-free}. By Kraft's inequality, $\sum_{y} 2^{-\K(wy)} \le 1$ in the case of prefix complexity, whereas the same sum over all finite strings $y$ diverges in the case of plain complexity. Though the ideas in this section can be reformulated to accommodate plain Kolmogorov complexity, we adopt this convention to match the classical literature.} Roughly, we predict the next block $x_i$ given the prefix $x_{<i}$ according to the Solomonoff (semi)measure using 
$$
S(x_i \mid x_{< i}) \propto \sum_y 2^{- \K(x_{<i} x_i y)},
$$
where the symbol $\propto$ means ``proportional to''. This is a
form of \emph{Occam's Razor}, i.e., short explanations are  preferred. It is then possible to show that, for sequences generated by a 
randomized program of description length $s$, the resulting predictor is accurate on most $i \in [m]$ once the number of observations $m \gg s/\varepsilon^2$.

 To show that this predictor works as claimed above, we combine the following ingredients:

 \begin{enumerate}
     \item \emph{Universality of the Solomonoff measure.} If the full sequence $x = (x_1,\ldots,x_m)$ is generated by a randomized
program $M$ of description length  $s$, then the measure $S$ dominates the true source $\mu_M$ up to a multiplicative factor $2^{-O(s)}$, i.e.,
\[
S(x_1,\ldots,x_m) \;\ge\; 2^{-O(s)} \cdot \mu_M(x_1,\ldots,x_m)
\]
for every sequence $(x_1, \ldots, x_m)$.
\item \emph{Log-loss bound and average KL-divergence.}   Taking logarithms and averaging over $x = (x_1,\ldots,x_m)\sim\mu_M$,
one obtains the log-loss bound
\[
\Ex_{x \sim \mu_M}\!\left[
\log \frac{\mu_M(x_1,\ldots,x_m)}{S(x_1,\ldots,x_m)}
\right]
\;\le\; O(s).
\]
The key point is that this global bound can be decomposed into one-step prediction errors. In more detail, for every $x=(x_1,\ldots,x_m)$ with $\mu_M(x)>0$, we can
factor both distributions sequentially:
\[
\mu_M(x) = \prod_{i=1}^m \mu_M(x_i \mid x_{<i}),
\qquad
S(x) = \prod_{i=1}^m S(x_i \mid x_{<i}).
\]
Consequently,
\[
\log \frac{\mu_M(x)}{S(x)}
\;=\;
\sum_{i=1}^m
\log \frac{\mu_M(x_i \mid x_{<i})}{S(x_i \mid x_{<i})}.
\]
Now taking expectation over $x \sim \mu_M$ and applying linearity of expectation, 
\begin{eqnarray}
    \Ex_{x \sim \mu_M}\!\left[\log \frac{\mu_M(x)}{S(x)}\right]
& = & \sum_{i = 1}^m \Ex_{x_{\leq i} \sim \mu_M} \left [ 
\log
\frac{\mu_M(x_i \mid x_{<i})}{S(x_i \mid x_{<i})}
\right]
 \nonumber \\ 
& = &  \sum_{i=1}^m
\Ex_{x_{<i} \sim \mu_M} 
\left[
\sum_{y \in \zo^n}
\mu_M(y \mid x_{<i})
\log
\frac{\mu_M(y \mid x_{<i})}{S(y \mid x_{<i})}
\right]. \nonumber
\end{eqnarray}
But the inner sum is exactly the KL-divergence from the true conditional distribution of the
next block to the Solomonoff posterior. Therefore, combining the equations above,
\[
\sum_{i=1}^m
\Ex_{x_{<i} \sim \mu_M}
\Bigl[
\mathsf{KL}\bigl(\mu_M(\cdot \mid x_{<i}) \,\|\, S(\,\cdot \mid x_{<i})\bigr)
\Bigr]
\;\le\; O(s).
\]
\item \emph{Passing to statistical distance.} Finally, for any large enough $m = O(s \cdot 1/\varepsilon^2 \cdot 1/\delta)$, using a standard Markov argument and  applying Pinsker's inequality, one obtains that the predicted next-block
distribution is $\varepsilon$-close in total variation distance to the true conditional distribution produced by $M$ except for at most a $\delta$-fraction of $i \in [m]$.
 \end{enumerate}
 
Thus, in the time-unbounded setting, Solomonoff induction solves the extrapolation problem once the sequence is sufficiently longer than the description of the underlying generator program.

\subsubsection{The Time-Bounded Setting}\label{sec:time-boundedSol}

We consider a simpler setting where $\delta > 0$ is a small constant, and describe a randomized algorithm that runs in time $\poly(n,m,t, 1/\varepsilon)$ and succeeds as long as $m = m(n,s,\varepsilon)$ is sufficiently large. The proof of \Cref{thm:HN25_learning}, which provides a polynomial-time procedure when $\delta = o(1)$, is significantly more involved. We refer to \citep{DBLP:journals/eccc/HN25} for the details.

To move to the time-bounded world, we replace the  universal distribution by the \emph{time-bounded} universal
distribution $Q_{t,z}$, where $z$ is an auxiliary string.  

\begin{definition}[Time-Bounded Universal Distribution with Advice]
Fix $t \in \N$ and an advice string $z \in \zo^*$. The $t$-time-bounded universal distribution
$\Q_{t,z}$ is the distribution induced by the random variable $U^\mathsf{r}(w,z, 1^t)$, where
$w \sim \mathcal{U}_t$ is chosen uniformly from $\zo^t$. Equivalently, for every string $y \in \zo^*$,
\[
\Q_{t,z}(y)
\;\eqdef\;
\Prob_{w \sim \mathcal{U}_t,\;U^\mathsf{r}}\bigl[\, U^\mathsf{r}(w,z,1^t) = y \,\bigr].
\]
When $z=\epsilon$, we omit the advice string and write $\Q_t$ instead of $\Q_{t,z}$. For convenience, we also define
\[
q^t(y \mid z) \;\eqdef\; - \log Q_{t,z}(y).
\]
\end{definition}

We also need the following definition.

\begin{definition}[Universal Extrapolation Distribution]\label{def:advUE}
The \emph{universal extrapolation distribution}
\[
\mathsf{Next}_n(y;Q_{t,z})
\]
 is the distribution of the next $n$ bits of a $Q_{t,z}$-sample, conditioned on
extending the prefix $y$.
\end{definition}

It is important to distinguish the two roles here: $z$ is an advice string in the
Kolmogorov complexity sense, while $y$ is the observed prefix in the probabilistic conditioning sense. This distinction is a crucial idea from \citep{DBLP:journals/eccc/HN25} to implement Solomonoff's theory in the
time-bounded setting.\footnote{In the classical theory, since computational efficiency is not relevant, there is no need to consider an advice string, and the entire observed prefix is used for conditioning. We stress that in both settings -- time bounded and time unbounded -- moving the entire prefix to the advice portion does not guarantee extrapolation.} Jumping ahead, what is special in the time-bounded setting is that one splits the given prefix into advice and context strings for efficiency reasons and then proves that this restricted-context extrapolation is still correct.

In more detail, prior work already gave a \emph{universal sampling procedure} for  approximately sampling from $\mathsf{Next}_n(y;Q_{t,z})$, but its running time is exponential in the \emph{computational depth} of $y$ given $z$  \citep{DBLP:conf/focs/ImpagliazzoL90, DBLP:conf/focs/HiraharaN23}.\footnote{This algorithm assumes that auxiliary-input one-way functions do not exist, which follows from the hypothesis of \Cref{thm:HN25_learning}.} In other words, consider the computational depth quantity
 \[
 \mathsf{cd}^t(y \mid z) \;\eqdef\; q^t(y \mid z) - \K(y \mid z),
 \]
 where we set $\mathsf{cd}_t(y) \eqdef \mathsf{cd}_t(y \mid \epsilon)$. If one tries to extrapolate directly from the full prefix by setting $y = x_{<i}$ and $z = \epsilon$, the runtime can be exponential in
$\mathsf{cd}_t(x_{<i})$.  For some sequences generated by an $s$-bit randomized program, this depth can be
as large as $\Omega(s)$. 

As alluded to above, a key idea from \citep{DBLP:journals/eccc/HN25} is to reorganize the observed prefix into an \emph{advice part} and a \emph{context part}. For appropriate parameters $w,b \in \mathbb{N}$ with $m = b \cdot w$, the
sequence $(x_1, \ldots, x_m)$ is partitioned into $b$ larger blocks, each consisting of $w$ of the original $n$-bit symbols. When
predicting position $j \in [w]$ inside block $i \in [b]$, the learner uses all previous blocks as advice, but keeps only the
prefix of the current block up to position $j-1$ as context. If we write $\bar{x}_i$ for the $i$-th large block,
$\bar{x}_{<i}$ for the concatenation of the previous blocks, and $\bar{x}_i^{<j}$ for the prefix of the current
block before position $j$, then the predictor outputs according to
\[
\mathsf{Next}_n(\bar{x}_i^{<j};Q_{t,\bar{x}_{<i}}).
\]
In other words, the long-range history is moved into \emph{advice}, while a short local prefix is kept as \emph{context}. 

Next, we explain in more detail how this approach reduces the running time while maintaining correctness of the corresponding predictor. For a formal presentation, see \citep{DBLP:journals/eccc/HN25}.\\

\noindent \emph{I. Bounding the running time.} For the runtime analysis, we need two ingredients: the \emph{slow growth law} and the \emph{chain rule for computational depth}.  

The slow growth law states that if $x \sim \mu_M = U^\mathsf{r}(M,\epsilon,1^t)$, where $M$ is a $t$-time randomized machine of description length $s$, then with high probability
$$
\mathsf{cd}^{\poly(t)}(x) \;\leq\; O(s).
$$
The chain rule states that, for a string $x = (x_1,\ldots, x_{m'})$ with each $|x_i| \leq n'$, and any large enough $t'$, we have
$$
\sum_{i=1}^{m'} \mathsf{cd}^{\poly(t')}(x_i \mid x_{<i}) \;\leq\; \mathsf{cd}^{t'}(x) + m' \cdot O(\log t').
$$
These two results are known to hold under the assumption of \Cref{thm:HN25_learning}, as established by \citep{DBLP:journals/eccc/HN25}. 

Putting the results above together, we get that the expected value of $\mathsf{cd}^{\poly(t)}(x_i \mid x_{<i})$ over a random position $i \sim [m]$ is $O(s/m + \log t) = O(\log t)$. By Markov's inequality,
$$
\Pr_{i \sim [m]}\Big [ \mathsf{cd}^{\poly(t)}(x_i \mid x_{< i}) \;\geq \; \delta^{-1} \cdot O(\log t) \Big ] \;\leq\; \delta.
$$
In other words, for most choices of $i$ and $x$, the time-bounded universal sampling procedure (approximately) samples from $\mathsf{Next}_n(x_i;Q_{t,x_{<i}})$  in time $2^{O(\delta^{-1} \cdot \log t)} = \poly(t)$, for $\delta = \Omega(1)$.\footnote{With low probability over $i$ and $x$, we can abort the algorithm and output ``fail'' if its running time reaches a fixed threshold.} 

One can perform a similar runtime analysis once we break the sequence $x = (x_1, \ldots, x_m)$ into $b$ larger blocks, each consisting of $w$ of the original $n$-bit symbols $x_i$, and use the universal extrapolation procedure to sample instead from $\mathsf{Next}_n(\bar{x}_i^{<j};Q_{t,\bar{x}_{<i}})$.
 (In order to guarantee correctness, the proof employs $w = O(\varepsilon^{-2} \cdot \delta^{-1} \cdot \log t)$, as explained below.) 
~\\

\noindent \emph{II. Correctness of the time-bounded predictor.} Recall that our predictor partitions $x = (x_1, \ldots, x_m)$ into $b$ length-$w$ blocks. Given an index $j \in [w]$ inside block $i \in [b]$, it uses the universal sampling procedure in order to predict according to $\mathsf{Next}_n(\bar{x}_i^{<j};Q_{t,\bar{x}_{<i}})$. By the analysis explained above, for most blocks the sampler is correct and we (approximately) sample from  $\mathsf{Next}_n(\bar{x}_i^{<j};Q_{t,\bar{x}_{<i}})$. It remains to show that this is close to the true distribution, i.e., 
$$
\Pr_{i,\;j,\;(\bar{x}_{<i},\bar{x}_{i}^{<j}) \sim \mu_M} \Big[  \Delta_\mathsf{TV} \big ( \bar{\mathcal{X}}^j_i \mid \bar{x}_{<i} \bar{x}^{< j}_{i}  ,\mathsf{Next}_n(\bar{x}_i^{<j};Q_{t,\bar{x}_{<i}}) \big )  \;\leq\; \varepsilon \Big ] \;\geq\; 1- \delta.
$$

To prove this, let $\bar{\mathcal{X}}_i$ denote the random $i$-th block under $\mu_M$, and let
$\bar{\mathcal{X}}_{<i}$ denote the concatenation of the previous blocks. Up to replacing $t$ by
$\poly(t)$ (which we suppress in the notation), similarly to the chain rule for computational depth described above, it is possible to show that for every $(\bar{x}_1,\ldots,\bar{x}_b)$,
\[
\sum_{i=1}^b q^t(\bar{x}_i \mid \bar{x}_{<i})
\;\leq\;
q^t(\bar{x}_1,\ldots,\bar{x}_b) + b \cdot O(\log t).
\]
On the other hand, the time-bounded universal distribution $Q_t$ dominates every $s$-bit  source $\mu_M$ up to a factor $2^{-O(s)}$. Equivalently,
\[
q^t(\bar{x}_1,\ldots,\bar{x}_b)
\;\leq\;
O(s) - \log \mu_M(\bar{x}_1,\ldots,\bar{x}_b)
\]
for every sequence in the support of $\mu_M$. Combining the two inequalities and expanding the source probability  block-by-block,
\[
\sum_{i=1}^b \Bigl(
q^t(\bar{x}_i \mid \bar{x}_{<i})
+ \log \mu_M(\bar{x}_i \mid \bar{x}_{<i})
\Bigr)
\;\leq\; O(s) + b \cdot O(\log t).
\]
Next, we average over $(\bar{x}_1,\ldots,\bar{x}_b) \sim \mu_M$. For each $i \in [b]$, we claim that the expectation of the $i$-th summand is exactly
\[
\Ex_{\bar{x}_{<i} \sim \mu_M}
\Bigl[
\KL\bigl((\bar{\mathcal{X}}_i \mid \bar{x}_{<i}) \,\|\, Q_{t,\bar{x}_{<i}}\bigr)
\Bigr].
\]
Indeed, by definition, $q^t(y \mid \bar{x}_{<i}) = -\log Q_{t,\bar{x}_{<i}}(y)$, and so
\[
q^t(y \mid \bar{x}_{<i}) + \log \mu_M(y \mid \bar{x}_{<i})
\;=
\log \frac{\mu_M(y \mid \bar{x}_{<i})}{Q_{t,\bar{x}_{<i}}(y)}.
\]
Taking expectation over $y \sim (\bar{\mathcal{X}}_i \mid \bar{x}_{<i})$ yields the desired statement. As a consequence,
\[
\Ex_{i \sim [b]}
\Ex_{\bar{x}_{<i} \sim \mu_M}
\Bigl[
\KL\bigl((\bar{\mathcal{X}}_i \mid \bar{x}_{<i}) \,\|\, Q_{t,\bar{x}_{<i}}\bigr)
\Bigr]
\;\leq\;
\frac{O(s)}{b} + O(\log t).
\]
In other words, the amortized KL-divergence over the $b$ blocks is $O(\log t)$ for $b = \Omega(s/\log t)$. 

The remaining steps follow from standard arguments. First, we apply
the chain rule for KL-divergence to distribute the bound of $O(\log t)$ across the $w$ positions in a
block, yielding an average KL-divergence of $O(\log t)/w$ over positions $j \in [w]$. For $w = \Omega(\varepsilon^{-2} \cdot \delta^{-1} \cdot \log t)$, we obtain a bound of $O(\log t)/w = \varepsilon^2 \cdot \delta$ per symbol. Then, from Markov's inequality, the per-symbol KL-divergence is at most $\varepsilon^2$ for a $(1-\delta)$-fraction of positions inside a block. Finally, the desired bound on statistical distance then follows from Pinsker's inequality.

To sum up, the predictor satisfies the desired guarantees provided that $m = b \cdot w = \Omega((s/\log t) \cdot (\varepsilon^{-2} \delta^{-1} \log t) ) = \Omega(s \cdot \varepsilon^{-2} \cdot \delta^{-1})$.\\

As mentioned above, in order to obtain an efficient predictor when $\delta = o(1)$, a more involved argument seems required, and the known lower bound on $m$ as a function of $\varepsilon$ and $\delta$ is sub-optimal.

\begin{problem}[Optimal Parameters in Extrapolation from Meta-Complexity]\label{prob:extrap_optimal}
Prove that if $\mathsf{GapK}^t\text{-}\mathsf{vs}\text{-}\mathsf{K} \in \pr\BPP$, then $(m,n,s,t,\varepsilon,\delta)$-extrapolation can be solved in randomized  time $\poly(n,m,s,t, 1/\varepsilon,1/\delta)$ for $m \geq \widetilde{O}(s \cdot \varepsilon^{-2} \cdot \delta^{-1})$.
\end{problem}

\subsection{Equivalences Between Learning and Meta-Complexity}\label{sec:learning_meta-equivalences}

In \Cref{sec:learn-via-dist} we described several settings where one can get efficient learning from  suitable distinguishers. In turn, the necessary distinguishers  follow from the existence of efficient algorithms for meta-computational problems. The latter allow one to distinguish randomness from structure, or equivalently, input strings that can be compressed from (a large fraction of) incompressible ones. 

We now discuss two settings where appropriate relaxations of the meta-computational problems are actually \emph{computationally equivalent} to the learning task. Let $\mathcal{C}$ be a class of functions (e.g., polynomial-size Boolean circuits), and let $\Gamma$ be a learning model (e.g., PAC learning under the uniform distribution from random examples). We present equivalences of the following form:
\begin{center}
 $\mathcal{C}$ admits efficient learning algorithms in model $\Gamma$
 $\;\Longleftrightarrow\;$
 The meta-computational problem $L_{\mathcal{C},\Gamma}$ is easy
\end{center}
In the equivalence, $L_{\mathcal{C},\Gamma}$ is a suitable computational problem that we associate to the learning model and concept class.

The idea that the design of learning algorithms is equivalent to the search for efficient procedures that detect patterns in data offers a fundamentally new perspective: while \emph{learning is inherently a sublinear task}, where an algorithm must inspect a small portion of a large object (the unknown function) and produce an approximate representation of it (the output hypothesis), a typical problem from meta-complexity, such as $\mathsf{MCSP}$ and its variants, involves deciding if an input string is ``random'' or ``structured'', a standard algorithmic problem where \emph{the algorithm has access to the entire input}. 

Concretely, as explained below, \citep{Ilango-Loff-Oliveira-CCC20} shows that a circuit class $\mathcal{C}$ can be PAC learned under the uniform distribution from random examples in polynomial time if and only if a variant of $\mathsf{MCSP}$ for partial functions and with respect to $\mathcal{C}$-circuit complexity is easy on average. Similarly, \citep{DBLP:conf/coco/CarmosinoIKK16} can be phrased (see \citep{DBLP:conf/coco/OliveiraS17}) as an equivalence between the efficient learnability of $\mathcal{C}$ with membership queries under the uniform distribution and the existence of a promise-BPP natural property against $\mathcal{C}$ (see \Cref{sec:natural_prop}).

Next, we explain in more detail the aforementioned equivalences. For simplicity, we mostly focus on the case where $\mathcal{C}$ corresponds to general Boolean circuits, though the equivalences hold more generally  under mild assumptions on $\mathcal{C}$ related to their closure properties and basic expressive power.

\paragraph{Learning from Random Examples.}

Ilango, Loff, and Oliveira \cite{Ilango-Loff-Oliveira-CCC20} rephrase the ideas developed in \cite{DBLP:conf/colt/Vadhan17} (see \Cref{sec:direct_product_random_examples}), making more explicit a connection between the easiness of meta-complexity and the possibility of learning. More specifically, the authors show that if a certain problem from meta-complexity, denoted $\pMCSP$, is easy on average, then PAC learning general Boolean circuits on the uniform distribution from random examples is easy. Conversely, a learning algorithm that uses random examples implies the easiness on average of $\pMCSP$.

In \cite{Ilango-Loff-Oliveira-CCC20}, $\pMCSP$ refers to the following formulation of a circuit minimisation problem on partially defined functions.
\begin{quote}

    An instance of the language $\pMCSP$ consists of an integer $n \in \nat$ represented in unary, a collection of input/output pairs $((x_i, b_i))_{i \in [m]}$ for some $m \in \nat$, where $x_i \in \bool^n$ and $b_i \in \bool$ for each $i \in [m]$, and a threshold parameter $s \in \nat$.
    \\
    \\$(1^n, ((x_i, b_i))_{i \in [m]}, s)$ is a yes-instance of $\pMCSP$ iff there exists a Boolean circuit $C$ of size at most $s$ such that $C(x_i) = b_i$ for all $i \in [m]$.
    
\end{quote}
The authors further consider an \emph{average-case} notion of $\pMCSP$, in which an algorithm $A$ is said to solve $\pMCSP$ on average over the uniform distribution with $m$ samples, advantage $\gamma$, and parameter $s$ if for every boolean function $f \in \SIZE[s]$, 
\[\left | \Pr_{\substack{A, \\ x_1, \ldots , x_m \sim \mathcal{U}^m_n}}[A(x_1, ..., x_m; f(x_1), ..., f(x_m)) = 1] - \Pr_{\substack{A, \\ x_1, ..., x_m \sim \mathcal{U}^m_n \\ b_1, ..., b_m \sim \mathcal{U}^m_1}}[A(x_1, ..., x_m; b_1, ..., b_m) = 1] \right | \geq \gamma(n).\]
For large enough $m(n) = \poly(s(n))$ and small enough $\gamma(n) = \Omega(1/\poly(n))$, we remark that this is a rather lenient notion of average-case easiness over the uniform distribution: it is enough to distinguish labels consistent with a function from the class from random labels.

Now observe that the task above is precisely the problem of breaking the Direct Product Generator $G_{DP}^f$ with the input $\overline{x}$ sampled uniformly at random and the target function $f$ belonging to $\SIZE[s]$. Applying the techniques discussed in the context of \cite{DBLP:conf/colt/Vadhan17}, such an algorithm may be transformed into a weak learner for $\SIZE[s]$ working on examples drawn from the uniform distribution. To obtain a strong PAC learner for $\SIZE[s]$ over the uniform distribution, the authors apply a result of \cite{DBLP:conf/colt/BonehL93}  to boost accuracy.

Conversely, it is not difficult to show that an algorithm that learns $\SIZE[s]$ under the uniform distribution from random examples can be used to solve the average-case version of $\pMCSP$ introduced above. Given a sequence $(x_1, \ldots, x_m; \ell_1, \ldots, \ell_m)$, where we assume $m(n)$ is large enough, our task is to decide whether the labels $\ell_1, \ldots, \ell_m$ are \emph{consistent} with a fixed but unknown $f \in \SIZE[s]$ or \emph{random}. We simulate the learner over the sample set $\{(x_1, \ell_1), \ldots, (x_{m-1}, \ell_{m-1})\}$. Let $h$ be its output hypothesis. If $h(x_m) = \ell_m$, we output $1$ (meaning ``consistent''); otherwise we output $0$ (meaning ``random''). On the one hand, if the labels are consistent with some function $f \in \SIZE[s]$, $h$ is likely to correctly label $x_m$ by $\ell_m$, in which case we output $1$. On the other hand, in the random case, no matter the output hypothesis $h$, since the label $\ell_m$ is an independent random bit we get that $h(x_m) = \ell_m$ with probability exactly $1/2$. Therefore, we distinguish the two cases with noticeable probability, i.e., we  solve $\pMCSP$ on average.

\paragraph{Learning from Membership Queries.}  
The preceding discussion used the term natural property informally. The precise randomized formulation relevant here is the notion of a $\pr\BPP$-natural property introduced in \Cref{sec:natural_prop}. Such a property is decided by a randomized algorithm running in time polynomial in the truth-table length $N=2^n$. It accepts with probability at least $2/3$ every truth table in some dense set, while accepting the truth table of each function in $\mathcal{C}[s]$ with probability at most $1/3$. No guarantee is imposed on the remaining truth tables: on these inputs the acceptance probability may lie between $1/3$ and $2/3$, which explains the promise in $\pr\BPP$. Informally, the algorithm efficiently recognizes many unstructured functions while rejecting every function of low $\mathcal{C}$-complexity. This can be seen as a weak average-case procedure for the $\MCSP$ problem for $\mathcal{C}$-circuits.

As discussed above, Carmosino, Impagliazzo, Kabanets, and Kolokolova \cite{DBLP:conf/coco/CarmosinoIKK16} turn such a property into a learner for $\mathcal{C}$ under the uniform distribution: the property distinguishes truth tables produced by the Nisan--Wigderson generator from random truth tables, and the NW reconstruction procedure converts this distinguisher into a hypothesis for the target. The reconstruction uses oracle access to the target, which accounts for the membership queries.

Oliveira and Santhanam \cite{DBLP:conf/coco/OliveiraS17} noted that the reverse implication from learning to natural properties also holds, even if one starts with a barely non-trivial learning algorithm. Concretely, a non-trivial learner is a randomized oracle algorithm $A^f$ that, for some constant $c\geq 1$ and every $f\in\mathcal{C}[s]$, outputs with probability at least $1-1/n$ a hypothesis $h$ satisfying
\[
    \Pr_{x\sim\mathcal{U}_n}[h(x)=f(x)]\geq \frac12+\gamma(n),
    \qquad \gamma(n)=\frac{1}{n^c},
\]
and runs in time $T(n)=2^n/n^{\omega(1)}$. The super-polynomial saving in the denominator is a convenient way to make the definition independent of $c$; a fixed sufficiently large polynomial saving also works.

The idea works as follows. Given $w\in\{0,1\}^{N}$, simulate $A^f$ for $f=\mathsf{fnc}^w$, answering each membership query directly from $w$. Estimate the agreement of the resulting $h$ with $f$ to additive error $\gamma/4$ using polynomially many fresh uniform inputs. Reject if the estimate exceeds $1/2+\gamma/2$, and accept otherwise. If $f\in\mathcal{C}[s]$, the learner succeeds and the test rejects with probability at least $2/3$.

To see why the test accepts most functions, fix a size bound $S\leq T(n)$. By a simple counting argument, there are at most $2^{O(S\log S)}$ Boolean circuits of size $S$. For any fixed circuit $h$ and a uniformly random function $f$, the values $\mathbf{1}[h(x)=f(x)]$ are independent unbiased bits. A Chernoff bound therefore gives
\[
    \Pr_f\left[\Pr_x[h(x)=f(x)]\geq \frac12+\frac{\gamma}{4}\right]
    \leq \exp(-\Omega(\gamma^2 2^n)).
\]
Thus the fraction of functions that can be approximated by a fixed hypothesis $h$ is at most $\exp(-\Omega(\gamma^2 2^n))$. A union bound over all possible $h$ (circuits of size at most $S$) shows that the fraction of functions approximated by any such circuit is at most $\exp(O(S\log S)-\Omega(\gamma^2 2^n))=o(1)$ under the non-trivial time bound. Since every hypothesis output by $A$ has size at most $T(n)$, the test accepts almost every truth table with high probability. It runs in time $T(n)\cdot\poly(n)=\poly(N)$, so it satisfies constructivity, density, and $\mathcal{C}[s]$-hardness from \Cref{sec:natural_prop}. Truth tables outside the two guaranteed cases may have intermediate acceptance probability, which is precisely why the promise class $\pr\BPP$ is appropriate. Together with the learning-from-natural-properties direction, this gives the desired equivalence, up to standard parameter translations.\\

It would be interesting to further develop the dictionary between meta-complexity and learning.

\begin{problem}[Equivalences Between Learning and Meta-Complexity]\label{prob:equiv_learning_meta}
    Can we characterize efficient learnability in other learning frameworks via meta-complexity?
\end{problem}

Next, we discuss a surprising application of the equivalence between learning and meta-complexity.

\subsection{Learning Speedups: Reducing the Running Time via Meta-Complexity}\label{sec:speedup}

A PAC learning algorithm for a class $\mathcal{C}$ of functions comes with three parameters of interest: the \emph{confidence} that it succeeds, the \emph{error} of its output hypothesis, and its \emph{running time}. While it is not difficult to amplify the confidence of a learning algorithm \citep[Chapter 4]{KearnsVazirani:94}, reducing its error is much more subtle. Originally developed to address a question in learning theory, \emph{boosting} (see \citep{BoostingBook}) is a key technique in machine learning that allows one to reduce the error of the output hypothesis of a PAC learner. Thus both the confidence and error parameters can be boosted. Could it be the case that the \emph{running time} of a learning algorithm can also be improved in a generic fashion?

Surprisingly, \citep{DBLP:conf/coco/OliveiraS17} relies on complexity-theoretic techniques \citep{DBLP:conf/coco/CarmosinoIKK16} to show that, in some cases, one can make learning algorithms run much faster! They termed this phenomenon \emph{learning speedups:} a generic approach to improve the \emph{running time} of a learning algorithm. How is this possible? Key to the learning speedup compiler is to explore a fourth dimension of learning algorithms: they tend to solve more difficult tasks than the ones we typically apply them to.

To achieve a learning speedup, we explore the \emph{complexity} of the unknown function $f \colon \{0,1\}^n \to \{0,1\}$ in $\mathcal{C}$ that we are challenged to learn, by \emph{reducing this task} to the learnability of a \emph{more complex} function $g \colon \{0,1\}^m \to \{0,1\}$ over a \emph{smaller} number $m$ of input bits. The rationale is simple: on an instance of the learning problem of length $n$, if we run the original (slow) learning algorithm on a problem of much smaller dimension $m \ll n$, we obtain a faster learning algorithm as a function of $n$. For this to succeed, we need to guarantee that $g$, while more complex than $f$, is still in $\mathcal{C}$ (since the original learner is only promised to learn functions in $\mathcal{C}$), and we must discover a way to implement the reduction between learning problems. To achieve the latter, we  \emph{implicitly} simulate the original learning algorithm on $g$ (while only having access to $f$), and show that a good output hypothesis for $g$ can be transformed into a good hypothesis for $f$.\footnote{At the technical level, this idea is implemented in a more sophisticated fashion, by reducing the learnability of $f$ to the task of learning a collection $\{g_i\}_{i \in \mathcal{S}}$ of functions, in a way that the smaller functions $g_i$ collectively encode enough information that allows us to approximate $f$ using approximators to randomly sampled functions $g_i$ that are implicitly constructed from $f$.}

Here we explain the proof of the following speedup result.

\begin{theorem}[From Nontrivial Learning to Subexponential Time Learning \citep{DBLP:conf/coco/OliveiraS17}]
    Let $\mathcal{C}$ be a circuit class at least as powerful as bounded-depth circuits with parity gates. Suppose that, for every constant $k \geq 1$, $\mathcal{C}$-circuits of size $n^k$ can be PAC learned under the uniform distribution with membership queries in time $\leq 2^{n}/n^{\omega(1)}$. Then, for every $k \geq 1$ and $\varepsilon > 0$, $\mathcal{C}$-circuits of size $n^k$ can be PAC learned under the uniform distribution with membership queries in time $O(2^{n^\varepsilon})$.
\end{theorem}

\begin{proof}[Proof Sketch.]
Fix $k\geq 1$ and $\varepsilon>0$. The argument is the composition of the two directions described in \Cref{sec:learning_meta-equivalences}. First, the assumed non-trivial learner yields a $\pr\BPP$-natural property (\Cref{def:nat_prop}). Indeed, on the truth table of an $m$-variate function $g$, we simulate the learner using the truth table to answer its membership queries, and then test whether its output hypothesis has noticeable agreement with $g$. The learner succeeds when $g$ is computed by a small $\mathcal{C}$-circuit, whereas the counting argument from the preceding section shows that almost every function cannot be noticeably approximated by a circuit of size $2^m/m^{\omega(1)}$. Reversing the output of this test therefore gives, for every constant $d$, a $\pr\BPP$-natural property against $\mathcal{C}[m^d]$. Its running time is polynomial in the truth-table length $2^m$, and hence is $2^{O(m)}$.

We now apply the learning-from-natural-properties construction discussed in \Cref{sec:learn-via-dist}. Let $f\in\mathcal{C}[n^k]$ be the target function, and choose a parameter $m=n^{\gamma}$, where $\gamma>0$ will be fixed below. The black-box Nisan--Wigderson generator associates with $f$ a family of auxiliary functions $g_z\colon\{0,1\}^m\to\{0,1\}$. Since $\mathcal{C}$ contains bounded-depth circuits with parity gates and has the standard closure properties, each $g_z$ is computed by a $\mathcal{C}$-circuit of size $m^d$, for some constant $d=d(k,\gamma)$. This is the reason the theorem assumes non-trivial learners for $\mathcal{C}[n^d]$ for every constant $d$, rather than only for one fixed size bound.

The natural property against $\mathcal{C}[m^d]$ rejects every truth table $\mathsf{tt}(g_z)$ while accepting a dense set of random truth tables. It therefore distinguishes the NW outputs associated with $f$ from a uniformly random function on $m$ bits. The NW reconstruction procedure converts this distinguisher into a hypothesis that approximates $f$ under the uniform distribution. As explained earlier, the reconstruction can form the required truth-table entries of the functions $g_z$ using oracle access to $f$; these become the membership queries of the resulting learner.

The natural property and the reconstruction run in time polynomial in $2^m$, so the total running time is $2^{O(m)}\cdot\poly(n)=2^{O(n^{\gamma})}$. Choose $\gamma<\varepsilon$ small enough to absorb the constant hidden in the exponent. Standard accuracy and confidence amplification then gives a PAC learner running in time $O(2^{n^\varepsilon})$, as desired.
\end{proof}

As explained above, the key conceptual insight behind the learning speedup compiler is that to learn a function $f \in \mathcal{C}[n^k]$ in sub-exponential time, we invoke the nontrivial learner on concepts of size $m^{\ell}$, where $\ell$ is larger than $k$ but the input length $m$ is smaller than $n$. Crucially, we save  time by using that the original learner is able to succeed on relatively harder learning tasks of smaller instance length. The tradeoff is considerably more general. In particular, a nontrivial learner that learns exponential size concepts can be transformed into a polynomial-time learner for polynomial-size concepts \citep{DBLP:conf/coco/OliveiraS17}.

Despite its obvious relevance and potential for applications, this intriguing phenomenon is currently restricted to PAC learning algorithms that work under the uniform distribution and can make membership queries. It also requires the class $\mathcal{C}$ to be reasonably expressive, since this is needed to implement the  techniques from \citep{DBLP:conf/coco/CarmosinoIKK16}.  

\begin{problem}[Learning Speedups in Different Learning Frameworks]\label{prob:learning_speedup}
    Systematically investigate learning speedups in different learning models. In
particular, is there a generic speedup for learning under the uniform distribution from random examples?
\end{problem}

\newpage

\section{Learning vs Compression via Shannon's Information Theory}\label{sec:shannon_connections}

We have seen several connections between learning and compression through the lens of meta-complexity and algorithmic information theory.  In this section, we explore connections from the perspective of Shannon's information theory, taking the theory of source coding as our point of departure. As we will see, the relationship between computational learning and meta-complexity can be understood as a complexity-theoretic counterpart of the classical relationship between learning and universal coding. Our goal is to review several tools for establishing such connections and to explain how notions from meta-complexity fit into this framework. We also briefly discuss how this perspective relates to language modeling, one of the central paradigms of modern machine learning.

\subsection{Starting from Shannon's Information Theory}\label{sec:information-theory}

Shannon's information theory~\cite{Shannon1948} provided a precise mathematical framework for quantifying ``information'' and characterized what can and cannot be achieved in data compression and communication.
Originally developed to address problems in communication engineering, the theory laid the foundations for the mathematical study of information and has profoundly influenced theoretical computer science.

One of the central results of Shannon's theory, and the one most relevant to our discussion, is the noiseless source coding theorem.
In the setting of this theorem, a message $x$ is drawn from a fixed distribution $\D$ over a finite message space $\mathcal{X}$, and $\D$ is known to both the encoder and the decoder.
The goal is to encode $x$ as a short binary string from which $x$ can be uniquely recovered.
By fixing a canonical binary representation, we may assume without loss of generality that $\mathcal{X}=\binset^n$ for some $n$.

Shannon's construction begins by ordering the messages in $\supp(\D)$ as $x_1,\ldots,x_m$ so that $p_1\geq p_2\geq\cdots\geq p_m$, where $p_i\eqdef\D(x_i)$.
For each $i$, let
\[
\ell_i\eqdef\left\lceil\log\frac{1}{p_i}\right\rceil
\qquad\text{and}\qquad
P_{<i}\eqdef\sum_{j<i}p_j.
\]
The encoder $\Enc$ maps $x_i$ to the first $\ell_i$ bits after the binary point in the binary expansion of $P_{<i}$. One can verify that the original message is uniquely determined by its encoding; we omit the argument.\footnote{For an introduction to data compression, we refer to \citep{blelloch2001introduction}.}

By construction, the expected codeword length satisfies
\[
\begin{aligned}
\Exp_{x\sim\D}\bigl[|\Enc(x)|\bigr]
&=\sum_{x\in\supp(\D)}\D(x)\left\lceil\log\frac{1}{\D(x)}\right\rceil\\
&<\sum_{x\in\supp(\D)}\D(x)\left(\log\frac{1}{\D(x)}+1\right)\\
&=\Ent{\D}+1.
\end{aligned}
\]

Notably, Shannon also established the converse, showing that entropy gives a lower bound on the expected codeword length of any reasonable coding.
In particular, every uniquely decodable\footnote{An encoding is uniquely decodable if every finite concatenation of codewords admits a unique decomposition into codewords. Shannon did not state his result using this notion; the corresponding formulation for uniquely decodable codes was later established by McMillan \cite{McMillan1956}.} binary encoding $C\colon\supp(\D)\to\binset^*$ satisfies
\[
\Exp_{x\sim\D}\bigl[|C(x)|\bigr]\geq\Ent{\D}.
\]
In this sense, Shannon's information theory \emph{solves} the compression problem, provided that the source distribution is known to both the encoder and the decoder.
 
For data arising in the real world, however, the source distribution is rarely given explicitly.
For instance, we encounter many different kinds of data in everyday life, including text, music, and images.
To apply Shannon's construction to all such data using a single encoder, we would first have to specify a distribution over the possible inputs.
But can we identify a natural distribution that captures the diverse forms of data generated by human activity?
Even in a much narrower setting, such as compressing an email written by a professor, the resulting code may differ substantially depending on whether we model the message as drawn from a distribution over all human-written text or from one restricted to academic writing.
How, then, should we specify the underlying source distribution?
Should a description of the distribution itself be included as part of the encoding?

\subsubsection{Universal Coding and Kolmogorov Complexity as Its Ultimate Goal}

The problem of compressing strings without knowing the underlying source distribution is known as \emph{universal coding} and dates back to the early works \cite{kolmogorov1965three,Fitingof1966OptimalCoding,Davisson1973UniversalNoiselessCoding}. For instance, Lempel--Ziv algorithms are prominent examples \cite{DBLP:journals/tit/ZivL77,DBLP:journals/tit/ZivL78}. For every stationary ergodic source (roughly, a source whose statistical behaviour is stable over time and can be inferred from a single sufficiently long sample sequence), their amortized code length converges to the entropy rate without requiring an explicit description of the source.
Such guarantees are source-universal but do not capture every simple individual sequence. For example, consider the infinite sequence obtained by concatenating the binary representations
\[ 0,1,10,11,100,\ldots \]
of the nonnegative integers. The sequence has a short generating program, but LZ78 does not asymptotically compress it.

Kolmogorov complexity and its time-bounded variants, which are central notions in meta-complexity, can be seen as the ultimate goal of universal coding.
Recall that the Kolmogorov complexity $\K(x)$ of a string $x$ is the length of the shortest program that outputs $x$ and thus measures the length of its shortest algorithmic description without presupposing any source distribution.
In this sense, $\K(x)$ represents the ideal compression length of $x$ in the absence of computational constraints.
Its time-bounded variant $\K^t(x)$ additionally requires the program to output $x$ within $t$ steps and thus provides the corresponding ideal when the decoding time is also constrained.

\subsubsection{Universal Coding and Learning}\label{sec:codinglearningstat}

We now relate universal coding to distribution learning.
Both problems concern an unknown source distribution: universal coding seeks to compress its outputs without knowing the distribution in advance, whereas distribution learning uses samples to approximate it.

\paragraph{Learning $\Rightarrow$ Universal Coding.}
In one direction, suppose that the unknown source distribution is $\D$ over $\binset^n$, and that we have a distribution learner for a class containing $\D$.

Let us first consider the case in which the learning guarantee is measured in KL divergence. Given samples from $\D$, suppose that the learner outputs a distribution $\mathcal{E}$ satisfying $\KL(\D\|\mathcal{E}) \leq \epsilon$.
We can then use $\mathcal{E}$ as the coding distribution in place of $\D$. A Shannon code based on $\mathcal{E}$ assigns each string $x$ a codeword of length at most $\log \frac{1}{\mathcal{E}(x)} + 1$.
Ignoring this standard one-bit rounding overhead, the excess expected code length over the entropy of $\D$ is exactly the KL divergence:
\[
\Exp_{x\sim\D}\left[\log\frac{1}{\mathcal{E}(x)}\right] - \Ent{\D} =
\KL(\D\|\mathcal{E}) \le \epsilon.
\]
Thus, distribution learning in KL divergence immediately yields a code whose expected length is close to the optimal expected length for $\D$.

Now consider the case in which the learning guarantee is measured in total variation distance, as in \Cref{def:distlearning}. Suppose that, given samples from $\D$, the learner outputs a distribution $\mathcal{E}$ satisfying $\TV(\D,\mathcal{E}) \leq \epsilon$.
Again, we would like to use $\mathcal{E}$ in place of $\D$ when constructing a Shannon code. However, closeness in total variation distance alone does not imply any bound on $\KL(\D\|\mathcal{E})$, which is precisely the expected excess ideal code length incurred by using $\mathcal{E}$ instead of $\D$. Indeed, the KL divergence can be infinite if $\mathcal{E}(x)=0$ for some $x$ satisfying $\D(x)>0$.

The following proposition provides a substitute for such a KL-divergence bound. Intuitively, it shows that, with high probability over $x\sim\D$, the ideal code length assigned by $\mathcal{E}$ is not much larger than that under the true distribution $\D$. We will use it to derive two consequences: a bound on the expected length of a Shannon code and a comparison with Kolmogorov complexity.

\begin{proposition}\label{prop:infexchange}
Let $\D$ and $\mathcal{E}$ be distributions over a finite set satisfying $\TV(\D,\mathcal{E})\leq\epsilon$.
Then, for every $\delta\in(0,1]$,
\[
\Pr_{x\sim\D}\left[
\log\frac{1}{\mathcal{E}(x)}
\leq
\log\frac{1}{\D(x)}+\log\delta^{-1}
\right]
\geq 1-\delta-\epsilon.
\]
\end{proposition}

\begin{proof}
Let
\[
B\eqdef\{x:\mathcal{E}(x)<\delta\D(x)\}.
\]
By the definition of $B$, we have $\mathcal{E}(B)<\delta\D(B)\leq\delta$.
Since total variation distance bounds the difference in probability assigned to any event,
\[
\D(B)\leq\mathcal{E}(B)+\TV(\D,\mathcal{E})\leq\delta+\epsilon.
\]
For every $x\in \bar{B} \cap \supp(\D)$, we have $\mathcal{E}(x)\geq\delta\D(x) (>0)$, which is equivalent to 
\[\log\frac{1}{\mathcal{E}(x)}
\leq
\log\frac{1}{\D(x)}+\log\delta^{-1}.\]
This completes the proof of the proposition.
\end{proof}

As a first consequence, we show that the expected length of a code based on $\mathcal{E}$ is close to the entropy of $\D$. In constructing the code, we mix $\mathcal{E}$ with the uniform distribution to avoid zero probabilities, which could otherwise make the KL divergence infinite.

\begin{corollary}\label{cor:statl2c}
Let $\D$ and $\mathcal{E}$ be distributions over $\binset^n$ satisfying $\TV(\D,\mathcal{E})\leq\epsilon$.
Define
\[
\bar{\mathcal{E}}\eqdef\frac{1}{2}\mathcal{E}+\frac{1}{2}\mathcal{U}_n.
\]
Then,
\[
\Exp_{x\sim\D}\left[\log\frac{1}{\bar{\mathcal{E}}(x)}\right]
\leq
\Ent{\D}+\epsilon n+O(1),
\]
where $O(1)$ is a universal constant.
\end{corollary}

\begin{proof}
For each $x\in \supp(\D)$, let $Z(x)\eqdef\log(1/\bar{\mathcal{E}}(x))-\log(1/\D(x))$ and $Z_+(x)\eqdef\max\{Z(x),0\}$.
Since $\bar{\mathcal{E}}(x)\geq2^{-n-1}$, we have $Z(x)\leq n+1$.
Moreover, since $\bar{\mathcal{E}}(x)\geq\mathcal{E}(x)/2$, for every integer $k\geq0$,
\[
\Pr_{x\sim\D}[Z(x)>k+1]
\leq
\Pr_{x\sim\D}\left[
\log\frac{1}{\mathcal{E}(x)}
>
\log\frac{1}{\D(x)}+k
\right]
\leq
2^{-k}+\epsilon,
\]
where the last inequality follows from \Cref{prop:infexchange} with $\delta=2^{-k}$.
Summing these tail bounds gives
\[
\Exp_{x\sim\D}[Z_+(x)]
\leq
1+\sum_{k=0}^{n}(2^{-k}+\epsilon)
\leq
\epsilon n+O(1).
\]
Since $\Exp_{x\sim\D}[\log(1/\bar{\mathcal{E}}(x))]=\Ent{\D}+\Exp_{x\sim\D}[Z(x)]$, the result follows.
\end{proof}

As a second consequence, we show that, with high probability, the ideal code length assigned by $\mathcal{E}$ is close to the Kolmogorov complexity, i.e., its idealized description length.

\begin{corollary}\label{cor:statl2calg}
Let $\D$ and $\mathcal{E}$ be distributions over $\binset^n$ satisfying $\TV(\D,\mathcal{E})\leq\epsilon$.
Then, for every $\delta\in(0,1]$,
\[
\Pr_{x\sim\D}\left[
\log\frac{1}{\mathcal{E}(x)}
\leq
\K(x)+O\!\left(\log n+\log\delta^{-1}\right)
\right]
\geq 1-\delta-\epsilon.
\]
\end{corollary}

\begin{proof}
For every $s>0$,
\[
\begin{aligned}
\Pr_{x\sim\D}\left[\log\frac{1}{\D(x)}>\K(x)+s\right]
=
\sum_{\substack{x\in\binset^n:\\\D(x)<2^{-\K(x)-s}}}\D(x)
\leq
2^{-s}\sum_{x\in\binset^n}2^{-\K(x)}.
\end{aligned}
\]
For each $k \geq 0$, by counting the number of programs of length $k$, there are at most $2^k$ strings of Kolmogorov complexity $k$. Moreover, for any string $x$ of length $n$, $\K(x) \leq n + O(\log n)$. Therefore, the expression above is upper-bounded by 
\[2^{-s}\cdot\sum_{k = 0}^{n + O(\log n)}\sum_{\substack{x \in \bool^n : \\ \K(x) = k}}2^{-k} = 2^{-s}\cdot\sum_{k = 0}^{n + O(\log n)} O(1) = 2^{-s}\cdot O(n).\]

Thus, by choosing $s=O(\log n+\log\delta^{-1})$, we obtain $\log(1/\D(x))\leq\K(x)+O(\log n+\log\delta^{-1})$ except with probability at most $\delta/2$.

By \Cref{prop:infexchange} with parameter $\delta/2$, we also have
\[
\log\frac{1}{\mathcal{E}(x)}
\leq
\log\frac{1}{\D(x)}+\log\frac{2}{\delta}
\]
except with probability at most $\delta/2+\epsilon$.
Combining the two inequalities and taking a union bound proves the claim.
\end{proof}

In particular, if $\mathcal{E}$ admits a short description, then incorporating this description into the code yields, for most $x\sim\D$, an encoding of $x$ whose length is close to optimal.

\paragraph{Universal Coding $\Rightarrow$ Learning.}
Conversely, any code $C\colon \mathcal{X}\to\bin^*$ for an unknown source distribution $\D$ over $\mathcal{X}$, with code-length function $c(x)\eqdef |C(x)|$, naturally induces a distribution $\D_c$ by normalizing $2^{-c(x)}$:
\begin{equation}\label{eq:defDc}
    \D_c(x)\eqdef\frac{2^{-c(x)}}{Z_c},
\qquad
Z_c\eqdef\sum_{y\in\mathcal{X}}2^{-c(y)}.
\end{equation}
This construction is natural since the ideal code length
$c(x)=-\log \D(x)$, with the convention $-\log 0=\infty$,
yields $\D_c=\D$. It is therefore natural to ask how close $\D_c$ is to the true distribution $\D$ when $c(x)$ only approximates the ideal code length.

The following proposition answers this question. It shows that the KL divergence from $\D$ to $\D_c$ is determined by the excess expected code length over the entropy of $\D$, together with the normalization term $\log Z_c$.

\begin{proposition}\label{prop:compression-to-distribution}
Let $\D$ be a distribution over $\mathcal{X}$, and let $C\colon \mathcal{X}\to\bin^*$ be a code with code-length function $c(x)\eqdef |C(x)|$. Then
\[
\KL(\D\Vert\D_c) =
\left(
\Exp_{x\sim\D}[c(x)]-\Ent{\D}
\right)
+\log Z_c.
\]
\end{proposition}

\begin{proof}
By the definition of $\D_c$, for every $x\in\supp(\D)$,
\[
\log\frac{1}{\D_c(x)}
=
\log\frac{Z_c}{2^{-c(x)}}
=
c(x)+\log Z_c.
\]
Therefore,
\[
\begin{aligned}
\KL(\D\Vert\D_c)
&=
\Exp_{x\sim\D}\left[
\log\frac{\D(x)}{\D_c(x)}
\right]\\
&=
\Exp_{x\sim\D}\left[
\log\frac{1}{\D_c(x)}
\right]
-\Ent{\D}\\
&=
\Exp_{x\sim\D}[c(x)]
-\Ent{\D}
+\log Z_c.
\end{aligned}
\]
\end{proof}

Consequently, if $\Exp_{x\sim\D}[|C(x)|]\leq\Ent{\D}+\epsilon$, then
\[
\KL(\D\Vert\D_c)\leq\epsilon+\log Z_c.
\]
Moreover, if $C$ is a prefix-free code on $\mathcal{X}$, meaning that no codeword is a prefix of another, then Kraft's inequality gives
\[
Z_c = \sum_{y\in\mathcal{X}}2^{-|C(y)|} \leq 1.
\]
Thus, $\log Z_c\leq 0$, and hence
\[
\KL(\D\Vert\D_c)\leq\epsilon.
\]
Applying Pinsker's inequality, we further obtain the total variation bound
\[
\TV(\D,\D_c)
\leq 
O(\sqrt{\epsilon}).
\]

Although the preceding proposition gives a direct transformation from universal coding to distribution learning, its hypothesis may appear somewhat stringent. The entropy $\Ent{\D}$ is the minimum expected code length attainable when the source distribution $\D$ is known, whereas a universal code must be constructed without knowledge of $\D$. It is therefore natural to ask what can be concluded when $c(\cdot)$ instead approximates the Kolmogorov complexity $\K(\cdot)$, viewed as an idealized universal code length.

The following proposition answers this question. It gives a similar bound at the cost of an additive term $|\D|$, where $|\D|$ denotes the length of the shortest description of a probabilistic program that samples exactly from $\D$ (with no restriction on its running time).

\begin{proposition}\label{prop:kolmogorov-to-distribution}
Let $\D$ be a distribution over a finite set $\mathcal{X}$, and let $|\D|$ denote the length of the shortest description of a probabilistic program that samples exactly from $\D$. Let $C\colon\mathcal{X}\to\binset^*$ be a code, and let $c(x)\eqdef |C(x)|$. Then
\[
\KL(\D\Vert\D_c) \leq \Exp_{x\sim\D}\bigl[c(x)-\K(x)\bigr] + |\D| + \log Z_c + O(1).
\]\end{proposition}

\begin{proof}
By the coding theorem,
\[
\K(x)\leq\left\lceil\log\frac{1}{\D(x)}\right\rceil+|\D|+O(1)
\]
for every $x\in\supp(\D)$.
Therefore,
\[
\Exp_{x\sim\D}[\K(x)]\leq\Ent{\D}+|\D|+O(1).
\]
From \Cref{prop:compression-to-distribution}, we obtain
\[
\begin{aligned}
\KL(\D\Vert\D_c)
&=
\Exp_{x\sim\D}[c(x)]-\Ent{\D} + \log Z_c\\
&=
\Exp_{x\sim\D}[c(x)-\K(x)] + \Exp_{x\sim\D}[\K(x)]-\Ent{\D} + \log Z_c\\
&\leq
\Exp_{x\sim\D}[c(x)-\K(x)] + |\D| + \log Z_c  + O(1).
\end{aligned}
\]
\end{proof}

Again, if $C$ is prefix-free and $\Exp_{x\sim\D}[c(x)-\K(x)]\leq\epsilon$, then
\[
\KL(\D\Vert\D_c)\leq|\D|+\epsilon+O(1).
\]
However, since $|\D|\ge1$ this bound does not yield a small total variation distance via Pinsker's inequality.

One might hope to make the contribution of $|\D|$ negligible so that
Pinsker's inequality yields a small total variation distance. This can be
achieved by compressing a batch of independent samples and thereby
amortizing $|\D|$ over the batch. Suppose that a prefix-free code
$C_m\colon\mathcal{X}^m\to\binset^*$ satisfies
\[
\Exp_{\bar{x}\sim\D^m}
\left[
|C_m(\bar{x})|-\K(\bar{x})
\right]
\leq
\epsilon_m,
\]
and let $\D_{c_m}$ be the distribution induced by the length function $c_m(\bar{x})\eqdef |C_m(\bar{x})|$.

An exact sampler for $\D^m$ simply runs an exact sampler for $\D$
independently $m$ times. It therefore has description length
\[
|\D^m|\leq|\D|+O(\log m).
\]
Consequently, \Cref{prop:kolmogorov-to-distribution} gives
\[
\KL(\D^m\Vert\D_{c_m})
\leq
|\D|+\epsilon_m+O(\log m).
\]

To turn this approximation of $\D^m$ into an approximation of $\D$,
choose $I$ uniformly at random from $[m]$, draw
$X_{<I}\sim\D^{I-1}$, and let $\widehat{\D}$ be the conditional distribution of the $I$-th coordinate under $\D_{C_m}$, conditioned on $I-1$ independent samples from $\D$.

By the chain rule for KL divergence (\Cref{lem:chainruleKL}), 
\[
\begin{aligned}
\Exp_{I,X_{<I}}
\left[
\KL(\D\Vert\widehat{\D})
\right]
&\leq
\frac{|\D|+\epsilon_m+O(\log m)}{m}.
\end{aligned}
\]
Here, the expectation is taken over both the uniformly random index $I$
and the prefix $X_{<I}\sim\D^{I-1}$. 
Thus, provided that
$\epsilon_m=o(m)$, the expected KL divergence tends to zero as
$m\to\infty$. Pinsker's inequality and Jensen's inequality then imply
\[
\begin{aligned}
\Exp_{I,X_{<I}}
\left[
\TV(\D,\widehat{\D})
\right]
&\leq
\sqrt{
\frac{1}{2}
\Exp_{I,X_{<I}}
\left[
\KL(\D\Vert\widehat{\D})
\right]
}\\
&\leq
\sqrt{
\frac{1}{2}
\cdot
\frac{|\D|+\epsilon_m+O(\log m)}{m}
}
\xrightarrow[m\to\infty]{}0.
\end{aligned}
\]

\paragraph{Conclusion.} Taken together, these elementary arguments give information-theoretic transformations in both directions.
A distribution learner yields a code whose expected length is close to the source entropy, while a code approximating Kolmogorov complexity can be converted into a learner for distributions having short sampler descriptions.
These transformations do not, however, address computational efficiency.
Kolmogorov complexity is uncomputable, and evaluating or sampling from the conditional distribution $\widehat{\D}$ need not be efficient.
This leads to the corresponding complexity-theoretic question: for efficiently samplable distributions, how is efficient learning related to efficient compression near time-bounded Kolmogorov complexity?
Efficient learning belongs to computational learning theory, whereas efficiently approximating time-bounded Kolmogorov complexity is a central problem in meta-complexity.

\newcommand{\q}{\mathrm{q}}
\subsection{Computational Learning and Meta-Complexity}
\label{sec:computational-learning-meta-complexity}

As we saw in \Cref{sec:information-theory}, distribution learning is closely related to compression, particularly universal coding, at the information-theoretic level.
The relationship between efficient distribution learning and meta-complexity can be viewed as a complexity-theoretic counterpart of this connection.
The natural question is how much of this connection survives when the transformations between learning and compression must be efficient and code lengths are measured against \emph{polynomial-time-bounded} Kolmogorov complexity.

More specifically, we consider both directions of this question.
Can an efficient learner be converted into an efficient encoder that produces short time-bounded descriptions of samples from the target distribution?
Conversely, can the ability to find short time-bounded descriptions be converted into an efficient learner?
The first direction requires turning an efficiently samplable distribution into a code, whereas the reverse direction requires turning description lengths into a probability distribution that can be efficiently sampled from and, in some settings, conditioned.
The purpose of this section is to survey how these two transformations are connected to meta-complexity.

For simplicity, we informally use $\K^{\poly}$ to denote $\K^t$ for an unspecified polynomial time bound $t$, and use $\rK^{\poly}$ and $\q^{\poly}$ analogously. 
We leave the particular polynomials implicit and do not track polynomial overheads in the time bounds.

\subsubsection{Learning to Compression via Arithmetic Coding}
\label{sec:learning-to-compression-arithmetic}

In the learning-to-compression transformation from \Cref{sec:codinglearningstat}, the learner produces a hypothesis distribution $\mathcal{E}$.
Applying Shannon coding to the smoothed distribution $\bar{\mathcal{E}}$ from \Cref{cor:statl2c} yields an expected description length close to the entropy of the source.
This argument, however, is information-theoretic.
The Shannon code described in \Cref{sec:information-theory} first orders all messages and then encodes each message using the cumulative probabilities of the preceding messages.
For the message space $\binset^n$, carrying out this construction directly may require examining exponentially many strings.
Moreover, a distribution learner may output only a sampler for $\mathcal{E}$, rather than an algorithm that evaluates $\mathcal{E}(x)$.
Given a sampling circuit for $\mathcal{E}$, computing $\mathcal{E}(x)$ amounts to counting the random seeds on which the circuit outputs $x$, and a brute-force computation may take exponential time.
Thus, the information-theoretic argument alone does not yield a polynomial-time encoder and decoder.
Indeed, this computational gap is inherent for general efficiently samplable sources: under the plausible assumption that $\mathsf{SZK}\neq\BPP$, some such sources cannot be compressed close to their entropy by any polynomial-time encoder~\cite{DBLP:journals/cc/TrevisanVZ05}.

One way to avoid this global enumeration is to require the learner to output a \emph{next-bit predictor}, rather than merely a sampler, and then apply \emph{arithmetic coding}.
For a distribution $\D$ over $\binset^n$, write
\[
\D(b\mid x_{<i})
\eqdef
\Pr_{X\sim\D}[X_i=b\mid X_{<i}=x_{<i}].
\]
A next-bit predictor approximates these conditional probabilities to inverse-polynomial additive accuracy.
This requirement is a quantitative strengthening of Yao's next-bit prediction criterion, which asks only for a non-negligible advantage over random guessing in predicting the next bit at a uniformly random position from the preceding prefix (see \Cref{sec:yao}). 
Given a next-bit predictor  satisfying this stronger quantitative guarantee, arithmetic coding accesses the learned distribution one bit at a time and constructs the desired encoder and decoder without enumerating its support.

We briefly review how arithmetic coding works.
At a high level, it reconstructs the cumulative-probability interval associated with a message using only next-bit conditional probabilities.
The procedure can be viewed as repeatedly zooming in on the probability interval $[0,1)$.
Initially, the encoder sets $s=0$ and $t=1$, so that the current interval is $[s,t)=[0,1)$.
After reading a prefix $x_{<i}$, it computes the splitting point
\[
u\eqdef s+(t-s)\D(0\mid x_{<i}).
\]
The encoder then reads the next bit $x_i$. If $x_i=0$, it updates $t\leftarrow u$ and retains the interval $[s,u)$.
If $x_i=1$, it updates $s\leftarrow u$ and retains the interval $[u,t)$.
In either case, the length of the current interval is multiplied by $\D(x_i\mid x_{<i})$.
Consequently, after reading the last bit, the retained interval has length
\[
\prod_{i=1}^n \D(x_i\mid x_{<i})=\D(x).
\]
The encoder then chooses a binary prefix whose corresponding dyadic interval is contained in the final interval.
Such a prefix can be chosen to have length
\[
\log\frac{1}{\D(x)}+O(1).
\]
The decoder reverses this procedure. Given the codeword and access to the same next-bit predictor, it interprets the codeword as specifying a dyadic interval. At each step, it partitions the current interval according to the predicted conditional probability and determines which of the two resulting subintervals contains the dyadic interval, thereby recovering the next bit.

Thus, encoding and decoding require only the $n$ conditional probabilities along the path determined by $x$, rather than an enumeration of the entire support. They can therefore be performed in polynomial time whenever the next-bit predictor is efficient. Moreover, decoding remains correct as long as the encoder and decoder obtain the same conditional probabilities on every query; the accuracy of these probabilities affects only the resulting code length.

If the next-bit predictor is randomized, its random coins must be fixed in advance and shared by the encoder and decoder so that they obtain identical answers to the same queries. Thus, if the learner outputs the hypothesis distribution $\mathcal{E}$ in the form of a randomized next-bit predictor, then, given shared access to this predictor, the construction yields a $\pK^{\poly}$-type compressor (namely, the resulting program may depend on the shared random string) whose description length is close to $\log(1/\mathcal{E}(x))$, up to the loss incurred by approximating the conditional probabilities.

One may ask how much stronger it is to require a next-bit predictor rather than merely a sampler.
In general, there is a gap if AIOWFs exist: if $G_z$ is an AIPRG constructed from an AIOWF and $\mathcal{U}$ is uniform, then $(G_z(\mathcal{U}),\mathcal{U})$ is efficiently samplable but admits no efficient next-bit predictor.
Conversely, if AIOWFs do not exist (as is implied by efficient distribution learning for $\P/\poly$-samplable distributions; see~\Cref{lem:dist-break-prf}), the distributional-inversion argument of Impagliazzo and Luby~\cite{DBLP:conf/focs/ImpagliazzoL89} can be used to construct an efficient next-bit predictor from an explicit description of any efficient sampler.
Essentially the same argument appears in our discussion of white-box learning in \Cref{sec:white_owf}.
\paragraph{Application: A Coding Theorem for $\rK^{\poly}$.} The arithmetic-coding construction above has an application in meta-complexity: a coding theorem for randomized time-bounded Kolmogorov complexity.
Recall that $\rK^{\poly}(x)$ is the minimum length of a randomized program that outputs $x$ in polynomial time with probability at least $2/3$.
Although this is a natural way to incorporate randomness into time-bounded Kolmogorov complexity, no optimal coding theorem is known for $\rK^{\poly}$ over all efficiently samplable distributions, in contrast to $\pK^{\poly}$.
Hirahara, Lu, and Nanashima~\cite{DBLP:conf/focs/HiraharaLN24} showed that every next-bits-predictable distribution (i.e., every distribution with a randomized next-bit predictor) admits an efficient coding theorem for $\rK^{\poly}$, with arithmetic coding as a central ingredient.

The main obstacle to obtaining $\rK^{\poly}$-type compression is ensuring that the encoder and decoder reconstruct exactly the same arithmetic-coding intervals. 
In the $\rK$ setting, however, they cannot share random coins, because the encoder must fix the description before the decoder's random coins are drawn. 
Independent executions of the randomized next-bit predictor may produce slightly different estimates of the same conditional probability, and even a small discrepancy can lead the encoder and decoder to use different interval boundaries, causing decoding to fail. 
To prevent this, the construction uses the ``add noise and round'' technique of Goldberg and Sipser~\cite{DBLP:journals/siamcomp/GoldbergS91} and Trevisan, Vadhan, and Zuckerman~\cite{DBLP:journals/cc/TrevisanVZ05}, which adds a small random shift to each estimate and rounds the result to a sufficiently coarse grid.
With high probability over the shift, independent accurate executions of the predictor produce the same rounded estimates.
Thus, only a logarithmic-length description of the shift, rather than the entire randomness of the predictor, needs to be included in the code.
Consequently, if $\D=\{\D_n\}$ is next-bits predictable, then, for every constant $\varepsilon>0$, there is a polynomial time bound for which
\[
\rK^{\poly}(x)
\leq
(1+\varepsilon)\log\frac{1}{\D_n(x)}+O(\log n)
\]
for every $x\in\supp(\D_n)$.

Here, the factor $1+\varepsilon$ arises from next bits having very small conditional probabilities.
For such a bit, the rounding error may be too large relative to its true probability, so the construction records the bit and its position explicitly rather than encoding it through the estimated probability.
The cost of these explicit records accounts for the multiplicative loss in the worst-case bound.
For a sample $x\sim\D_n$, however, such low-probability next bits occur only with correspondingly small probability.
Their contribution can therefore be absorbed into the additive overhead on average, removing the multiplicative $1+\varepsilon$ loss.
Consequently, for every polynomial $q$, there exists a polynomial $p$ such that
\[
\Pr_{x\sim\D_n}\left[
\rK^{p(n)}(x)
\leq
\log\frac{1}{\D_n(x)}+\log p(n)
\right]
\geq
1-\frac{1}{q(n)}.
\]
The construction also gives efficient encoding and decoding algorithms whose expected encoding length is at most $\Ent{\D_n}+O(\log n)$.

Combining this average-case theorem with the distributional-inversion argument mentioned above yields an average-case coding theorem for $\rK^{\poly}$ over arbitrary efficiently samplable distributions under the nonexistence of one-way functions.
Such coding theorems relate $\rK^{\poly}(x)$ to both the statistical information content $\log(1/\D_n(x))$ and the unbounded Kolmogorov complexity $\K(x)$.
This connection allows characterizations of one-way functions previously formulated using $\pK^{\poly}$ to be extended to $\rK^{\poly}$.

\subsubsection{Compression to Learning over Polynomially Many Messages}
\label{sec:compression-to-learning-poly-alphabet}
We now turn to the reverse direction: transforming an \emph{efficient}
compressor, ideally one whose code lengths approximate time-bounded
Kolmogorov complexity, into an \emph{efficient} learner. Recall that in
\Cref{sec:codinglearningstat}, any code-length function $c$ was converted
into a distribution $\D_c$ whose probability mass at each string $x$ is
proportional to $2^{-c(x)}$.

Observe that the KL-divergence bound in
\Cref{prop:kolmogorov-to-distribution} uses no particular property of
$\K(\cdot)$ beyond its coding theorem: for a source distribution $\D$,
\[
\K(x)\lesssim \log\frac{1}{\D(x)}.
\]
Consequently, the same argument applies when $c(\cdot)$ approximates a
complexity measure other than time-unbounded Kolmogorov complexity,
provided that the measure satisfies an analogous coding theorem for the
source distributions under consideration. In particular, when $\D$ is
efficiently samplable and $c(x)\approx\pK^{\poly}(x)$,
the coding theorem for $\pK^{\poly}$ yields the same KL-divergence bound
for the induced distribution $\D_c$. Thus, if we disregard the efficiency
of constructing and sampling from $\D_c$, a compressor approximating
$\pK^{\poly}$ already gives an information-theoretic transformation into
a distribution learner for efficiently samplable distributions. The
remaining question is whether this transformation itself can be
implemented efficiently.

From a complexity-theoretic perspective, the immediate obstacle is efficient sampling.
Even if each value $c(u)$ can be evaluated efficiently, computing the normalization constant and sampling from $\D_c$ may require enumerating exponentially many possible values of $u$.
Indeed, overcoming this sampling problem is one of the central obstacles to deriving efficient learners from meta-complexity assumptions.

One simple way to set aside this obstacle is to consider the case in which each outcome has logarithmic length, so that the sample space has polynomial size.
In this case, both the normalization constant and a sampler for $\D_c$ can be computed by direct enumeration.
We therefore begin with this setting.
Taken by itself, however, this setting is not particularly interesting from the perspective of distribution learning, since an arbitrary distribution over a polynomial-size domain can already be learned efficiently by estimating the probability of every possible outcome from samples.

Nevertheless, generating an output from a polynomial-size alphabet becomes a nontrivial and well-motivated learning problem in the conditional setting.
Suppose that the learner receives i.i.d.\ samples from a joint distribution $\D$ over pairs $(X,Y)$, where $X$ is an observed input or context and $Y$ is the output to be predicted.
We allow $X$ to have length $n$, whereas $Y$ has length $O(\log n)$, so its range $\mathcal{Y}$ has polynomial size. This setting is well motivated. Rather than generating an entire pair consisting of an image and its label, we may wish to estimate the probability of each possible label for a given image.
Similarly, we may wish to generate an answer conditioned on a given question, rather than generating a pair consisting of a question and an answer.
In addition, an output of arbitrary length can be divided into logarithmic-length blocks and generated sequentially.
We return to this idea in \Cref{sec:compression-to-learning-arbitrary-alphabet}.

Accordingly, in the remainder of this subsection, we focus on the following conditional learning task.
Given i.i.d.\ samples from $\D$ and a new input $x$, the learner aims to construct a hypothesis distribution approximating $\D(\cdot\mid x)$, where
\[
\D(y\mid x)\eqdef\Pr_{(X,Y)\sim\D}[Y=y\mid X=x].
\]
Since the range of $X$ may be exponentially large, this task does not reduce to estimating the probability of every pair $(x,y)$.
By contrast, since $\mathcal{Y}$ has polynomial size, for any fixed $x$, once a weight can be evaluated efficiently for every $y\in\mathcal{Y}$, the learner can normalize these weights and sample from them by direct enumeration.
The remaining question is which description length should be used to define these conditional weights. Below, we discuss two natural candidates.

The first is the conditional description length $c(y\mid x)$, where $x$ is treated as auxiliary input available to both the encoder and the decoder.
This gives the hypothesis distribution
\[
\widehat{\D}^{\mathrm{cond}}_c(y\mid x)
\eqdef
\frac{2^{-c(y\mid x)}}{\sum_{z\in\mathcal{Y}}2^{-c(z\mid x)}}.
\]
This construction requires access to a good conditional compressor, which can be substantially more demanding than unconditional compression.
Even in the idealized setting, it is currently unclear whether efficient computation of $\K^{\poly}(\cdot)$ implies efficient computation of $\K^{\poly}(\cdot\mid\cdot)$, whereas the converse is immediate by taking the auxiliary input to be the empty string. Thus, even given an efficient unconditional compressor that closely approximates $\K^{\poly}(\cdot)$, this construction itself does not yield an efficiently computable conditional hypothesis distribution.

The second candidate requires only an unconditional description-length measure. It uses the increase in description length caused by appending $y$ to $x$:
\[
\widehat{\D}^{\mathrm{inc}}_c(y\mid x)
\eqdef
\frac{2^{-(c(xy)-c(x))}}{\sum_{z\in\mathcal{Y}}2^{-(c(xz)-c(x))}}
=
\frac{2^{-c(xy)}}{\sum_{z\in\mathcal{Y}}2^{-c(xz)}}.
\]
This construction has a direct interpretation in terms of the distribution $\D_c$ defined in \Cref{eq:defDc}:
\[
\widehat{\D}^{\mathrm{inc}}_c(y\mid x)=\D_c(y\mid x).
\]
Thus, whenever $\D_c$ is close to $\D$ in KL divergence, the chain rule
for KL divergence (\Cref{lem:chainruleKL}) implies that
$\widehat{\D}^{\mathrm{inc}}_c(\cdot\mid x)$ is close to
$\D(\cdot\mid x)$ in KL divergence on average over $x\sim\D_X$:
\begin{align*}
\Exp_{x\sim\D_X}\left[
\KL\left(
\D(\cdot\mid x)
\middle\|
\widehat{\D}^{\mathrm{inc}}_c(\cdot\mid x)
\right)
\right]
&=
\Exp_{x\sim\D_X}\left[
\KL\left(
\D(\cdot\mid x)
\middle\|
\D_c(\cdot\mid x)
\right)
\right]\\
&=
\KL(\D\|\D_c)
-
\KL\left(\D_X\|(\D_c)_X\right)\\
&\leq
\KL(\D\|\D_c),
\end{align*}
where $\D_X$ and $(\D_c)_X$ denote the respective marginals of $\D$
and $\D_c$ on $\mathcal{X}$. Since $\mathcal{Y}$ has polynomial size,
if $c(\cdot)$ is efficiently computable and polynomially bounded, then
$\widehat{\D}^{\mathrm{inc}}_c(\cdot\mid x)$ can be efficiently computed
and sampled from by enumerating all $y\in\mathcal{Y}$.

The preceding argument is sufficient when the induced joint distribution $\D_c$ is close to $\D$. However, requiring a good approximation to the entire joint distribution can be stronger than necessary for conditional learning, which aims only to approximate $\D(\cdot\mid x)$ and need not model the marginal distribution of $X$. It is therefore natural to ask whether the desired conditional guarantee can be obtained directly, without first approximating the entire joint distribution. Deriving such a conditional learner from meta-complexity assumptions without implicitly learning the marginal distribution of the conditioning variable remains a challenging problem. In meta-complexity, this question is closely connected to two notions, \emph{conditional coding} and \emph{symmetry of information}, which we discuss at a conceptual level below.

Conditional coding is the conditional counterpart of the coding theorem discussed above.
Informally, it asks whether a sample $y$ drawn from $\D(\cdot\mid x)$ admits a $\K^{\poly}$-description of length close to the conditional information content $\log(1/\D(y\mid x))$, with $x$ available as an auxiliary input.
More generally, the same question can be formulated for any conditional description-length measure.
For the learning application, an average-case guarantee is sufficient, since the learner is evaluated on a fresh pair $(x,y)\sim\D$.
If such a guarantee is available for an efficiently computable conditional description-length measure, the same argument as \Cref{prop:kolmogorov-to-distribution} can be applied conditionally to analyze $\widehat{\D}^{\mathrm{cond}}_c$.
For time-unbounded Kolmogorov complexity $\K$, conditional coding follows immediately from the ordinary coding theorem, since the conditional distribution of a samplable joint distribution is itself samplable when no running-time bound is imposed.
By contrast, the validity of its time-bounded analogues is closely connected to complexity-theoretic and cryptographic assumptions.
Consequently, conditional coding and the complexity of time-bounded conditional descriptions have received considerable attention in meta-complexity~\cite{DBLP:conf/stoc/HiraharaILNO23,DBLP:conf/icalp/LuS24,DBLP:conf/focs/HiraharaLN24,HLO26,DBLP:conf/stoc/HuMO26}.

Symmetry of information connects the second construction, which uses only unconditional description lengths, to the first, which describes $y$ with $x$ available as auxiliary input.
Roughly speaking, it asks whether $\K^{\poly}(xy)$ can be decomposed as $\K^{\poly}(x)+\K^{\poly}(y\mid x)$ up to a small additive loss.
The analogous question can also be posed for other description-length measures.
If such a decomposition holds, the incremental description length $\K^{\poly}(xy)-\K^{\poly}(x)$ behaves like the conditional description length $\K^{\poly}(y\mid x)$.
For time-unbounded Kolmogorov complexity $\K$, symmetry of information holds unconditionally, up to the logarithmic additive loss~\cite{zvonkin1970complexity}.
By contrast, whether its time-bounded analogues hold is closely connected to central questions in complexity theory and cryptography, including the $\P$ versus $\NP$ problem and the existence of one-way functions~\cite[e.g.,][]{DBLP:journals/poit/Levin03,DBLP:journals/ipl/LongpreM93,DBLP:journals/iandc/LongpreW95}.
Accordingly, time-bounded forms of symmetry of information, their connections to one-way functions and average-case complexity, and unconditional failures in certain formulations have been studied extensively in the meta-complexity literature~\cite{DBLP:journals/tcs/LeeR05,DBLP:conf/stoc/Hirahara21,DBLP:conf/coco/LiuP22a,GKLO22,GK22,DBLP:conf/coco/Hirahara22,DBLP:conf/stoc/HiraharaILNO23,DBLP:conf/tcc/HiraharaLO24,DBLP:conf/focs/HiraharaLN24,HLO26,DBLP:journals/eccc/GoldbergHLL025,DBLP:conf/stoc/HuMO26}.

We have seen why conditional coding and symmetry of information are natural tools for analyzing the two constructions above.
To the best of our knowledge, however, existing results do not directly convert an unconditional compression guarantee over the target distribution into a conditional KL guarantee for the distribution obtained by normalizing the incremental weights $2^{-(c(xy)-c(x))}$, except in settings where the assumptions under consideration already imply that the context $x$ is also compressed nearly optimally on average under the corresponding marginal distribution.
Namely, it remains open whether conditional coding and symmetry of information can instead yield direct conditional-learning guarantees beyond those obtained by approximating the entire joint distribution.
In particular, it would be of great interest to establish such a transformation for a restricted class of distributions without relying on implausible assumptions that rule out standard cryptography.

Another closely related result is the KL-sampling characterization of conditional pseudoentropy due to Vadhan and Zheng~\cite{DBLP:conf/stoc/VadhanZ12}. For polynomial-size output alphabets, their result characterizes conditional pseudo-Shannon entropy in terms of the minimum average conditional KL divergence achievable by an efficient sampler. This provides another complexity-theoretic link between conditional information and efficient conditional sampling, although it is formulated in terms of pseudoentropy rather than description-length scores. Recent work has revisited and generalized this characterization using multicalibration and weight-restricted calibration~\cite{DBLP:conf/stoc/CasacubertaDV24,DBLP:conf/tcc/HuV25}.

\subsubsection{Compression to Learning over General Messages}
\label{sec:compression-to-learning-arbitrary-alphabet}

For polynomial-length messages, the output space may be exponentially large, making it infeasible to evaluate the score of every possible $y$ and normalize the resulting weights as in \Cref{sec:compression-to-learning-poly-alphabet}.
We discuss two approaches to this setting.
It is worth keeping in mind, however, that the connection between distribution learning and meta-complexity in this time-bounded regime remains largely open.

\paragraph{Universal Extrapolation.} The first approach is to use universal extrapolation.
Recall that $\q^t(x)\eqdef-\log\Q^t(x)$, where $\Q^t$ is the time-bounded universal distribution. 
Thus, for the code-length function $c(x)=\q^t(x)$, the distribution $\D_c$ defined in \Cref{eq:defDc} coincides exactly with $\Q^t$.
The universal distribution satisfies the following domination property. For every distribution $\D$ samplable in time $t$ by a sampler of description length $|\D|$, and every $x\in\supp(\D)$, 
\[\Q^{\poly(t)}(x)\ge 2^{-O(|\D|)}\cdot \D(x)\]
Indeed, the universal sampler can use an initial portion of its random input to specify the sampler for $\D$ and use the remaining bits as its random seed. The probability of selecting this sampler is at least $2^{-O(|\D|)}$, and, conditioned on doing so, the output distribution is exactly $\D$. This bounds the KL divergence from an unknown samplable distribution
$\D$ to $\Q^t=2^{-\q^t}$ in terms of the description length of a
sampler for $\D$, since the domination property immediately implies the following coding-theorem type
bound:
\[
\q^t(x)
=
-\log\Q^t(x)
\leq
\log\frac{1}{\D(x)}
+
O(|\D|).
\]
The claimed KL-divergence bound then follows by the same argument as in
\Cref{prop:kolmogorov-to-distribution}.
To obtain the total variation bound, recall the batch-amortization
argument following \Cref{prop:kolmogorov-to-distribution}, in which the
description-length term is amortized by conditioning a uniformly random
coordinate on the preceding independent samples. Implementing this
argument here requires sampling from the time-bounded universal
distribution conditioned on an observed prefix of the batch.  
This conditional-sampling task is exactly what is called universal
extrapolation in \cite{DBLP:conf/focs/HiraharaN23}, which is equivalent
to the task described by \Cref{def:advUE} together with the procedure
that follows it, specialized here to the case where no additional advice
is given, as required for distribution learning.

Under the nonexistence of infinitely-often one-way functions, universal extrapolation (with no advice) can be performed in time polynomial in the number of bits to be extrapolated, the time bound, the inverse accuracy, and $2^{\q^t(x)-\K(x)}$~\cite{DBLP:conf/focs/ImpagliazzoL90,DBLP:conf/focs/HiraharaN23}.
Its running time is therefore exponential in the computational depth $\q^t(x)-\K(x)$ of the conditioning prefix $x$.
The slow-growth law for computational depth (see \Cref{sec:time-boundedSol}) states that efficient randomized computation increases computational depth by at most a logarithmic amount with high probability. 
Consequently, if the description of a sampler has low computational depth, then the conditioning prefixes generated by that sampler also have low computational depth with high probability.
One concrete example is a sampler that is secretly selected by a uniform efficient randomized algorithm.
Although the selected sampler may itself be a nonuniform object of polynomial length, its computational depth is logarithmic with high probability.
For samplers of logarithmic computational depth, universal extrapolation therefore runs in polynomial time and, combined with the batch-amortization
argument following \Cref{prop:kolmogorov-to-distribution}, yields the following result without any restriction on the size of the message space. 

\begin{theorem}[\cite{DBLP:conf/focs/HiraharaN23}]
If infinitely-often one-way functions do not exist, then every efficiently samplable distribution represented by a sampler of logarithmic computational depth is distribution learnable in polynomial time (with respect to total variation distance, as defined in \Cref{def:distlearning}).
\end{theorem}

A limitation of this approach is that it requires ruling out one-way functions.
Known constructions of universal extrapolation rely on consequences of the nonexistence of one-way functions, such as distributional inversion, and this dependence is unavoidable for general distribution classes because efficient universal extrapolation itself implies the nonexistence of one-way functions.
Consequently, even for a restricted class of distributions that might be efficiently learnable unconditionally, this approach still invokes the seemingly implausible assumption that one-way functions do not exist.
Thus, universal extrapolation reveals strong connections among learning, cryptography, and meta-complexity, but it remains unclear whether it can yield unconditional learning results for restricted distribution classes that exclude cryptographically hard sources.

Another major open question concerns the complexity of universal extrapolation itself.
Existing techniques based on distributional inversion require time exponential in the computational depth of the conditioning prefix.
They therefore do not give efficient conditional generation when the prefix has large computational depth.

\begin{problem}
What is the weakest assumption under which universal extrapolation can be performed efficiently for every conditioning prefix?
In particular, is there a meta-complexity assumption that characterizes this worst-case form of universal extrapolation?
\end{problem}

\paragraph{Sequential Generation.} The second approach is to divide a polynomial-length message into polynomially many short blocks and generate them autoregressively.
Write a message as $x=x_1\cdots x_m$, where each block $x_i$ has length $O(\log n)$.
After generating the preceding blocks $x_{<i}$, the learner constructs a conditional distribution over the possible values of $x_i$ and samples the next block from this distribution.
Since each block has only polynomially many possible values, the required normalization can be performed by direct enumeration, as in \Cref{sec:compression-to-learning-poly-alphabet}.

Which compression guarantee is sufficient to construct such a sequential generator?
To formulate this question, first observe the easy converse direction: the same conditional distributions can be used to compress the message block by block.
Once $x_{<i}$ has been decoded, the conditional distribution associated with this prefix determines a code for $x_i$ through the arithmetic-coding construction reviewed in \Cref{sec:learning-to-compression-arithmetic}.
This form of sequential blockwise coding is known as \emph{prequential compression}.
A prequential compressor processes $x_1,\ldots,x_m$ online, encoding each block $x_i$ with the preceding blocks $x_{<i}$ available as auxiliary input.
Ideally, its total conditional description length should be comparable to a one-shot time-bounded description of the entire message $x$, such as $\K^{\poly}(x)$ or one of its variants that admits optimal coding, up to an additive $m\cdot O(\log n)$ overhead for encoding the blocks separately.

This objective is formalized by an \emph{algorithmic chain rule}, as studied in~\cite{DBLP:journals/eccc/KK25,DBLP:journals/eccc/HN25}. For $\q^{\poly}$, suppressing polynomial changes in the time bounds, such a chain rule takes the form
\[
\sum_{i=1}^{m}\q^{\poly}(x_i\mid x_{<i})
\leq
\q^{\poly}(x_1\cdots x_m)+m\cdot O(\log n).
\]

For a source distribution $\D$, we define the corresponding average-case notion by requiring the expected sum of conditional description lengths to be close to the source entropy, up to the same accumulated blockwise overhead\footnote{This is equivalent, in the constant-parameter regime, to the formulation in~\cite{DBLP:journals/eccc/HN25}, which allows an additive $\eta N$ term, where $N$ denotes the total message length. Indeed, if each block has length $b$, then $m=\Theta(N/b)$ and hence $m\cdot O(\log n)=O(N\log n/b)$. Choosing $b=\Theta(\eta^{-1}\log n)$ makes this overhead $O(\eta N)$. Conversely, when $b=\Theta(\log n)$ and $\eta$ is constant, an additive $\eta N$ term is of the form $m\cdot O(\log n)$. Thus, the two formulations are interchangeable by adjusting the constant in the block length.}.
Namely, average-case prequential $\q^{\poly}$-compression requires the following bound
\[
\Exp_{X\sim\D}\left[
\sum_{i=1}^{m}\q^{\poly}(X_i\mid X_{<i})
\right]
\leq
\Ent{\D}+|\mathcal{D}|+m\cdot O(\log n),
\]
where the constant in the ``$O(\log n)$'' term is fixed independently of the block length. Here, $|\D|$ denotes the description length of a sampler for $\D$ and is absorbed into the final term whenever $m\geq|\D|$. We emphasize that this is purely a statement about the source distribution $\D$; we do not assume the existence of an explicit efficient encoder that realizes the conditional description lengths $\q^{\poly}(\cdot\mid\cdot)$.
Replacing $\q^{\poly}$ with $\rK^{\poly}$ gives the corresponding average-case notion of prequential $\rK^{\poly}$-compression.

Hirahara and Nanashima~\cite{DBLP:journals/eccc/HN25} use this viewpoint to give an unconditional, distribution-by-distribution characterization.
Here, it is important to distinguish the elementary symbols to be predicted from the longer blocks used for compression.
Let $Z=(Z_1,\ldots,Z_N)\sim\D$, where each $Z_i$ is a bit or, more generally, a symbol of constant length.
We group consecutive symbols into blocks $X_1,\ldots,X_m$, each consisting of $b=O(\log n)$ symbols.
This grouping is used for prequential compression, while the sequential generator still predicts one elementary symbol $Z_i$ at a time.

Let $\mathcal{A}(Z_{<i})$ denote the output distribution of a sequential generator after observing $Z_{<i}$.
For accuracy $\epsilon>0$ and confidence error $\delta>0$, the desired guarantee is
\[
\Pr_{\substack{Z\sim\D\\ I\sim[N]}}\left[
\TV\bigl(\mathcal{A}(Z_{<I}),(Z_I\mid Z_{<I})\bigr)
\leq\epsilon
\right]
\geq 1-\delta,
\]
where $I$ is uniform over $[N]$.
The authors show that, in the constant-accuracy and constant-confidence regime, the existence of an efficient sequential generator satisfying this guarantee is equivalent to average-case prequential $\q^{\poly}$-compression of the block sequence $X_1,\ldots,X_m$ and, equivalently, to average-case prequential $\rK^{\poly}$-compression. 

Although this does not yield the total variation guarantee required in \Cref{def:distlearning}, it does yield the following Yao-style next-bit learning guarantee, with no restriction on the support size of the target distribution. This follows from an averaging argument over prefixes
and a standard application of Markov's inequality.

\begin{theorem}\label{thm:preq2nbl}
For every $\epsilon,\delta\in(0,1]$, there exist a constant $C$ and a
polynomial-time learner $L$ such that the following holds. Let $\D$ be
any $\Ppoly$-samplable distribution over $\bin^n$. Suppose that, for
every batch size $m$, the sequence obtained by concatenating $m$
independent samples from $\D$ and partitioning the resulting string into
consecutive blocks of length $C\log n$ admits average-case
prequential $\q^{\poly}$-compression or average-case prequential
$\rK^{\poly}$-compression.

Then, given sample access to $\D$, the learner $L$ outputs a
polynomial-time randomized next-bit predictor $h$ satisfying
\[
\Pr_{h\leftarrow L^\D}
\left[
\Exp_{\substack{x\sim\D\\ i\sim[n]}}
\left[
\TV\left(
h(x_{<i}),
(\D_i\mid \D_{<i}=x_{<i})
\right)
\right]
\leq
\epsilon
\right]
\geq
1-\delta.
\]

\end{theorem}

The transformation from prequential $\q^{\poly}$-compression to sequential generation applies the advised universal extrapolation procedure discussed in \Cref{sec:time-boundedSol}.
As reviewed there, the completed blocks $X_{<j}$ are treated as advice. 
The algorithmic chain rule for prequential compression incurs an additional $O(\log n)$ KL loss per block.
To obtain the total variation bound stated above, this loss must once again be amortized through conditional sampling.
Rejection sampling provides the required conditioning.
When generating $Z_i$ in the current block $X_j$, rejection sampling is used only to condition on the previously generated symbols in $X_j$, since the completed blocks $X_{<j}$ are already available as advice.
Since the current block contains only $b=O(\log n)$ symbols, every such conditioning prefix has length $O(\log n)$.
The corresponding conditioning event therefore has inverse-polynomial probability, and rejection sampling runs in polynomial time.
Unlike the universal-extrapolation approach discussed above, this transformation does not require distributional inversion or any assumption ruling out one-way functions.

Accordingly, the result does not claim that every efficiently samplable distribution can be learned, which would be incompatible with cryptographic hardness.
Instead, it gives a distribution-by-distribution characterization: a distribution family admits efficient sequential generation precisely when it satisfies the corresponding prequential-compression condition.

The short-block construction balances coding overhead against the running time of rejection sampling.
To reduce an additive $O(\log n)$ overhead per block to an $\epsilon$ fraction of the block length, the block length must be $\Omega(\epsilon^{-1}\log n)$.
Rejection sampling on such a block takes time $2^{O(\epsilon^{-1}\log n)}=n^{O(1/\epsilon)}$.\footnote{
More precisely, here we use the elementary fact that the universal
distribution assigns probability at least $2^{-O(b)}$ to any
$b$-bit prefix, since an $O(b)$-bit program can simply print an
extension of that prefix, ignoring the advice.}
This is polynomial in $n$ for every constant $\epsilon$, but it is not polynomial jointly in $n$ and $\epsilon^{-1}$.
The current characterization is therefore limited to constant accuracy and confidence, and extending it to general parameters remains open.

\begin{problem}
Does the equivalence between efficient sequential generation, prequential $\q^{\poly}$-compression, and prequential $\rK^{\poly}$-compression extend to arbitrary polynomial-length blocks and arbitrary accuracy and confidence parameters, with running time polynomial in the message length, the block length, $\epsilon^{-1}$, and $\delta^{-1}$?
\end{problem}

In particular, it is important that $\epsilon$ can be made inverse-polynomially small in $n$, as this allows the next-bit prediction guarantee in \Cref{thm:preq2nbl} to be converted into the standard total variation guarantee via the triangle inequality.

\subsubsection{Language Modeling}
\label{sec:language-modeling-compression}

The framework developed above appears directly in modern language modeling.
A \emph{token} is a symbol from a finite vocabulary, such as a byte or a fragment of a word.
An \emph{autoregressive language model} $\mathcal{E}$ assigns, to every prefix $x_{<i}$, a distribution $\mathcal{E}(\cdot\mid x_{<i})$ over the next token.
It generates a message by repeatedly sampling the next token from this distribution.
Thus, with each token viewed as a short block, language modeling is precisely the sequential conditional-generation setting considered in \Cref{sec:compression-to-learning-arbitrary-alphabet}.

The standard objective for training an autoregressive language model is the next-token log loss \cite[e.g.,][]{radford2019language,DBLP:conf/nips/BrownMRSKDNSSAA20}: 
\[
L_{\mathcal{E}}(x)
\eqdef
\sum_{i=1}^{N}\log\frac{1}{\mathcal{E}(x_i\mid x_{<i})} = 
\log\frac{1}{\mathcal{E}(x)}.
\]
As reviewed in \Cref{sec:learning-to-compression-arithmetic}, arithmetic coding uses the same conditional probabilities to encode $x$ in $L_{\mathcal{E}}(x)+O(1)$ bits.
Hence, minimizing log loss is equivalent to minimizing the length of the corresponding lossless encoding.
In this sense, a language model is also a compressor, provided that the model is available to both the encoder and the decoder.

The reverse connection uses the normalized incremental weights discussed in \Cref{sec:compression-to-learning-poly-alphabet}, assigning each possible continuation $y$ the weight $2^{-(c(xy)-c(x))}$.
For each prefix $x$, normalizing these weights defines a distribution over the next token, and iterating this construction turns a compressor into an autoregressive generator.
A good overall compression rate, however, does not in general guarantee that the resulting conditional distributions are accurate.

An early neural implementation of this connection predates modern language models. Schmidhuber and Heil~\cite{SchmidhuberHeil1996} trained a feedforward network to estimate conditional next character probabilities and coupled its outputs with Huffman or arithmetic coding for lossless text compression; in their offline scheme, the trained predictor was shared by the encoder and decoder.

Del\'etang et al.~\cite{DBLP:conf/iclr/DeletangRDCGMGW24} investigate both directions of this correspondence experimentally.
They combine the next-token probabilities produced by large language models with arithmetic coding to obtain lossless compressors for text, images, and audio.
Conversely, they normalize incremental code lengths produced by compressors such as gzip and use the resulting conditional distributions for autoregressive generation.
Related experiments use prequential code length to evaluate neural predictors trained sequentially on previously revealed examples~\cite[e.g.,][]{DBLP:conf/nips/BlierO18,DBLP:conf/iclr/BornscheinLH23}.
Together, these recent developments place the relationship between sequential prediction and compression directly within modern machine learning.

The complexity-theoretic results discussed above provide an idealized time-bounded formulation of the same prediction--compression principle, with $\K^{\poly}$ and related measures such as $\rK^{\poly}$ and $\q^{\poly}$ serving as ideal universal benchmarks for the shortest descriptions that can be decoded in polynomial time.
The existing correspondences thus describe the relationship between prediction and compression at this idealized time-bounded level, but do not yet yield guarantees for concrete large language models.
This raises a natural question: can meta-complexity shed light on concrete language models, for instance, by providing new methods that transfer compression guarantees to sequential-generation guarantees without quantitative loss while preserving computational efficiency?
Developing a connection between the idealized time-bounded theory and modern language models remains an important direction for future research.

\newpage

\section{Meta-Complexity and Computational Hardness from Learning}\label{sec:learning_to_hardness}

Ideas from complexity have long informed the design of new learning algorithms (see, e.g., \citep{DBLP:journals/jacm/LinialMN93}); conversely, ideas from learning have yielded insights in complexity (see, e.g., \citep{DBLP:journals/jcss/ImpagliazzoW01}). More recently, there have been indications of deep and far-reaching connections between the two fields that go well beyond these results. They are part of an emerging theory that shows useful links between \emph{proving complexity lower bounds} and \emph{designing new algorithms}. The learning algorithm of \citep{DBLP:conf/coco/CarmosinoIKK16} and the complexity lower bound of \citep{DBLP:journals/jacm/Williams14} are examples of  breakthroughs achieved in this way.

Learning algorithms in different  models are known for several important classes of Boolean functions, including DNFs \citep{DBLP:journals/jcss/Jackson97}, constant-depth circuits with AND, OR, and NOT gates \citep{DBLP:journals/jacm/LinialMN93}, and constant-depth circuits extended with MOD$_p$ gates, where $p$ is prime \citep{DBLP:conf/coco/CarmosinoIKK16}. However, despite major efforts, a learning algorithm for the class $\mathsf{TC}^0_2$, i.e., two layers of polynomially many threshold gates, is unknown and corresponds to a central problem in computational learning theory. Note that every reasonable concept class can be learned in time $2^n \cdot \poly(n)$, simply by querying the unknown function $f \colon \{0,1\}^n \to \{0,1\}$ on all input strings. 
In contrast, even if we allow the learning algorithm to run in \emph{quantum time} slightly below $2^n$ and make \emph{quantum membership queries}, there is currently no known algorithm for learning $\mathsf{TC}^0_2$.

In a seminal work, \citep{DBLP:conf/colt/FortnowK06} demonstrated that the existence of polynomial-time learning algorithms implies complexity lower bounds. Subsequent results strengthened this connection by showing that even learning in time slightly below $2^n$ suffices to derive lower bounds. Put differently, any ``non-trivial'' learning algorithm can be transformed into a source of hard functions. Hardness can be extracted from non-trivial deterministic  learning algorithms \citep{DBLP:conf/coco/KlivansKO13}, randomised PAC learning algorithms \citep{DBLP:conf/coco/OliveiraS17}, and quantum PAC learning algorithms \citep{DBLP:conf/focs/ArunachalamGGOS21}. In general, the existence of any non-trivial learning algorithm for a class $\mathcal{C}$ of functions (such as $\mathsf{TC}^0_2$) implies that there are explicit problems that cannot be computed in $\mathcal{C}$.\footnote{Interestingly, previous work  \citep{DBLP:conf/innovations/ServedioT17, DBLP:conf/alt/Karchmer24} has shown that non-trivial learning algorithms of the required form can be designed for a large number of circuit classes for which lower bounds are already known.} These results shed some light on the difficulty of designing provably correct learning algorithms for classes such as $\mathsf{TC}^0_2$, for which super-polynomial size complexity lower bounds remain unknown.

In the next sections, we survey results of this kind and present a concrete example illustrating how learning algorithms can be used to prove lower bounds in meta-complexity.

\subsection{Deterministic Learners}

Below we consider \emph{typical} circuit classes $\mathcal{C}$ that are closed under restrictions and complementation,\footnote{Recall that a class is closed under restrictions if any function $g$ obtained by setting a subset of input bits of $f \in \mathcal{C}$ to $0$ and $1$ is also in $\mathcal{C}$. Closure under complementation means that if a Boolean function $f \in \mathcal{C}$ then the function $g(x) = 1 - f(x) \in \mathcal{C}$.} and contain the identically constant functions.\footnote{If a circuit class $\mathcal{C}$ does not satisfy one of these properties, it is not difficult to use the $\mathcal{C}$-circuit evaluation problem to define an explicit language in the complexity class $\mathsf{P}$ that is not in $\mathcal{C}$.}

As discussed above, any class $\mathcal{C}$ can be  learned in time $2^n\cdot \poly(n)$. For this reason, we say that a learning algorithm is \emph{nontrivial} if its running time is of the form $2^n/n^{\omega(1)}$, i.e., if it admits an arbitrarily large polynomial speed-up over the trivial time bound. Our next result shows that the existence of nontrivial deterministic learning algorithms has significant complexity-theoretic consequences.

Recall that in $\MKtP$ (Minimum $\Kt$ Problem), we are given a string $x \in \{0,1\}^n$ and a parameter $s \in \mathbb{N}$, and the goal is to decide if $\Kt(x) \leq s$. Here $\Kt(x)$ denotes the Levin time-bounded Kolmogorov complexity of $x$ \citep{DBLP:journals/iandc/Levin84}, which is defined as 
		\[
		\mathsf{Kt}(x) \eqdef \min_{\substack{\Pi \in \{0,1\}^* \\ t\in\mathbb{N}}} \mleft\{|\Pi| + \lceil \log t \rceil\,\,\middle\vert\,\,  \text{program } \Pi \textnormal{ outputs $x$ within $t$ steps}\mright\}.
		\]
We can also think of $\MKtP$ as a language contained in $\{0,1\}^*$ by fixing the threshold parameter $s$ to a large bound such as $s(n) = n - \log n$. Showing that $\MKtP \notin \P$ is a longstanding open problem in meta-complexity \citep{DBLP:journals/siamcomp/AllenderBKMR06}.

\begin{theorem}[$\MKtP$ Lower Bounds from Deterministic Learning] \label{thm:lb_from_det_learning}The following results hold:
\begin{itemize}
    \item[\emph{(}i\emph{)}] \emph{Non-Uniform Hardness for $\MKtP$ from Subexponential Time Deterministic Learning:\\} Suppose that a typical circuit class $\mathcal{C} = \{\mathcal{C}_n\}_{n \geq 1}$ can be PAC-learned to error $\varepsilon \leq 0.499$ under the uniform distribution using membership queries in deterministic time $2^{n/2}/n^{\omega(1)}$. Then $\MKtP \notin \mathcal{C}$.
    \item[\emph{(}ii\emph{)}] \emph{Uniform Hardness for $\MKtP$ from Nontrivial Deterministic Learning:\\}
    Suppose that, for every $k \geq 1$, Boolean circuits of size $n^k$ can be PAC-learned to error $\varepsilon \leq 0.499$  under the uniform distribution  using membership queries in deterministic time $2^{n}/n^{\omega(1)}$. Then $\MKtP \notin \P$.
\end{itemize}
\end{theorem}

\begin{proof}[Proof Sketch]
    To prove the first item, let $A^{\mathsf{MQ}(\cdot)}(1^n)$ be a deterministic algorithm running in time $T(n)$ that learns every $f \in \mathcal{C}_n$ to error $\varepsilon$. Let $\bar{0} \colon \{0,1\}^n \to \{0,1\}$ be the constant zero function. Suppose towards a contradiction that $\MKtP \in \mathcal{C}$. Consider the Boolean function $f_n \colon \{0,1\}^n \to \{0,1\}$ such that $f_n(x) = 1$ if and only if $\Kt(x) \geq n - \log n$. Then $\{f_n\}_{n \geq 1} \in \mathcal{C}$ under the assumption that $\MKtP \in \mathcal{C}$, since $f_n$ is the complement of a restriction of the language $\MKtP$.

    Next, consider the execution of $A^{\mathsf{MQ}(\cdot)}(1^n)$ on $\bar{0}$, and let $h_{\bar{0}}$ be its output hypothesis. Similarly, let $h_{f_n}$ be the output hypothesis produced on $f_n$. Then
    $$
    \mathsf{dist}(h_{\bar{0}}, \bar{0}) \;\leq\; \varepsilon \quad \text{and} \quad \mathsf{dist}(h_{f_n},f_n) \;\leq\; \varepsilon.
    $$
    We claim that if $T(n) = 2^{n/2}/n^{\omega(1)}$ and $n$ is large enough, then $h_{\bar{0}} \equiv h_{f_n}$, i.e., they compute the same function. It is enough to argue that $A^{\mathsf{MQ}(\bar{0})}(1^n)$ and $A^{\mathsf{MQ}(f_n)}(1^n)$ have the same computation transcript, and therefore produce the same hypotheses. To see this, we argue that every membership query to $\mathsf{MQ}(f_n)$ produces $0$ (as in $\mathsf{MQ}(\bar{0})$). Indeed, it is not hard to prove by induction on $i$ that if $q_i \in \{0,1\}^n$ is the $i$-th query, where $i \leq T(n)$, then $\Kt(q_i) < n - \log n$. This is because we can describe $q_i$ by knowing the code of $A$, $n$, and $i$, and by  simulating $A$ on input $1^n$ for at most $T(n)$ steps. Consequently,
    $$
\Kt(q_i) \leq (|A| + \log n + \log i) + \log (T(n) \cdot \poly(n)) \leq O(\log n) + 2 \log T(n) \leq n - \omega(\log n) < n - \log n,
    $$
    where we used that $T(n) = 2^{n/2}/n^{\omega(1)}$. Consequently, 
    $$
\mathsf{dist}(h_{\bar{0}}, h_{f_n}) \;=\; 0.
    $$
    Applying the triangle inequality, we obtain that 
    $$
\mathsf{dist}(\bar{0}, f_n) \;\leq\; 2 \varepsilon.
    $$
    In contrast, since a random $n$-bit string has $\Kt$ complexity $\geq n -\log n$ with probability $1 - o(1)$, we have
    $$
\mathsf{dist}(\bar{0}, f_n) \;\geq\; 1 - o(1).
    $$
   This contradicts our choice of $\varepsilon \leq 0.499$. 

For the proof of the second item, an additional idea is used. Assume now that $T(n) = 2^n/n^{\omega(1)}$ and $\MKtP \in \P$. For each $k$, let $\mathcal{C}^k$ be the class of circuits of size $n^k$. Adapting the proof of the first item, we argue below that $\MKtP \notin \mathcal{C}^k$. Since this holds for all $k$, it implies in particular that $\MKtP \notin \P$, which is a contradiction.

Fix $k \in \mathbb{N}$. Note that, in the argument above, we had to assume a lower running time for the learning algorithm because of the additional $\log T(n)$ contribution coming from the encoding of $i \in [T(n)]$, the index of the query. We avoid the extra term coming from the encoding of the index. We will use that $\MKtP \in \P$ to show that $\Kt(q_i) < n - \log n$ for every query $q_i \in \{0,1\}^n$ made by the learner. In more detail, towards a contradiction, suppose that for some $i$ we have $\Kt(q_i) \geq n - \log n$, and let $i$ be the first such index. Then, we can describe $q_i$ as follows: $q_i$ is the first query made by $A(1^n)$ during the simulation such that $\Kt(q_i) \geq n - \log n$. Since the latter condition can be checked in polynomial time under the assumption that $\MKtP \in \P$, it is not hard to see that 
$$
\Kt(q_i) \;\leq\; O(1) + \log n + \log (T(n) \cdot \poly(n)) \;<\; n - \log n,
$$
for large enough $n$. This contradiction completes the proof.
\end{proof}

Complementing this result, we mention that separations against $\mathcal{C}$ (i.e., non-uniform lower bounds) for some problem in $\mathsf{E} = \mathsf{DTIME}[2^{O(n)}]$ (instead of $\MKtP \in \mathsf{E}$) also follow from  deterministic learners running in time $T(n) = 2^{n}/n^{\omega(1)}$. For more details, see the results and techniques in \citep{DBLP:conf/coco/KlivansKO13}.

\subsection{Randomized and Quantum Learners}

As we consider more powerful learning algorithms, the task of constructing a computationally hard function from them becomes more elusive. For instance, in the proof of \Cref{thm:lb_from_det_learning}, we used in a crucial way that the set of queries produced by the learner is uniquely determined once we fix the function $f \in \mathcal{C}_n$ for the learning task. If the learning algorithm is \emph{randomized}, this property no longer holds, and the strategy described above simply does not work.

Nevertheless, using more sophisticated arguments that rely on the theory of computational pseudorandomness, it is possible to extract complexity lower bounds from any non-trivial learning algorithm. Let $\mathsf{BPE} \eqdef \mathsf{BPTIME}[2^{O(n)}]$  denote randomized exponential time, i.e., the randomized analogue of $\mathsf{E}$. 

\begin{theorem}[Lower Bounds from Nontrivial Randomized Learning \citep{DBLP:conf/coco/OliveiraS17}]\label{thm:rand_learning_lbs}
    Let $\mathcal{C}$ be a typical circuit class. If polynomial-size $\mathcal{C}$-circuits can be learned using membership queries with error $\varepsilon < 0.499$ under the uniform distribution in randomized time $2^n/n^{\omega(1)}$, then $\mathsf{BPE} \nsubseteq \mathcal{C}(\poly)$.
\end{theorem}

It is also possible to consider a more expressive class of \emph{quantum} learning algorithms (see \citep{arunachalam2017guest} for more on quantum learning theory).  Let $\mathsf{BQE} \eqdef \mathsf{BQTIME}[2^{O(n)}]$  denote quantum exponential time. 

\begin{theorem}[Lower Bounds from Nontrivial Quantum Learning \citep{DBLP:conf/focs/ArunachalamGGOS21}]\label{thm:quant_learning_lbs}
        Let $\mathcal{C}$ be a typical circuit class. If polynomial-size $\mathcal{C}$-circuits can be learned using quantum membership queries with error $\varepsilon < 0.499$ under the uniform distribution in quantum time $2^n/n^{\omega(1)}$, then $\mathsf{BQE} \nsubseteq \mathcal{C}(\poly)$.
\end{theorem}

In particular, it follows that if there is a randomized (or even quantum) learning algorithm for the class of depth-$2$ threshold circuits ($\mathsf{TC}^0_2$) that runs in time $2^n/n^{\Omega(1)}$, then a new complexity lower bound would follow (see \Cref{fig:learningcomplexityfrontier} for a pictorial representation and a related discussion). We refer to \citep{DBLP:conf/stoc/AlmanPS26} and references therein for progress on the learnability of subclasses of depth-$2$ threshold circuits.

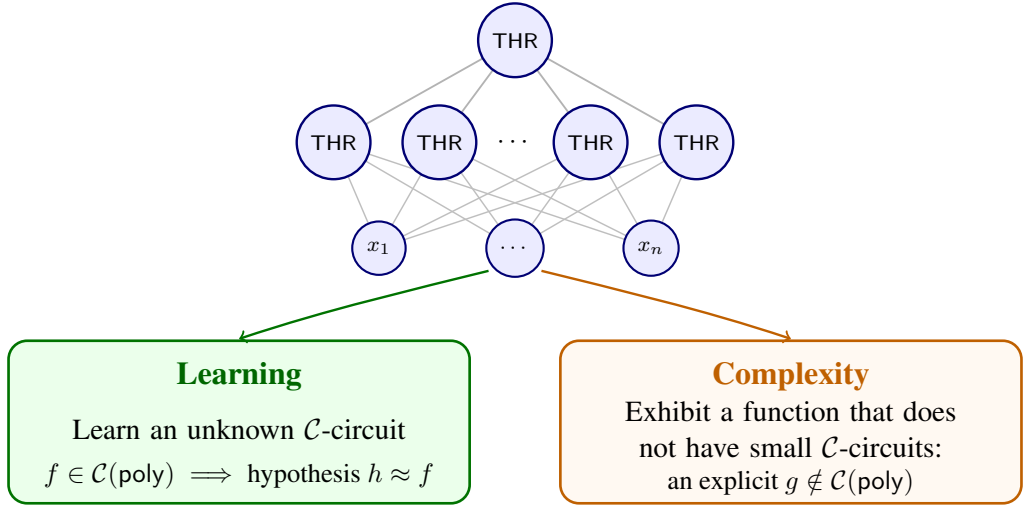
\begin{figure}[ht]
\centering
\begin{tikzpicture}[
    circuitgate/.style={
        circle,
        draw=blue!45!black,
        fill=blue!8,
        line width=0.9pt,
        minimum size=0.72cm,
        font=\scriptsize
    },
    circuitinput/.style={
        circle,
        draw=blue!45!black,
        fill=blue!8,
        line width=0.8pt,
        minimum size=0.48cm,
        font=\scriptsize
    },
    circuitdots/.style={
        inner sep=0pt,
        font=\small
    },
    learningpanel/.style={
        rounded corners=7pt,
        draw=green!45!black,
        fill=green!8,
        line width=1pt,
        minimum width=6.1cm,
        minimum height=2.15cm
    },
    complexitypanel/.style={
        rounded corners=7pt,
        draw=orange!75!black,
        fill=orange!5,
        line width=1pt,
        minimum width=6.1cm,
        minimum height=2.15cm
    }
]

\coordinate (top) at (7,6.2);
\coordinate (lowerone) at (4.6,4.85);
\coordinate (lowertwo) at (6.0,4.85);
\coordinate (lowerdots) at (7,4.85);
\coordinate (lowerthree) at (8.0,4.85);
\coordinate (lowerfour) at (9.4,4.85);
\coordinate (inputleft) at (5.2,3.45);
\coordinate (inputcentre) at (7,3.45);
\coordinate (inputright) at (8.8,3.45);

\foreach \g in {lowerone,lowertwo,lowerthree,lowerfour}
{
    \draw[gray!60,line width=0.7pt] (\g) -- (top);
}

\foreach \i in {inputleft,inputcentre,inputright}
{
    \foreach \g in {lowerone,lowertwo,lowerthree,lowerfour}
    {
        \draw[gray!50,line width=0.55pt] (\i) -- (\g);
    }
}

\node[font=\bfseries] at (7,8.1)
{A typical circuit class $\mathcal{C}$};

\node[font=\small] at (7,7.35)
{Example: $\mathcal{C}=\mathsf{TC}^0_2$};

\node[circuitgate] at (top) {$\mathsf{THR}$};
\node[circuitgate] at (lowerone) {$\mathsf{THR}$};
\node[circuitgate] at (lowertwo) {$\mathsf{THR}$};
\node[circuitdots] at (lowerdots) {$\cdots$};
\node[circuitgate] at (lowerthree) {$\mathsf{THR}$};
\node[circuitgate] at (lowerfour) {$\mathsf{THR}$};

\node[circuitinput] at (inputleft) {$x_1$};
\node[circuitinput] at (inputcentre) {$\cdots$};
\node[circuitinput] at (inputright) {$x_n$};

\node[learningpanel] at (3.35,1.15) {};
\node[complexitypanel] at (10.65,1.15) {};

\draw[green!45!black,line width=0.9pt,->]
    (6.65,3.15) ..
    controls (5.65,2.9) and (4.55,2.65) ..
    (3.35,2.25);

\draw[orange!75!black,line width=0.9pt,->]
    (7.35,3.15) ..
    controls (8.35,2.9) and (9.45,2.65) ..
    (10.65,2.25);

\node[
    font=\bfseries\large,
    text=green!40!black
]
at (3.35,1.75)
{Learning};

\node[
    font=\normalsize,
    align=center,
    text width=5.4cm
]
at (3.35,1.05)
{Learn an unknown $\mathcal{C}$-circuit};

\node[font=\small] at (3.35,0.45)
{$f\in\mathcal{C}(\poly)
\;\Longrightarrow\;
\text{hypothesis }h\approx f$};

\node[
    font=\bfseries\large,
    text=orange!75!black
]
at (10.65,1.75)
{Complexity};

\node[
    font=\normalsize,
    align=center,
    text width=5.4cm
]
at (10.65,1.05)
{Exhibit a function that does not have small
$\mathcal{C}$-circuits:};

\node[font=\small] at (10.65,0.35)
{$\text{an explicit }g\notin\mathcal{C}(\poly)$};

\end{tikzpicture}

\caption{\small{
Two closely related questions about a typical circuit class
$\mathcal{C}$, illustrated by the class $\mathsf{TC}^0_2$ of
depth-$2$ threshold circuits (see \Cref{sec:prelim_circuits}).
By \Cref{thm:rand_learning_lbs} and \Cref{thm:quant_learning_lbs},
a nontrivial classical or quantum learner for $\mathcal{C}$
yields a new circuit size lower bound against $\mathcal{C}$
for an explicit function.
Conversely, a combinatorial lower bound that naturalises yields
a constructive property that can be converted into a learning
algorithm for $\mathcal{C}$, as discussed in
\Cref{sec:learn-via-dist}.
A recurring pattern across circuit classes at the frontier is
that the two questions have the same status: either learners and
lower bounds are known, or both remain open.
In particular, for $\mathsf{TC}^0_2$, neither a nontrivial learner
nor a superpolynomial circuit size lower bound for an explicit
function is known.
}}
\label{fig:learningcomplexityfrontier}
\end{figure}

The original proof of \Cref{thm:rand_learning_lbs} \citep{DBLP:conf/coco/OliveiraS17} makes use of the learning speedups discussed in \Cref{sec:speedup}, which reduces the result to a proof that \emph{subexponential} time learning in time $2^{n^{o(1)}}$ implies complexity lower bounds. On the other hand, \citep{DBLP:conf/focs/ArunachalamGGOS21} explores a more direct perspective  that can be used to derive both \Cref{thm:rand_learning_lbs} and  \Cref{thm:quant_learning_lbs}. We refer to \citep{DBLP:conf/coco/OliveiraS17, DBLP:conf/focs/ArunachalamGGOS21} for more information about these theorems and their proofs. 

While the results discussed in this section apply to a significantly larger class of learning algorithms, they  come with a technical caveat: the hard function extracted from the learning assumption is hard only on infinitely many input lengths $n$. In contrast, it can be checked that the lower bound arguments  in \Cref{thm:lb_from_det_learning} establish hardness for every sufficiently large input length $n$. This caveat is a consequence of the more indirect nature of their proofs. It would be interesting to introduce techniques to close this gap.

\begin{problem}[Almost-Everywhere Hardness from Learning]
    Show that if a typical circuit class $\mathcal{C}$ admits nontrivial randomized learning algorithms, then there is a language in $\BPE$ that is not in $\mathcal{C}(\poly)$ on every large enough input length.
\end{problem}

We refer to \citep{DBLP:conf/colt/Volkovich16} for additional settings where one can derive complexity-theoretic consequences from learning algorithms.

\subsection{Unconditional Meta-Complexity Lower Bounds via Learning}\label{sec:uncond_lbs}

Nontrivial learning algorithms have been designed for a large number of circuit classes \citep{DBLP:conf/innovations/ServedioT17, DBLP:conf/alt/Karchmer24}. However, in all existing results, lower bounds are already known against the corresponding circuit class. A concrete class for which both lower bounds and learning remain elusive is the class of Boolean functions that can be approximated by low-degree Torus polynomials (see \citep{DBLP:conf/innovations/BhrushundiHLR19}).\footnote{
Degree-$d$ Torus polynomials and degree-$d$ polynomial threshold
functions (PTFs) can both be viewed as weighted majorities over the same
feature representation.
Let $\Phi_d(x)$ denote the vector of all multilinear monomials in
$x$ of degree at most $d$.
In both cases, one first computes a weighted sum
$p(x)=\langle w,\Phi_d(x)\rangle$.
A PTF applies a \emph{sign} 
activation $\mathbf{1}_{[p(x)\geq 0]}$, whereas a Torus polynomial uses
a \emph{periodic} activation: it reduces $p(x)$ modulo $1$ and interprets
values close to $0$ and $1/2$ as the Boolean values $0$ and $1$,
respectively.
Thus, the two models use the same degree-bounded feature expansion
but different activation functions.
As discussed in \Cref{sec:classical_connections}, degree-$d$
PTFs can be learned in time
$\poly\bigl(n^d,1/\varepsilon,\log(1/\delta)\bigr)$.
}

\begin{problem}[Nontrivial Learnability of Torus Polynomials]\label{prob:torus}
Is there a nontrivial randomized or quantum learning algorithm for the class of Boolean functions that can be $1/n$-approximated by $\poly(\log n)$-degree Torus polynomials?
\end{problem}

Since any function in  the circuit class $\mathsf{ACC}^0$ admits such approximating polynomials \citep{DBLP:conf/innovations/BhrushundiHLR19}, a positive answer to \Cref{prob:torus} leads to a new lower bound against $\mathsf{ACC}^0$ via \Cref{thm:quant_learning_lbs}.\footnote{Recall that an $\mathsf{ACC}^0$ circuit is a constant-depth polynomial-size Boolean circuit consisting of unbounded fan-in AND, OR, NOT, and MOD$_m$ gates, where $m$ is a fixed integer and MOD$_m$ computes the sum of the input bits modulo $m$.} The latter would be a breakthrough result in complexity theory. We refer to \citep{DBLP:journals/jacm/Williams14} for more information on $\mathsf{ACC}^0$.

A closely related problem is that of learning the class of Boolean circuits $f \colon \{0,1\}^n \to \{0,1\}$ that can be written as a disjoint union (DOR) of quasi-polynomially many MAJ-AND$_t$ circuits, where the bottom AND gate has fan-in $t = \poly(\log n)$.\footnote{Since MAJ-AND$_t$ circuits of bottom fan-in $t = \poly(\log n)$ can be viewed as degree-$t$ PTFs, MAJ-AND$_t$ circuits (without a top DOR gate) can be learned in quasi-polynomial time.}

\begin{problem}[Nontrivial Learnability of DOR-MAJ-AND Circuits]\label{prob:DORMAJAND}
    Is there a nontrivial randomized or quantum learning algorithm for the class of quasi-polynomial size DOR-MAJ-AND circuits of polylogarithmic bottom fan-in?
\end{problem}

It is not hard to prove that a positive solution to this problem also implies lower bounds against $\mathsf{ACC}^0$, by combining \Cref{thm:quant_learning_lbs} with the simulation $\mathsf{ACC}^0 \subseteq  \mathsf{SYM}^+ \subseteq \mathsf{DOR}_k \text{-} \mathsf{MAJ} \text{-} \mathsf{AND}_t$ for $k = \exp(\poly(\log n))$ and $t = \poly(\log n)$ \citep{beigel1994acc}.\\

\noindent \textbf{Complexity Lower Bounds from Learning.} While \Cref{thm:rand_learning_lbs} and \Cref{thm:quant_learning_lbs} have not yet led to new complexity lower bounds in a \emph{direct} way, it is possible to use these results in an \emph{indirect} way to establish \emph{unconditional} complexity lower bounds. We describe next a general template:
$$
\mathsf{Easiness~Assumption} \quad \xRightarrow[]{(i)}  \quad  \mathsf{Learning~Algorithm}  \quad  \xRightarrow[]{(ii)}  \quad  \mathsf{Computational~Hardness} 
$$
In more detail, we proceed as follows:
\begin{enumerate}
    \item[(\emph{i})] First, from a suitable $\mathsf{Easiness~Assumption}$, we show that nontrivial learning \emph{is} possible.
    \item[(\emph{ii})] Next, we invoke a result such as \Cref{thm:rand_learning_lbs} to derive some form of $\mathsf{Computational~Hardness}$.
    \item[] To obtain an unconditional lower bound, we carefully balance the initial $\mathsf{Easiness~Assumption}$ with the resulting $\mathsf{Computational~Hardness}$ consequence, with the goal of achieving a contradiction between the two.
\end{enumerate}
If we are able to implement this plan, it follows that the original $\mathsf{Easiness~Assumption}$ is \emph{false}, i.e., a corresponding complexity lower bound follows.\\

\noindent \textbf{Hardness of Estimating Complexity.} This approach has borne fruit in the form of an unconditional complexity lower bound in Meta-Complexity. In \citep{DBLP:conf/icalp/Oliveira19}, we considered the Minimum $\mathsf{rKt}$  Problem ($\mathsf{MrKtP}$): given a string $w \in \{0,1\}^n$, estimate its $\mathsf{rKt}$ complexity, a measure of the compressibility, or regularity, of $w$.\footnote{Random strings have $\mathsf{rKt}$ complexity close to $n$, while structured strings can be succinctly described and have much smaller $\mathsf{rKt}$ complexity.} This is a randomised analogue of a fundamental measure introduced by L.~Levin in 1984 \citep{DBLP:journals/iandc/Levin84} called $\mathsf{Kt}$ complexity. While the complexity of estimating the $\mathsf{Kt}$ complexity of an input string ($\mathsf{MKtP}$) remains open, we were able to show that $\mathsf{MrKtP} \notin \mathsf{BPP}$, i.e., the $\mathsf{rKt}$ complexity of a string cannot be computed in probabilistic polynomial time. 

This can be shown by analysing the following $\mathsf{Easiness~Assumption}$: $\mathsf{MrKtP} \in \mathsf{BPP}$. From this and the technique of \citep{DBLP:conf/coco/CarmosinoIKK16} described in \Cref{sec:learn-via-dist}, it follows that polynomial-size Boolean circuits can be learned in randomised polynomial time. This is Item (\emph{i}) in the template presented above. 

On the other hand, the existence of such a learning algorithm implies by \Cref{thm:rand_learning_lbs} that some problem in $\mathsf{BPE}$ (randomised exponential time) cannot be computed by polynomial size Boolean circuits. This corresponds to Item (\emph{ii}) in the aforementioned template. 

Finally, since every problem in $\mathsf{BPE}$ can be reduced using polynomial-size Boolean circuits to $\mathsf{MrKtP}$, it must be the case that $\mathsf{MrKtP} \notin \mathsf{BPP}$. This contradicts the original $\mathsf{Easiness~Assumption}$. 

In other words, it follows unconditionally that $\mathsf{MrKtP} \notin \mathsf{BPP}$. Note that our original problem, $\mathsf{MrKtP}$, makes no reference to learning algorithms!\\

Recall that it is a longstanding open problem to show the analogous result for $\MKtP$, i.e., to prove that $\MKtP \notin \P$. In light of the above discussion and \Cref{thm:lb_from_det_learning}, it is not hard to see that to obtain this result it would be sufficient to show that if  $\MKtP \in \P$ then nontrivial \emph{deterministic} learning of polynomial-size Boolean circuits is possible. In turn, this would be feasible if we had a \emph{deterministic} analogue of the techniques from \Cref{sec:learn-via-dist}.

\begin{problem}[Deterministic Learning from Meta-Complexity]
    Is there a way to obtain deterministic learning algorithms assuming the easiness of meta-computational problems such as $\MKtP$?
\end{problem}

This is closely related to the problem of derandomizing learning algorithms. For more about this line of work, we refer the reader to \citep{DBLP:conf/approx/OliveiraS18}.

To conclude this section, we note that we have only discussed lower bounds that can be extracted from PAC learning algorithms. In principle, it should be possible to obtain stronger complexity-theoretic consequences from the existence of algorithms for more expressive learning tasks (see \Cref{sec:complexity_learning}). We leave this as an interesting direction for further research.

\begin{problem}[Complexity-Theoretic Consequences of Learning]
    Derive stronger complexity-theoretic consequences from the existence of algorithms for learning tasks that are more expressive than PAC learning under the uniform distribution, such as distribution learning and sequence extrapolation.
\end{problem}

\newpage 

\section{Related Work}\label{sec:related_work}

Here we highlight some related lines of work that are not covered in the preceding sections, following roughly the order of the survey. We also provide pointers to some references that cover more practical aspects of learning and its connections to compression. 

\paragraph{Search vs Decision in Meta-Complexity.}
The mismatch between recognising the existence of a short description and producing one is highlighted in \Cref{sec:introduction}. The learning constructions surveyed in \Cref{sec:learning_via_meta} often circumvent this mismatch by  converting a decision procedure into a predictor through a reconstruction argument. Search-to-decision reductions for time-bounded Kolmogorov complexity address the issue directly, asking when a decision procedure suffices to recover a short program. We refer to \citep{MP24, DBLP:conf/coco/HiraharaKLO24, DBLP:conf/focs/HiraharaLN24} for recent work on this direction and related open problems. 

\paragraph{Compression in Supervised Learning.} The term \emph{model compression} (\Cref{sec:classical_connections}) has a different meaning in modern machine learning, where one starts with a trained model and seeks a smaller  approximation. Representative approaches include transferring the predictive behaviour of a large model or ensemble to a smaller network through knowledge distillation \citep{HintonVinyalsDean2015}, compressing an explicitly stored network through pruning, quantization, and entropy coding \citep{HanMaoDally2016}, and training stochastic network parameters subject to an explicit information-theoretic coding budget \citep{HavasiPeharzHernandezLobato2019}. Dataset distillation moves the compression to the data side, replacing a large training set by a small collection of synthetic examples from which a prescribed training procedure recovers comparable performance \citep{WangZhuTorralbaEfros2018}. 

\Cref{sec:classical_connections} also develops the basic connection between sample compression and generalisation, mentions the qualitative equivalence with finite VC dimension, and states the sample compression conjecture. The classical binary sample compression framework was introduced by Littlestone and Warmuth \citep{littlestone1986relating} and developed further by Floyd and Warmuth \citep{FloydWarmuth1995}. Subsequent work extends or refines this perspective in several directions, including multiclass and general loss learning \citep{DavidMoranYehudayoff2016}, stable compression schemes \citep{HannekeKontorovich2021}, list compression and its limitations \citep{HannekeMoranWaknine2024}, and reductions from compression problems in other learning models to classical binary compression \citep{AttiasHannekeRamaswami2025}.

\paragraph{Distributional Restrictions and the Complexity of Learning.} Departing from the pessimistic distribution-free definition of PAC learning, the learning paradises of \Cref{sec:paradises} deliberately restrict the example distributions to feasible ones. This choice can change the computational complexity of the learning task substantially. Blanc et al.\ \citep{DBLP:conf/innovations/Blanc0LST26} show that requiring success only on samplable distributions can enlarge the power of efficient learning, even at the level of sample complexity.  Blanc, Lange, Strassle, and Tan \citep{BlancLangeStrassleTan2025} instead ask when a learner for a restricted family of example distributions can be lifted to arbitrary example distributions. The latter result is governed by the complexity of expressing the new distribution as a mixture of distributions from the original family. This is distinct from the distributional inversion argument in \Cref{sec:learning_via_meta}, which moves between example distributions by inverting their samplers.

\paragraph{Practical Analogues of the Learning Paradises.}
The assumptions P.I through P.V from \Cref{sec:paradises} are much stronger than anything currently achieved in practice, but the tasks they formalize have recognizable counterparts in modern machine learning. These analogies concern the form of each learning task, not the strength of its guarantee. Practical methods exploit regularities in natural data and address restricted families of instances, whereas the paradises require uniform guarantees over broad families of sources, example distributions, or circuits.

P.I is closest to predicting the continuation of a sequence from an observed context, as in models for text, video, or dynamical systems. For instance, Del\'etang et al.~\citep{DBLP:conf/iclr/DeletangRDCGMGW24} study the practical connection between prediction and compression in this setting, while Grau-Moya et al.~\citep{pmlr-v235-grau-moya24a} investigate neural predictors trained to approximate universal prediction.

P.II resembles generative modeling, where independent examples are used to learn a procedure for producing new samples. The denoising diffusion model of Ho, Jain, and Abbeel~\citep{HoJainAbbeel2020} is a prominent example. Generative models can also serve as compressors. Yang, Mandt, and Theis~\citep{YangMandtTheis2023} survey this area, while Kingma et al.~\citep{KingmaSalimansPooleHo2021} turn a variational diffusion model into a lossless compressor whose rate is close to its variational coding bound. From the perspective of this survey, these systems also expose resources not captured by code length, notably encoding and decoding time and the size of the shared learned model. We note that practical systems are generally evaluated by prediction loss, sample quality, or application performance, rather than by guarantees that hold uniformly over all efficiently samplable distributions, as P.I and P.II require.

P.III is closest to learning when labels may be noisy and the target need not belong to the chosen hypothesis class. Haussler~\citep{DBLP:journals/iandc/Haussler92} extends the PAC model to general loss minimization, while Kearns, Schapire, and Sellie~\citep{DBLP:journals/ml/KearnsSS94} formulate the agnostic model, in which performance is compared with that of the best predictor in a class. Tractable results for noisy labels generally assume additional structure. For example, Natarajan et al.~\citep{NatarajanDhillonRavikumarTewari2013} study loss correction when the probability of a label error depends only on the true label. This is one example from a vast theoretical and empirical literature, which we do not attempt to survey here.

P.IV considers learning through chosen queries. Angluin~\citep{Angluin1988Queries} provides a foundational framework for query learning. Practical counterparts arise when a learner can choose which inputs are labelled, as in active learning \citep{CohnGhahramaniJordan1996}, or can experiment with an implemented system, as in black box checking \citep{PeledVardiYannakakis1999} and model extraction \citep{TramerZhangJuelsReiterRistenpart2016}. Membership queries provide stronger access than the common pool-based model of active learning, since they allow the learner to query any input rather than only examples from a fixed unlabelled sample.

P.V has a broader practical analogue. A heuristic learning algorithm may succeed on concepts that arise naturally in an application even if it fails on adversarial or otherwise atypical concepts from the same representation class. This is precisely the relaxation captured by P.V.

The empirical success of modern machine learning remains  compatible with the hardness results discussed above. As explained before, the learning paradises  concern uniform guarantees over broad families of efficient sources, distributions, or circuits, which can be far more general than success on the structured instances encountered in practice. These  stronger theoretical guarantees enable the connections to cryptography and average-case complexity, which in turn yield hardness of learning results.

\paragraph{Cryptographic Hardness of Learning.}
The cryptographic literature  approaches learning from the opposite direction to much of this survey, using cryptographic assumptions to prove that particular concept classes are hard to learn. Classical examples include the reductions of Pitt and Valiant \citep{PittValiant1988} and of Kearns and Valiant \citep{KearnsValiant1994}, followed by other representative examples such as Klivans and Sherstov \citep{DBLP:journals/jcss/KlivansS09}  and Daniely et al.\ \citep{10.1145/2591796.2591820} (see also \citep{DBLP:conf/focs/AllenOW15} and subsequent work \citep{DBLP:conf/stoc/Daniely16, pmlr-v134-daniely21a}). There is an extensive literature on the  hardness of learning tasks under a variety of cryptographic conjectures, and reviewing these results is beyond the scope of this section. 

A related line of work investigates whether weaker complexity-theoretic assumptions are sufficient to rule out efficient learning. Applebaum, Barak, and Xiao identify barriers to basing such results on standard worst-case assumptions \citep{DBLP:conf/focs/ApplebaumBX08}.  Goldberg and Kabanets \citep{GoldbergKabanets2025Witness} show that certain formulations of NP-hardness of learning characterize witness encryption, while Blanc et al.\ \citep{BlancKochStrassleTan2025} derive computational and statistical tradeoffs from NP hardness. 

The discussion of white-box learning from \Cref{sec:wb} concentrates on one-way functions. Liu, Mazor, and Pass \citep{DBLP:conf/innovations/LiuMP25} obtain a corresponding characterization of public-key encryption through a richer white-box learning problem. The existence of secure key-agreement protocols has also been characterized via notions from Kolmogorov complexity and meta-complexity \citep{DBLP:conf/focs/BallLMP23}.

\paragraph{Description Length in Statistical Inference and Prediction.}
The treatment in \Cref{sec:shannon_connections} uses coding as a bridge between probabilistic prediction and compression, with particular attention to computational efficiency. Minimum Message Length (MML), originating in the work of Wallace and Boulton \citep{WallaceBoulton1968}, and Minimum Description Length (MDL), introduced by Rissanen \citep{Rissanen1978}, instead use description length as a criterion for statistical inference and model selection. Barron, Rissanen, and Yu \citep{BarronRissanenYu1998} give a standard account of the connections among MDL, universal coding, and statistical inference, while Vit\'anyi and Li \citep{VitanyiLi2000} relate MDL and MML to Bayesian inference and Kolmogorov complexity.

Dawid \citep{Dawid1984} formulated the prequential (``predictive and sequential'') approach, which evaluates a statistical method through the predictions it issues sequentially; under log loss, these predictions themselves define a code. Rissanen \citep{Rissanen1984} connects universal coding with prediction and estimation, while Feder, Merhav, and Gutman \citep{FederMerhavGutman1992} study universal prediction of individual sequences. Cesa-Bianchi and Lugosi \citep{CesaBianchiLugosi2006} give a standard treatment of sequential prediction and regret. These traditions are close to the coding interpretation developed above, but they generally concern statistical performance rather than the efficient computational transformations that are central here.

Building on Kolmogorov's earlier ideas about separating structure from randomness, G\'acs, Tromp, and Vit\'anyi \citep{GacsTrompVitanyi2001} develop algorithmic statistics, asking how much of an individual string's shortest description should be regarded as an explanatory model and how much as unexplained randomness. This is related to the use of algorithmic information theory for extrapolation in \Cref{sec:solomonoff}, but its primary object is the structure of a given string rather than prediction from a sequence.

\paragraph{Similarity via Compression.}
A different use of compression in machine learning, not covered in the survey, treats the lengths produced by a compressor as a measure of similarity between individual objects. Li et al.\ \citep{DBLP:journals/tit/LiCLMV04} introduce normalized information distance and demonstrate practical approximations using standard compressors, while Cilibrasi and Vit\'anyi develop normalized compression distance into a general clustering method \citep{DBLP:journals/tit/CilibrasiV05}. More recently, Jiang et al.\ \citep{DBLP:conf/acl/JiangYTTDL23} illustrate the same idea in text classification by combining gzip with a nearest neighbor rule. This perspective is complementary to the connections between learning and compression discussed in the survey: compression is used directly as a data analysis primitive, without learning an explicit  model.

\newpage 

\small

\bibliographystyle{alpha}	
\bibliography{references}

\end{document}